\documentclass[12pt]{article}

\usepackage{amsmath, amsthm, amssymb, amscd}

\usepackage{graphicx}
\usepackage{bm}
\usepackage{tikz}
\usepackage{dsfont}
\usepackage{bbm}
\usepackage{epstopdf}
\usepackage{tikzsymbols}
\usepackage{hyperref}
\usepackage{caption}
\usepackage{natbib}
\usepackage{verbatim}

\newcommand{\bv}{\mathbf{v}}
\newcommand{\sd}{\mathbb{S}^{d-1}}

\newcommand{\R}{\mathbb{R}}

\newcommand{\td}{\widetilde }
\newcommand{\ol}{\overline }
\newcommand{\ca}{\mathcal }

\newcommand{\tn}{\, \textnormal}

\newcommand{\beq}{\begin{equation}\label }
\newcommand{\eeq}{\end{equation} }
\newcommand{\bal}{\begin{align*} }
\newcommand{\ve}{\varepsilon}

\newcommand{\bs}{\boldsymbol}
\newcommand{\eal}{\end{align*} }

\numberwithin{equation}{section}
\newtheorem{thm}{Theorem}[section]
\newtheorem{lem}[thm]{Lemma}

\newtheorem{assump}{Assumption}

\newtheorem{thm2}{Theorem}[section]
\newtheorem{lem2}{Lemma}[section]

\newtheorem{prop2}{Proposition}[section]

\newcommand{\Var}{\operatorname{Var}}

\newcommand{\Vol}{\operatorname{Vol}}
\newcommand{\E}{\mathbb{E}}
\renewcommand{\P}{\mathbb{P}}

\newcommand{\Unif}{\operatorname{Unif}}

\newcommand{\bbeta}{{\boldsymbol{\beta}}}
\newcommand{\btheta}{{\boldsymbol{\theta}}}

\newcommand{\bTheta}{{\boldsymbol{\Theta}}}
\newcommand{\bgamma}{{\boldsymbol{\gamma}}}
\newcommand{\bt}{\mathbf{t}}
\newcommand{\bX}{\mathbf{X}}

\newcommand{\bb}{\mathbf{b}}
\newcommand{\bu}{\mathbf{u}}

\begin{document}

\title{\vspace{-40pt}Tests for qualitative features in the random coefficients model}

\author{\begin{tabular}{ccc}
Fabian Dunker\footnote{School of Mathematics and Statistics, University of Canterbury, Private Bag 4800, Christchurch 8140, New Zealand
}& & Konstantin Eckle\footnote{Mathematical Institute of the University of Leiden, Niels Bohrweg 1, 2333 CA Leiden, Netherlands
} \footnote{Corresponding autoher, email k.j.eckle@math.leidenuniv.nl}\\
\small{University of Canterbury} & & \small{University of Leiden}\\[7pt]
 Katharina Proksch\footnote{Institute for Mathematical Stochastics, Georg-August-University of Goettingen, Goldschmidtstrasse 7, 37077 Goettingen, Germany
 } & & Johannes Schmidt-Hieber\footnote{Mathematical Institute of the University of Leiden, Niels Bohrweg 1, 2333 CA Leiden, Netherlands
 }\\
 \small{University of Goettingen} & & \small{University of Leiden}\\[-30pt]
 \end{tabular}
}


\date{}
\maketitle

\begin{abstract}
The random coefficients model is an extension of the linear regression model that allows for unobserved heterogeneity in the population by modeling the regression coefficients as random variables. Given data from this model, the statistical challenge is to recover information about the joint density of the random coefficients which is a multivariate and ill-posed problem.  Because of the curse of dimensionality and the ill-posedness, pointwise nonparametric estimation of the joint density is difficult and suffers from slow convergence rates. Larger features, such as an increase of the density along some direction or a well-accentuated mode can, however, be much easier detected from data by means of statistical tests. In this article, we follow this strategy and construct tests and confidence statements for qualitative features of the joint density, such as increases, decreases and modes. We propose a multiple testing approach based on aggregating single tests which are designed to extract shape information on fixed scales and directions. Using recent tools for Gaussian approximations of multivariate empirical processes, we derive expressions for the critical value. We apply our method to simulated and real data.
\end{abstract}

%
%
\paragraph{Keywords:}  Gaussian approximation; mode detection; monotonicity; multiscale statistics; shape constraints; Radon transform; ill-posed problems. 

\section{Introduction}
In the random coefficients model, $n$ i.i.d. random vectors $(\mathbf{X}_i,Y_i),$ $i=1,\ldots,n$ are observed, with $\mathbf{X}_i=(X_{i,1}, \ldots, X_{i,d})$ a $d$-dimensional vector of design variables and
\begin{align}
	Y_i = \beta_{i,1}X_{i,1} + \beta_{i,2} X_{i,2} + \ldots + \beta_{i,d} X_{i,d}, \quad i=1,\ldots,n.
	\label{eq:model}
\end{align} 
The unobserved random coefficients $\bbeta_i=(\beta_{i,1}, \ldots, \beta_{i,d}),$ $i=1,\ldots,n,$ are i.i.d. 
realizations of an unknown $d$-dimensional distribution $F_\bbeta$ with Lebesgue density $f_\bbeta.$ Design variables and random coefficients are assumed to be independent. The statistical task is to recover properties of the joint density $f_\bbeta,$ which is assumed to belong to some nonparametric class. In this work, we derive tests for increases and modes of $f_\bbeta.$

For $d=1,$ the random coefficients model simplifies to nonparametric density estimation. For $d>1,$ recovery of $f_\bbeta$ is an inverse problem with ill-posedness depending on the distribution of the design vectors $\mathbf{X}_i.$ If the design is sufficiently regular, the inverse problem is mildly ill-posed. Otherwise, the model can be severely ill-posed or even be non-identifiable. In this work, we study the mildly ill-posed regime and consider in particular the random coefficients model with random intercept
\begin{align}
	Y_i = \beta_{i,1} + \beta_{i,2} X_{i,2} + \ldots + \beta_{i,d} X_{i,d}, \quad i=1,\ldots,n,
	\label{eq:model_with_int}
\end{align} 
which can be obtained from \eqref{eq:model} setting $X_{i,1} = 1,$ almost surely.

Random coefficients models appear in econometrics and epidemiology and are used to model unobserved heterogeneity in the population. While the standard linear regression model accounts for unobserved heterogeneity only by an intercept that varies across the population, the random coefficients model allows in addition that different individuals have different slopes. Applications in epidemiology are considered by \cite{Greenland00, Greenland06}. In economics, random coefficients models are frequently used to evaluate panel data, cf. \cite{Hsiao14} or \cite{Hasio04}, Chapter 6, for an overview. Modeling and estimating consumer demand in industrial organization and marketing often makes use of random coefficients \cite{BLP:95, Petrin:02, Nevo:01, Berry07, Dube12}. In all these works, parametric assumptions on $f_{\bbeta}$ are imposed. Recently, nonparametric approaches for random coefficients became popular in microeconometrics \cite{hoderlein2008, Masten15, HHM:15, DHK:17}, frequently combined with binary choice \cite{Ichimura98, Gautier12, Gautier13, Masten14, DHK:13, Fox16, DHKS:18}, among others.

The random coefficients model also includes quantum homodyne tomography. In this case, we observe an angle $\Phi_i$ and
\begin{align}\label{qht}
	Y_i = Q_i  \cos(\Phi_i) + P_i \sin(\Phi_i) , \quad i=1,\ldots,n,
\end{align}
with $(Q_i,P_i)$ i.i.d. random variables which are unobserved and independent of $\Phi_i.$ The angles $\Phi_i$ can be chosen 
by the experimenter and are typically uniform on $[0,\pi].$ The interest is in reconstruction of the Wigner function which 
takes the role of the joint density of $(Q_i,P_i).$ Because $P_i$ and $Q_i$ are not jointly observable, the Wigner 
function can take negative values. For more on quantum homodyne tomography and the Wigner function, see \cite{butucea2007}.

We propose a nonparametric test for shape information of the joint density $f_\bbeta$ in the random coefficients model. The focus will be on a test for directional derivatives and modes. The nonparametric estimation theory for $f_{\bbeta}$ has been developed in \cite{Beran1992,Beran1996,Feuerverger2000,hoderlein2008}. Due to the ill-posedness of the problem and the curse of dimensionality induced by $d,$ pointwise estimation rates are slow. The reason is that small perturbations in the signal are indistinguishable given the data. Nevertheless, we can get good detection rates for larger features, such as an accentuated mode or a strong increase in the joint density along some direction. From a practical point of view, the  relevant information regarding an unknown density is typically its shape rater than its precise, full reconstruction. It is therefore essential to recover increases/decreases and the modes of a density.  If, say, two modes in the joint density of two random quantities are detected, this indicates that two different groups can be identified. Hence, shape information allows to interpret a given dataset.

Larger features of the density will also be discovered by a nonparametric estimator even if it suffers from slow pointwise 
convergence. There are, however, two important reasons why a testing approach might be more appropriate. Firstly, with a significance test of level $\alpha$
 we can conclude that with probability $1-\alpha$ a detected feature is not an artifact. Secondly, for an estimator we need to pick one bandwidth or smoothing parameter while detection of different features might require different bandwidth choices depending on the size of the hidden features themselves. Indeed, a short and steep increase will be best detected on a small scale whereas for finding a longer and less strong increase the choice of a larger bandwidth is beneficial. Using multiple testing methods, it is possible to combine a whole range of smoothness parameters into one test and to adapt to different shapes of features.

We construct a so called multiscale test, aggregating single tests on different scales and directions. Multiscale tests can be viewed as a multiple testing procedure specifically designed for nonparametric models. Given a model, the theoretical challenge is to prove that a multiscale statistic can be approximated by a distribution free statistic which is independent of the observations. This allows us then to compute quantiles and to find approximations for the critical values of the multiscale statistic. So far, qualitative feature detection based on multiscale statistics has been studied for various nonparametric models, including the Gaussian white noise model \cite{duembgen2001}, density estimation \cite{duembgen2008} and  deconvolution \cite{schmidthieber13}.
In multivariate settings the classical KMT approximation suffers from the curse of dimensionality which then leads to very restrictive conditions on the usable scales. Instead, very recent results on Gaussian approximations of suprema of multivariate empirical processes developed by  \cite{Chernozhukov2017} can be used \cite[in the context of multivariate deconvolution and multivariate linear inverse problems with additive noise, respectively]{Eckle,proksch16}. In this work, we extend these techniques. The main difficulties are twofold.  First, we need to derive specific properties of the inverse Radon transform for general dimension $d$. Second, in contrast to the other works  on multiscale inference, no distribution free approximation can be obtained  and we therefore need to study the approximating process if several unobserved functions are replaced by estimators.

In order to study the power of the multiscale test, a theoretical detection bound and numerical simulations are provided. The theoretical result gives conditions under which a mode can be detected. In a numerical simulation study, we investigate the power of the test for increases/decreases along some direction and mode detection in dependence on the sample size and the design variables. We also analyze real consumer demand data from the British Family Expenditure Survey. 

Let us briefly summarize related literature on testing in the random coefficients model. Under a parametric assumption on the density $f_\bbeta,$ \cite{Beran1993} considers goodness-of-fit testing and \cite{swamy1970, andrews2001} test whether some of the random coefficients are deterministic. The only test based on a nonparametric assumption was proposed recently by \cite{Breunig16}. It allows to assess  whether a given set of data  follows the random coefficients model.

This paper is organized as follows. In Section \ref{sec:inv_prob}, we describe the connection between the random coefficients model and the Radon transform. Rewriting the model as an inverse problem in terms of the Radon transform reveals the ill-posed nature of the model. This allows us to construct and to analyze the multiscale test in Section \ref{sec:multiscale_test}. In this part we also derive the asymptotic theory of the estimator and obtain theoretical detection bounds. In Section \ref{Sec4} the test is studied for simulated data. As a real data example, consumer demand is analyzed in Section \ref{sec:real_data}. Proofs and technicalities are deferred to a supplement.  An R package and Python code is available as a supplement as well, \url{https://arxiv.org/abs/1704.01066}.

{\it Notation:} Throughout the paper, vectors are displayed by bold letters, e.g. $\mathbf{X}, \bbeta.$ Inequalities between vectors are understood componentwise. The Euclidean norm on $\mathbb{R}^d$ is 
denoted by $\|\cdot\|$ and the corresponding standard inner product by $\langle \cdot, \cdot \rangle.$ We further denote by 
$\mathbf{e}_1,\ldots,\mathbf{e}_d\in\mathbb{R}^d$ the standard ON-basis of the $d$-dimensional Euclidean space, $\mathbb{S}^{d-1}$ denotes the unit 
sphere in $\mathbb{R}^d$ and we write $\mathcal{Z}$ for the cylinder $\mathcal{Z}=\mathbb{R}\times \mathbb{S}^{d-1}$. Furthermore, we write 
$\bv$ for any direction $\bv=\sum_{j=1}^d\mathrm{v}_j\mathbf{e}_j\in \mathbb{S}^{d-1}$. For two positive sequences $(a_n)_n,$ $(b_n)_n,$ $a_n \lesssim b_n$ or $b_n\gtrsim a_n$ mean that for some positive constant $C,$ $a_n \leq Cb_n$ for all $n.$ As usual, we write $a_n \asymp b_n$ if $a_n \lesssim b_n$ and $b_n \lesssim a_n.$

\section{The random coefficients model as an inverse problem}
\label{sec:inv_prob}

\begin{minipage}{0.5\textwidth}
The random coefficients model can be written in terms of the Radon transform \cite[cf.][]{Beran1996}. This allows us then to interpret
the model as an inverse problem. In this section, we summarize the main steps and review relevant results on the inversion of 
the Radon transform. Let $H^s$ denote the $L^2$-Sobolev space. The Radon transform is the operator $R: H^{s}(\mathbb{R}^d)\rightarrow H^{s+\frac{d-1}{2}}(\mathcal{Z}),$ with
\end{minipage}
\hspace{1cm}
\begin{minipage}{0.35\textwidth}
\includegraphics[scale=0.45]{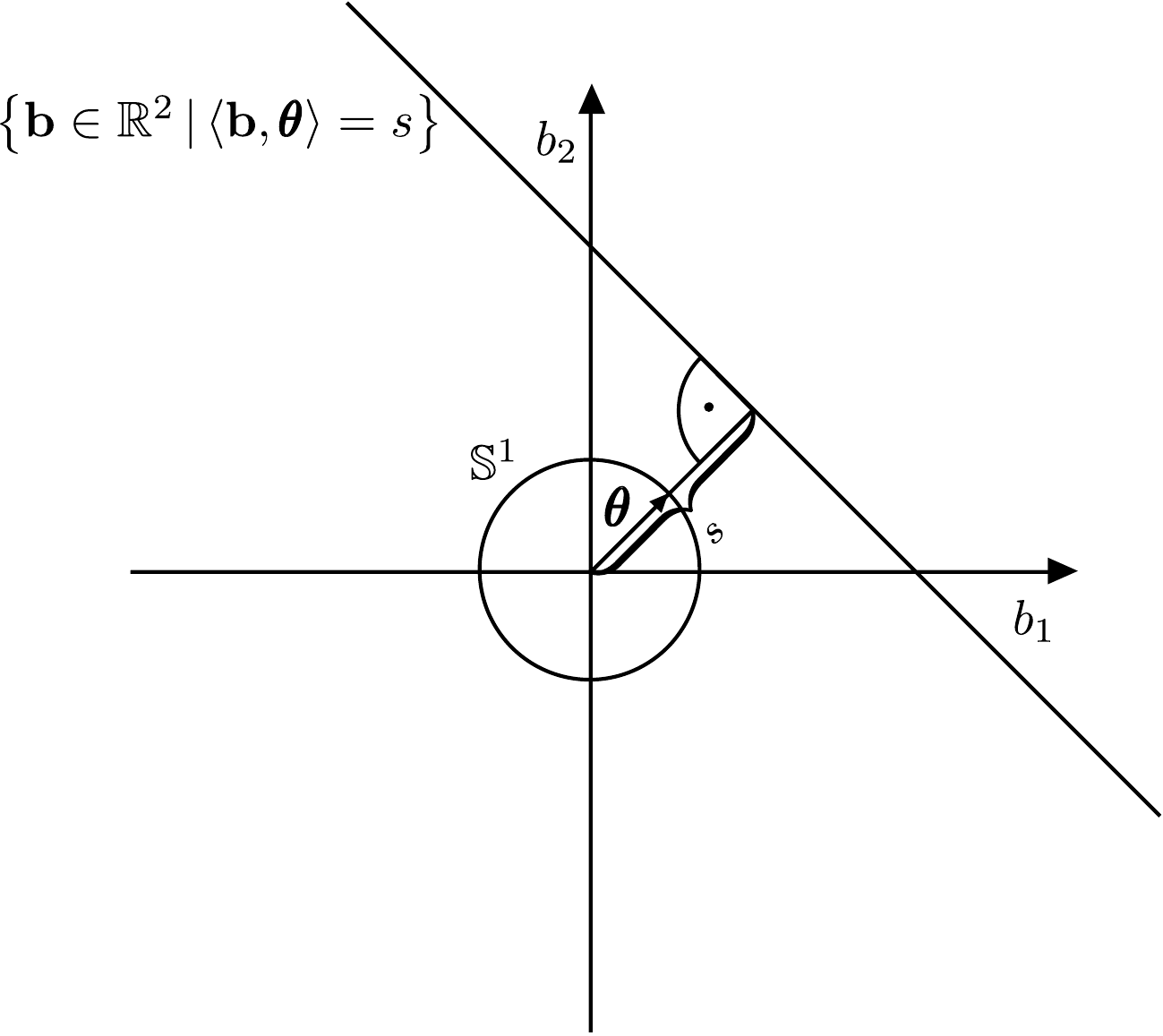}
\end{minipage}


  \begin{align*}
    Rf(s,\btheta)=\int_{\langle\bb,\btheta\rangle=s}f(\bb)\,d\mu_{d-1}(\bb)
  \end{align*}  
and $\mu_{d-1}$ the surface measure on the $(d-1)$-dimensional hyperplane $\{\bb\in\mathbb{R}^d:\langle\bb,\btheta\rangle=s\}$. The Radon transform maps therefore a function to all its integrals over hyperplanes parametrized by $(s,\btheta)\in\mathcal{Z}.$ The figure above shows the parametrization in two dimensions.

For the connection between the Radon transform and the random coefficients model \eqref{eq:model} consider the normalized observations
\begin{align*}
	S_i :=  \frac{Y_i}{\|\mathbf{X}_i\|}, \quad \bTheta_i := \frac{\mathbf{X}_i}{\|\mathbf{X}_i\|}, \quad i=1, \ldots, n.
\end{align*}
The random vectors $\bTheta_i$ take values in the $(d-1)$-dimensional sphere $\mathbb{S}^{d-1}$. In the random coefficients 
model with intercept \eqref{eq:model_with_int}, $\bTheta_i$ is always in the upper hemisphere, i.e. the first component of 
$\bTheta_i$ is positive. In this case, we extend the distribution of $\bTheta_i$ to the whole sphere by randomizing the signs of the
design variables. For this purpose, we generate independent random variables $\zeta_i$, $i=1,\hdots,n$, with 
$\P(\zeta_i=1) =\P(\zeta_i=-1) =1/2,$ which are independent of the data $(\mathbf{X}_i, Y_i),$ $i=1,\ldots,n$, and define 
$S_i := \zeta_i Y_i/\|\mathbf{X}_i\|$ and $\bTheta_i := \zeta_i \mathbf{X}_i/\|\mathbf{X}_i\|.$ Independent of the 
symmetrization, we have $S_i = \langle \bTheta_i , \bbeta_i \rangle$. The conditional distribution of 
$S_1 | \bTheta_1$ is therefore  
\[F_{S|\bTheta}(x|\btheta)=\mathbb{P}(S_1\leq x|\bTheta_1=\btheta)=\mathbb{P}(\langle\bTheta_1,\bbeta_1\rangle\leq x|
\bTheta_1=\btheta)=\int_{-\infty}^{x}(Rf_{\bbeta})(s, \btheta) \,ds, \]
and the conditional density becomes
\begin{align}\label{1.3}
f_{S|\bTheta}(s | \btheta) = (Rf_{\bbeta})(s, \btheta).
\end{align}
Recall that we have access to an i.i.d. sample $(S_i, \bTheta_i)$ of the joint density $f_{S, \bTheta}$. This allows for nonparametric estimation of $f_{S|\bTheta}$ via $f_{S|\bTheta}(s | \btheta) = f_{S, \bTheta}(s , \btheta)/f_\bTheta(\btheta)$. Applying the inverse Radon transform to this estimate gives an estimator for the joint density $f_\bbeta.$ This inversion scheme suffers from two sources of ill-posedness. Firstly, dividing by $f_\bTheta$ might result in very unstable reconstructions if $f_\bTheta$ is small. This happens if the normalized design variables $\bTheta_i =\bX_i/\|\bX_i\|$ systematically miss observations from some directions. In this case the 
problem becomes unevenly harder and only logarithmic convergence rates can be obtained  \cite[see][]{Davison1983,Frikel2013,Hohmann2016}. When the support of $\bTheta_i$ does not contain an open ball, $f_\bbeta$ might be non-identifiable. Secondly, even with regularity on the distribution of the design, the Radon inversion is known to be an ill-posed problem with degree of ill-posedness $(d-1)/2$. Hence, regularization of the inversion scheme is necessary. 

In this work, we study the mildly ill-posed case where the random directions $\bTheta_i,$ $i=1,\ldots,n,$ are sufficiently regularly distributed over the sphere and the ill-posedness is only due to the inversion of the Radon transform. The precise assumptions on the design are stated in Section \ref{sec.assump_on_design}.

Our approach makes use of the following explicit inversion formula of the Radon transform. Define the operator $\Lambda$ via
\begin{align}\label{.3}
	\Lambda f(s, \btheta) =
	\mathcal{H}_d\partial_s^{d-1} Rf (s, \btheta),
\end{align}
where $\mathcal{H}_d$ denotes the identity for $d$ odd and the Hilbert transform 
\[
	\ca H_d f(u)= \frac{1}{\pi}\lim_{\epsilon\rightarrow 0^+} 
	\int_{(-\infty, u-\epsilon] \cup [u+\epsilon, \infty) }  \frac{f(s)}{u-s}  ds =\frac{1}{\pi} \text{p.v.}  \int_{-\infty}^{\infty}  \frac{f(s)}{u-s}  ds\quad\big(u\in\mathbb R\big)
\]
 for $d$ even. Let $c_d^{-1}=(-1)^{({d-1})/{2}} 2^{-d}\pi^{{1-d}}$ for $d$ odd and ${c}_d^{-1}= -(-1)^{{d}/{2}}2^{-d}\pi^{{1-d}}$ for $d$ even. If $\varphi$ is a Schwartz function on $\R^d$, then we have the inversion formula
\begin{align}
	\varphi=c_d^{-1}R^*\Lambda\varphi,
	\label{eq.inversion_formula}
\end{align}
cf. Theorem 3.8 in \cite{Helgason2011}. The so called back projection operator $R^*$ is the adjoint of the Radon transform with respect to the $L^2$ scalar product. Notice that our constant $c_d$ differs from the constant in  \cite{Helgason2011} as we use the standard definition of the Hilbert transform and define $R^*$ as the adjoint of the Radon transform (as opposed to the dual transform).

\section{Multiscale tests for qualitative features}\label{sec:multiscale_test}
\subsection{Multiscale inference}\label{subsec}

The goal of this work is to derive confidence statements for qualitative features of the joint density of 
the random coefficients. In particular, we are interested in the detection of modes (local maxima) of the density. Following the approach of
\cite{schmidthieber13}, we express the features in terms of differential operators. To be precise, for a collection of compactly supported,
non-negative and sufficiently smooth test functions $\phi_{\bt,h}$ consider the integral 
 \begin{align}\label{test}
\int_{\mathbb{R}^d} \phi_{\bt,h}(\bb)\partial_\bv f_{\bbeta} (\bb)\, d\bb=\sum_{k=1}^d\mathrm{v}_k\int_{\mathbb{R}^d} \phi_{\bt,h}(\bb)\frac{\partial}{\partial \mathrm{b}_k}f_{\bbeta}(\bb)\,d\bb
\end{align} for some directional vector $\bv=(\rm v_1,\hdots,\rm v_d)^\top$ and $d\geq2$. Since there should not be any favored direction, we consider in the following radially symmetric test functions,
 \begin{align}
	\phi_{\bt,h}(\cdot) = \frac 1{h^d \Vol(\mathbb{S}^{d-2})}\phi\Big(\frac{\| \cdot -\bt\|}{h}\Big)
	\label{eq.def_radially_symm}
\end{align}
with a non-negative and sufficiently smooth kernel 
 $\phi:[0,\infty)\rightarrow [0,\infty)$ with $\int_0^\infty \phi(u) du=1$ and support on $[0,1]$.
Moreover, $\Vol(\mathbb{S}^{d-2})$ denotes the volume of the sphere $\mathbb{S}^{d-2}\subset \mathbb{R}^{d-1}$ and $\Vol(\mathbb{S}^0):=2$. Notice that $\phi_{\bt,h}$ is supported on the ball $B_h(\bt)$ with center $\bt$ and radius $h.$ 
The normalization for $\phi_{\bt,h}$ turns out to be convenient but does not entail that $\phi_{\bt,h}$ integrates to one.

If the integral \eqref{test} is positive, there exists a subset of $B_h(\bt)$ with positive Lebesgue measure on which $\partial_\bv f_{\bbeta}$ is positive.  On this subset, $f_{\bbeta}$ is thus strictly increasing in direction $\bv.$ Similarly, we can recover a decrease if the integral \eqref{test} is negative. To construct a statistical test for increases and decreases it is therefore natural to use an empirical counterpart of the functional defined in \eqref{test}. 

Let $\mathcal{T} = \{(\bt,h,\bv) : h\in(0,1], \mathbf a_1+h\leq \bt \leq \mathbf{a}_2-h,\bv\in \mathbb{S}^{d-1}\}$ for $\mathbf{a}_1\leq\mathbf{a}_2\in\mathbb{R}^d$,  where the inequalities for the 
vectors $\mathbf a_1,\mathbf a_2,\bt$ are understood componentwise. For statistical inference regarding the sign of the directional derivatives of $f_\bbeta$, we   fix a subset $\mathcal{T}_{n}\subset\mathcal{T}$ and test for all $(\bt,h,\bv)\in\ca{T}_n$ simultaneously the corresponding hypotheses of the form 
\begin{align}
\label{t5neu}H_{0,+}^{\bt,h,\bv}: \int_{\R^d} \phi_{\bt,h}(\bb)\partial_{\bv}f_\bbeta(\bb){d}\bb\leq0 ~~   
&\mbox{versus}  ~ ~H_{1,+}^{\bt,h,\bv}:\int_{\R^d} \phi_{\bt,h}(\bb)\partial_{\bv}f_\bbeta(\bb){d}\bb>0
\end{align}
and
\begin{align}
\label{t5neu1}H_{0,-}^{\bt,h,\bv}: \int_{\R^d} \phi_{\bt,h}(\bb)\partial_{\bv}f_\bbeta(\bb){d}\bb\geq0 ~~   
&\mbox{versus}  ~  ~H_{1,-}^{\bt,h,\bv}:\int_{\R^d} \phi_{\bt,h}(\bb)\partial_{\bv}f_\bbeta(\bb){d}\bb<0.
\end{align}
For constructing global tests, we can now argue as in \cite{Eckle}. Our main interest are the following three global testing problems (i)-(iii).

{\bf (i) Testing for the presence of a mode at a fixed location.} Tests for the hypotheses \eqref{t5neu} and \eqref{t5neu1} can be used for the detection of specific shape constraints such as a mode at a given point $\bb_0\in\R^d$. For this purpose, we consider several bandwidths/scales $h$ and for each $h$ consider pairs 
$(\bt_1,\bv_1),\hdots,(\bt_p,\bv_p)$, where $\bv_j,$ $j=1,\hdots,p$, are 
directional vectors and the test locations $\bt_j$ are points on the line $\{\bb_0+r\bv_j:r\geq h\}$ $(j=1,\hdots,p)$ in a neighborhood of $\bb_0$.
Inference
for the presence of a mode at the point $\bb_0$ can now be conducted by studying the testing problem
\begin{align}
\label{t5neu2}H_{0,-}^{\bt_j,h,\bv_j} ~~   
&\mbox{versus}  ~  ~H_{1,-}^{\bt_j,h,\bv_j},
\end{align}
with $h$ ranging over all chosen scales and $j=1,\hdots,p.$ Level and power of the mode test \eqref{t5neu2} for different designs are reported in Section \ref{loctest}. In Section \ref{multi}, we also show that it is essential to include several bandwidths/scales $h$ in order to separate modes which are close.

{\bf (ii) A global testing procedure for all modes.}
Simultaneous tests for the hypotheses \eqref{t5neu} and \eqref{t5neu1} can be used for a global testing procedure to detect all modes of the density on a domain. Compared to the previous case, we search for evidence over a range of different $\bb_0$ which then inflates the number of local tests.

{\bf (iii) A graphical representation of the local monotonicity behavior for bivariate densities.} Let $d=2$ and define a subset 
$\ca{T}_n=\{(\widetilde\bt_j,h_0,\widetilde\bv_j):j=1,\hdots,p\}$ for a fixed scale $h_0$ of the form
$\ca{T}_n=\ca{T}_{\mathbf{t}}\times\{h_0\}\times\ca{T}_{\mathbf{v}}$,
where $\ca{T}_{\mathbf{t}}$ contains the $p/|\ca{T}_{\mathbf{v}}|$ vertices of an equidistant grid of width $2h_0$ and 
$\ca{T}_{\mathbf{v}}$ contains the directions. We restrict the testing procedure to one fixed scale for an easy to read graphical representation. We consider four equidistant directions on 
$\mathbb S^1$ given by $\ca{T}_{\mathbf{v}}=\{\bv_1,-\bv_1,\bv_2,-\bv_2\}.$ Since $\ca{T}_{\mathbf{v}}=-\ca{T}_{\mathbf{v}}$, we have symmetry in the hypotheses, i.e. $H_{0,+}^{\widetilde\bt_j,h_0,\widetilde\bv_j}=H_{0,-}^{\widetilde\bt_j,h_0,-\widetilde\bv_j}$. Therefore, we test only $H_{0,-}^{\widetilde\bt_j,h_0,\widetilde\bv_j}$
for all triples $(\widetilde\bt_j,h_0,\widetilde\bv_j)\in\ca{T}_n.$
Figure \ref{g25} displays an example for the test outcome with the hypotheses in \eqref{t5neu} and $\ca T_n$ as above.
 An arrow in a direction  $\widetilde\bv_j$ at a location $\widetilde\bt_j$ represents a rejection of the corresponding hypothesis 
 $H_{0,-}^{\widetilde\bt_j,h_0,\widetilde\bv_j}$
 and provides an indication of a negative directional derivative 
 of $f_\bbeta$ in direction $\widetilde\bv_j$ at the  location $\widetilde\bt_j$. Thus, Figure \ref{g25} provides strong evidence that the density is trimodal with modes close to the locations $(-0.5,-0.5)^\top$, $(1.5,-0.5)^\top$, and $(0.5,1.5)^\top$. A detailed description of the settings used to generate Figure \ref{g25} and an analysis of the results is given in Section \ref{g27}. 
\begin{figure}[h]
\begin{center}
\includegraphics[width=0.35\textwidth]{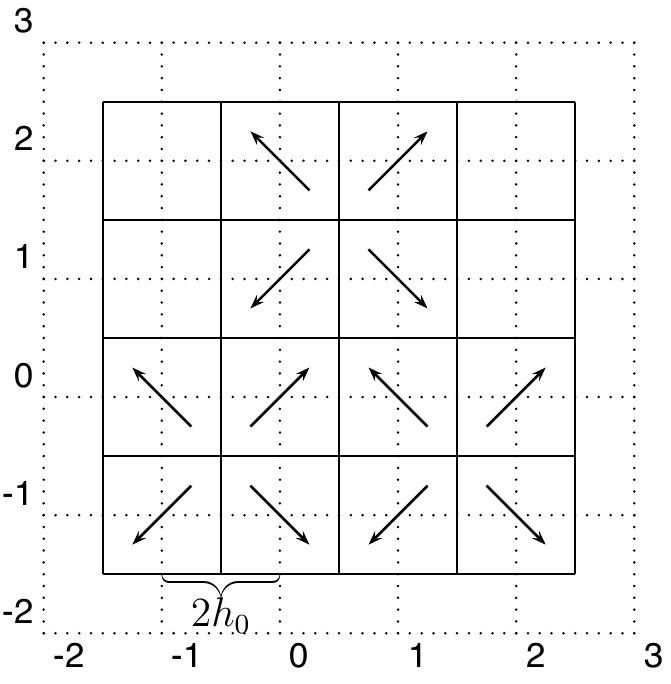}
\caption{\it Example of a global map for monotonicity of a bivariate density. Each arrow indicates an decrease in the respective direction.\\[-40pt]} 
\label{g25}
\end{center}
\end{figure}

We now derive an empirical counterpart of  the functional \eqref{test} in terms
of the Radon transform $Rf_\bbeta$. We make the following assumptions for the inversion of the Radon transform.

\begin{assump}\label{ass:radial_symmetric}
Suppose that
the function $\phi$ in \eqref{eq.def_radially_symm}  is $(d+2)$-times
continuously diffe\-rentiable with $\phi'(0)=\phi''(0)=0$.
\end{assump}
\begin{assump}\label{ass:f_beta}
Suppose that the density $f_\bbeta$ is compactly supported, continuously differentiable 
 and bounded from below in the test region by a constant $c_\bbeta>0$
\[
 f_\bbeta(\bb)\geq c_\bbeta \ \text{ for all } \ \bb\in[\mathbf a_1,\mathbf a_2].
\]
\end{assump}

\newenvironment{thmbis}[1]
  {\renewcommand{\theassump}{\ref{#1}$'$}%
   \addtocounter{assump}{-1}%
   \begin{assump}}
  {\end{assump}}

Assumption \ref{ass:f_beta} is too restrictive for quantum homodyne tomography (model \eqref{qht}), where the density $f_\bbeta$ is given by
the Wigner function. The Wigner function can take negative values and is not compactly supported. In this case, we replace Assumption \ref{ass:f_beta} by the following conditions. 
\begin{thmbis}{ass:f_beta}
Suppose that $f_\bbeta$ is continuously differentiable and, for some $\gamma',\ve>0$,
\begin{enumerate}
 \item  $\bb \mapsto \|\bb\|^{d+\ve} \displaystyle|f_\bbeta(\bb)|$ is bounded;
  \item $\displaystyle|f_\bbeta(\bb)-f_\bbeta(\bb')|\lesssim \frac{\|\bb-\bb'\|^{\gamma'}}{(1+\min\{\|\bb\|,\|\bb'\|\})^{d+\gamma'+\ve}}$
 for all $\bb,\bb'\in\R^d$;
 \item  
There exist constants $\delta, c_\bbeta>0$ such that for every hyperplane $P\subset\R^d$ with $P\cap [\mathbf a_1-\delta,\mathbf a_2+\delta] \neq\emptyset$, $ \int_P f_\bbeta(\bb)d\mu_{d-1}(\bb)\geq c_\bbeta.$
\end{enumerate}
\end{thmbis}

Under the assumptions above the inversion formula \eqref{eq.inversion_formula} holds for partial derivatives of the test functions $\partial_{\bv}\phi_{\bt,h}$. This is a direct consequence of 
Theorem 3.8 in \cite{Helgason2011}. The following lemma analyzes the structure of a partial derivative of the test function transformed by the operator $\Lambda$ introduced in \eqref{.3} and how this transform depends on $h$.

\begin{lem}
\label{lem.rad_symm_kernel_explicit}
Work under Assumption \ref{ass:radial_symmetric} and let
\begin{align}
	\widetilde{\phi}(z) := \int_0^\infty r^{d-2} \frac{\partial}{\partial z} \phi\left(\sqrt{z^2+r^2}\right) dr \qquad \mbox{with } z \in \R.
	\label{eq.tilde_phi_n_def}
\end{align}
Then
\begin{align*}
	\Lambda(\partial_{\bv} \phi_{\bt,h})(s,{\btheta}) = \frac{\langle {\btheta} , \bv \rangle}{h^{d+1}}
	(\mathcal{H}_d \widetilde{\phi}^{(d-1)})\Big(\frac{s - \langle \bt, {\btheta} \rangle}{h}\Big),
\end{align*}
where $\Lambda$ and $\phi_{\bt,h}$ are defined in \eqref{.3} and \eqref{test}, respectively.
Moreover,

\begin{tabular}{ll}
(i) & \  $\|\Lambda(\partial_{\bv} \phi_{\bt,h})\|_\infty \lesssim h^{-d-1}$; \\[0.3cm] 
(ii) & \  $\|\Lambda(\partial_{\bv} \phi_{\bt,h})\|_{L^k(\ca{Z})}^k\lesssim h^{-dk-k+1}$ for $k> 1$. \\
\end{tabular}
\end{lem}

For a given triple $(\bt,h,\bv)\in\ca T$ we study the  statistic
\begin{align*}
	T_{\bt,h,\bv}:= \frac{1}{n\sqrt{h}} \sum_{i=1}^n 
	\frac{\langle \bTheta_i, \bv \rangle }{f_{\bTheta}(\bTheta_i)} (\mathcal{H}_d\widetilde \phi^{(d-1)} )\Big( \frac{S_i - \langle \bt, \bTheta_i \rangle}{h} \Big).
\end{align*}
By Lemma \ref{lem.rad_symm_kernel_explicit}, the expectation of this statistic can be written as
\begin{align*}
	\E\big[T_{\bt,h,\bv}\big] = \int_{\mathbb{S}^{d-1}}\int_{\mathbb{R}} h^{d+ 1/2}\Lambda(\partial_{\bv} \phi_{\bt,h}) (s,{\btheta})  f_{S|\bTheta}(s|{\btheta}) ds d{\btheta},
\end{align*}
with $d\btheta$ being the surface measure on $\mathbb{S}^{d-1}$, i.e. $|S_\bTheta|=\int_{ S_\bTheta}\,d\btheta$ 
for any measurable $S_\bTheta\subseteq \mathbb{S}^{d-1}.$ By an application of the inversion formula introduced in \eqref{eq.inversion_formula} and Lemma 5.1 in  \cite{Helgason2011}, we obtain
\begin{align}\label{expectation}
	\E\big[T_{\bt,h,\bv}\big] 
	= -c_dh^{d+ 1/2}\int_{\mathbb{R}^d}\phi_{\bt,h}(\bb) \partial_{\bv} f_{\bbeta}(\bb) d\bb.
\end{align}
Up to rescaling, $T_{\bt,h,\bv}$ is thus an empirical counterpart of the functional defined in \eqref{test}.

The statistic $T_{\bt,h,\bv}$  depends on the density $f_\bTheta.$ In quantum homodyne tomography this density is known. For many other applications, however, $f_\bTheta$ needs to be estimated from the data. In this case, we use a standard cut-off kernel density estimator $\widetilde f_{\bTheta}$ for $f_\bTheta$ based on an additional sample $(S_i,\bTheta_i),$ $i=n+1,\hdots,2n$ which is independent of $(S_i,\bTheta_i),$ $i=1,\hdots,n$. See also Appendix \ref{3.1} for the definition of the estimator. We replace $f_\bTheta$ by its estimator and consider the test statistic
\begin{align}\label{.2}
	\widehat{T}_{\bt,h,\bv}:= \frac{1}{n\sqrt{h}} \sum_{i=1}^n 
	\frac{\langle \bTheta_i, \bv \rangle }{\widetilde f_{\bTheta}(\bTheta_i)} (\mathcal{H}_d\widetilde \phi^{(d-1)} )
	\Big( \frac{S_i - \langle \bt, \bTheta_i \rangle}{h} \Big).
\end{align} 

\subsection{Assumptions on the design}
\label{sec.assump_on_design}

As mentioned in Section \ref{sec:inv_prob}, the inverse problem might become severely ill-posed or non-identifiable if the density $f_\bTheta$ approaches zero for some directions. This section provides conditions on the design which ensure that $f_\bTheta$ has positive H\"older smoothness and is bounded from below and above. These results are of independent interest. 

In the random coefficients model \eqref{eq:model}, the density $f_\bTheta$ can be expressed in terms of the density
$f_\bX$ via $f_\bTheta(\btheta)=\int_{0}^\infty r^{d-1}f_\bX(r\btheta)dr$. To enforce that $f_\bTheta$ is bounded from below we restrict ourselves to designs where $\int_{0}^\infty r^{d-1}f_\bX(r\btheta)dr$ is bounded away from $0$. The formula also allows to relate the smoothness of $f_\bX$ to the smoothness of $f_\bTheta.$

Although, the random coefficients model with intercept \eqref{eq:model_with_int} could be viewed as a special case of the more general model \eqref{eq:model} it requires a different set of assumptions. For model \eqref{eq:model_with_int}, we write 
$f_\bX$ as a function of $\bs x =(x_2, \ldots, x_d) \in \mathbb{R}^{d-1}$ and obtain
\begin{align}
	f_\bTheta(\btheta)=\frac 1{2|\theta_1|^{d}}f_\bX\Big(\frac{\theta_2}{\theta_1}, \ldots, \frac{\theta_d}{\theta_1}\Big),
	\label{eq.rcm_with_i_fthete_fX}
\end{align}
see Appendix \ref{B} for a proof. A necessary condition to ensure that  $\inf_{\btheta} f_\bTheta(\btheta) >0,$ is given by
$f_\bX(\bs x) \gtrsim \|\bs x\|^{-d}$ as $\|\bs x\|\rightarrow \infty$. This corresponds to Cauchy-type tails of 
the design variables. Thinner tails will increase the ill-posedness of the problem. In order to avoid very technical proofs, we consider 
in the random coefficients model with intercept only the case where $(X_{i,2}, \ldots, X_{i,d})$ follows a multivariate Cauchy distribution, i.\,e.,
\begin{align}
	 f_\bX(\bs x)=\frac{\Gamma(d/2)}{\pi^{d/2}|\Sigma|^{1/2}\big(1+(\bs x-\bs \mu)^\top \Sigma^{-1}
 (\bs x-\bs\mu)\big)^{d/2}}\quad \tn{for}\quad\bs x\in\R^{d-1},
 	\label{eq.mult_cauchy}
\end{align}
with $ \bs\mu\in\R^{d-1}$ and $\Sigma \in\R^{(d-1)\times (d-1)}$ a symmetric and positive definite matrix. We can compute $f_\bTheta$ explicitly using \eqref{eq.rcm_with_i_fthete_fX}
\begin{align}
	 f_\bTheta(\btheta)=\frac{\Gamma(d/2)}{2\pi^{d/2}|\Sigma|^{1/2}\big(\theta_1^2+
 ((\mathrm{sgn}(\theta_1)\theta_j-|\theta_1|\mu_j)_{j=2}^d)^\top \Sigma^{-1}
 ((\mathrm{sgn}(\theta_1)\theta_j-|\theta_1|\mu_j)_{j=2}^d)\big)^{d/2}},
 	\label{eq.ftheta_explicit}
\end{align}
where $\mathrm{sgn}(\cdot)$ denotes the signum function.
In this case, $f_\bTheta$ is bounded from above and below and is continuously differentiable on the hemispheres 
$\sd_+:= \{\btheta \in \sd\mid\theta_1 > 0\}$ and $\sd_- := \{\btheta \in \sd\mid\theta_1 < 0\}$. In particular, if
$(X_{i,2}, \ldots, X_{i,d})$ is standard Cauchy, then the density $f_\bTheta$ is constant. This leads to the following assumptions on the design.
  
\begin{assump}\label{p1} In model \eqref{eq:model}, suppose that

\begin{tabular}{ll}
(i) & \  $f_\bX(\bs x)\lesssim \|\bs x\|^{-d-\ve} $ for all $\bs x\in\R^d$ and some $\ve>0$; \\[0.3cm] 
(ii) & \  $\displaystyle\int_{0}^\infty r^{d-1}f_\bX(r\btheta)dr\geq c>0$ for all $\btheta\in\sd$; \\[0.3cm] 
(iii) & \  $\displaystyle |f_\bX(\bs x)-f_\bX(\bs x')|\lesssim\frac{\|\bs x-\bs x'\|^\gamma}{1+\|\bs x\|
^{d+\gamma+\ve}}$ for $\bs x,\bs x'\in\R^d$, $\|\bs x\|=\|\bs x'\|$, and  $\gamma>0.$ \\
\end{tabular}

In model \eqref{eq:model_with_int}, assume that  $f_{\bX}$ is of the form \eqref{eq.mult_cauchy} with $\Sigma$ a symmetric and positive definite matrix.
\end{assump}
In quantum homodyne tomography we set a global $\gamma$ equal to the minimum of $\gamma'$ from Assumption 2' and $\gamma$ from Assumption \ref{p1}.

It is important to notice that statistical testing in the random coefficients model relies on two unrelated sets of assumptions. Firstly, there are assumptions on the density of the random coefficients as introduced in Section \ref{subsec}. Restrictions of this kind are common in statistical inference for an unknown density. On the other hand, there are assumptions (see Assumption \ref{p1} above) on the design. These assumptions control the ill-posedness of the problem.

Note that Assumption \ref{p1} (ii) can be weakened to $\int_{-\infty}^\infty r^{d-1}f_\bX(r\btheta)dr\geq c>0 \quad \text{for all }\btheta\in\sd$. For example, if the support of $\Theta$ is a hemisphere this condition can hold while Assumption \ref{p1} (ii) is violated. This relaxation can be achieved by multiplying independently generated $\zeta_i$ to $\Theta_i$ and $S_i$ as proposed for model \eqref{eq:model_with_int} in Section \ref{sec:inv_prob}.

\subsection{Asymptotic properties}

This section presents the main theoretical result of the paper stating that the standardized and properly calibrated test statistic \eqref{.2} can be uniformly approximated by a maximum of a Gaussian process. For that we need the definition of a Gaussian process on the cylinder $\mathcal{Z}$. To this end, let $\mathcal{B}(\mathcal{Z})$ be the Borel $\sigma$-algebra on $\mathcal{Z}$. Define the $\sigma$-finite measure 
  \begin{align*}
    \nu:
    \begin{cases}
      \mathcal{B}(\mathcal{Z})&\rightarrow\mathbb{R}^+_0,\\
      E &\mapsto\nu(E )=\int_{\mathbb{S}^{d-1}}\int_{\mathbb{R}}\mathds{1}_{E }(\btheta,s)\,ds\,d\btheta.
    \end{cases}
  \end{align*}
Let $\bigl(\mathcal{B}(\mathcal{Z})\bigr)_{\nu}$ denote the collection of all sets of finite $\nu$-measure and let $W$ denote 
Gaussian $\nu$-noise on $\bigl(\mathcal{B}(\mathcal{Z}),\nu\bigr)$. For disjoint sets 
$E_1,E_2 \in\bigl(\mathcal{B}(\mathcal{Z})\bigr)_{\nu}$ this implies 
  \begin{align*}
    W(E_1)\sim\mathcal{N}(0,\nu(E_1)),\quad W(E_1\cup E_2)=W(E_1)+W(E_2)
    \text{ a.s.}\quad\text{and}\quad W(E_1)\,\bot\, W(E_2)
  \end{align*}
 \cite[][ Chapter 1.4.3]{Adler2007}. $W$ is a random, finitely additive, signed measure. Integration w.r.t. $W$ can be defined similarly to Lebesgue-integration, starting with a definition for simple functions and an extension to general $f\in L^2(\nu)$ via approximation by simple functions in the $L^2$-limit. Integration with respect to $W$  yields
   \begin{align*}
     \int_{E }\,W(dsd\btheta)=W(E )\sim\mathcal{N}(0,\nu(E ))\quad\text{for}\quad E\in\bigl(\mathcal{B}(\mathcal{Z})\bigr)_{\nu},
   \end{align*}
   \begin{align*}
     W(f):=\int_{\mathbb{S}^{d-1}}\int_{\mathbb{R}}f(s,\btheta)\,W(dsd\btheta)\sim\mathcal{N}(0,\|f\|_{L^2(\nu)})\quad\text{for}\quad f\in L^2(\nu),
   \end{align*}
and $\mathrm{Cov}(W(f),W(g))=\langle W(f),W(g)\rangle_{L^2(\mathbb{P})}=\langle f,g\rangle_{L^2(\nu)}$ for $f,g\in L^2(\nu),$ where $L^k(\mathbb{P})$ denotes the collection of all random variables whose first $k$ absolute moments exist.   For more details, cf.  \cite{Adler2007}, Chapter 5.2.

Let us provide some heuristic for the Gaussian approximation of $T_{\bt,h,\bv}$. The process $(\bt,h,\bv) \mapsto \sqrt{n}(T_{\bt,h,\bv}-\E[T_{\bt,h,\bv}])$ has in the important case $\E[T_{\bt,h,\bv}]=0$ the same mean and covariance structure as the Gaussian process
\begin{align}\label{Zth}
	X_{\bt,h,\bv} = h^{-1/2}\int_{\mathbb{S}^{d-1}}\int_{\mathbb{R}}\langle \btheta, \bv\rangle 
	(\mathcal{H}_d\widetilde \phi^{(d-1)})\Big(\frac{s-\langle \bt, \btheta \rangle}{h}\Big) 
	\frac{\sqrt{f_{S,\bTheta}(s,{\btheta})}}{{f_{\bTheta}({\btheta})}} W(dsd\btheta).
\end{align}
In the proof of Theorem \ref{thm.main} below we show that the expectation $\E[T_{\bt,h,\bv}]$ is asymptotically negligible in the limit process. 
The test statistic and the Gaussian process depend, however, on the unknown densities $f_{S,\bTheta}$ 
and $f_{\bTheta}$ which have to be estimated from the second part of the sample. To this end, we use the standard cut-off kernel density estimates $ \widetilde f_{\bTheta}$ and $\widetilde 
f_{S,\bTheta}$  defined in Appendix \ref{3.1}. For Gaussian $\nu$-noise 
$W$ that is independent of the data let
\begin{align*}
\widehat X_{\bt,h,\bv} :=  h^{-1/2}\int_{\mathbb{S}^{d-1}}\int_{\mathbb{R}}  \langle {\btheta} , \bv \rangle(\mathcal{H}_d
\widetilde \phi^{(d-1)} )\Big(\frac{s - \langle \bt, {\btheta} \rangle}{h}\Big) \frac{\sqrt{\widetilde 
f_{S,\bTheta}(s  ,\btheta)}}{\widetilde f_{\bTheta}(\btheta)} W(dsd\btheta)
\end{align*}
and
\begin{align}
	\widehat\sigma_{\bt,h,\bv} :=  \Big( \int_{\mathbb{S}^{d-1}}\int_{\mathbb{R}} \langle \btheta, \bv \rangle^2
	\big((\mathcal{H}_d\widetilde \phi^{(d-1)})(s) \big)^2 \frac{\widetilde f_{S,\bTheta}(\langle \bt, \btheta \rangle+hs,
	\btheta)}{\widetilde f_{\bTheta}(\btheta)^2} dsd\btheta  \Big)^{1/2}.
	\label{eq.sigma_th_def}
\end{align}
The Gaussian approximation result for the family of test statistics $\widehat{T}_{\bt,h,\bv}$ 
holds for a finite subset $\mathcal{T}_{n} \subset \mathcal{T}$. Its cardinality may, however, grow polynomially
of arbitrary degree with the sample size. Moreover, the range of bandwidths must be bounded from above and
below by $h_{\max}$ and $h_{\min}$, both converging to zero as $n$ goes to infinity. The precise conditions are summarized in the following assumption.
\begin{assump}\label{1.20}
When working in model \eqref{eq:model}, let $\gamma$ be as defined in Assumption \ref{p1}. When considering model \eqref{eq:model_with_int}, let $\gamma=1.$ Set 
$h_{\min} :=  \min \{h : (\bt,h,\bv) \in \mathcal{T}_{n}\}$ and $h_{\max} :=  \max \{h : (\bt,h,\bv) \in 
\mathcal{T}_{n}\}.$ Suppose that 
  $|\ca{T}_n|=p\lesssim n^{L}$ for some $L>0$ and  $h_{\tn{max}}\lesssim \log(n)^{-14\gamma/(d-1)-5}n^{2\gamma/(d-1)-1}\wedge o(1)$,
  $h_{\min} \gtrsim n^{-1+\ve}$ for some $\ve>0$.
\end{assump}

Let $\ca{A}_p$ be the set of half-open hyperrectangles in $\R^p$, i.e. every $A\in\ca{A}_p$ has the representation $A=\{\bs x\in\R^p:-\infty< \bs x\leq\bs a\}$ for some $\bs a\in\R^p$. For finite sets $S_n$ and two stochastic processes $(X_{s,n})_{s\in S_n}$ and $(\widetilde X_{s,n})_{s\in S_n},$ which are defined on the same probability space, we write
\begin{align*}
	(X_{s,n})_{s\in S_n} \leftrightarrow (\widetilde X_{s,n})_{s\in S_n}
\end{align*}
if $\lim_n \sup_{A \in \ca{A}_{|S_n|}} | \P( (X_{s,n})_{s\in S_n} \in A ) - \P( (\widetilde X_{s,n})_{s\in S_n} \in A )| = 0.$
\begin{thm}
\label{thm.main} For the calibration of the standardized statistic define 
  \begin{align*}
    \alpha_h:= \sqrt{(3d-1)\log(1/h)}\quad\text{and}\quad\beta_h:=\frac{\sqrt{\log(e/h)}}{\log(\log(e^e/h))}.
  \end{align*}
Then, under Assumptions \ref{ass:radial_symmetric}-\ref{1.20},
\begin{align*}
\Big( \beta_h\Bigl(\sqrt{n} \frac{|\widehat T_{\bt,h,\bv}-\E[T_{\bt,h,\bv}]|}
{\widehat{\sigma}_{\bt,h,\bv}} - \alpha_h\Bigr)\Big)_{(\bt,h,\bv)\in \ca{T}_n}
\leftrightarrow \Big(\beta_h\Bigl( \frac{|\widehat X_{\bt,h,\bv}|}{\widehat{\sigma}_{\bt,h,\bv}} - \alpha_h\Bigr)\Big)
_{(\bt,h,\bv)\in \ca{T}_n}.
\end{align*}
Furthermore, almost surely conditionally on the estimators $\widetilde f_{S,\bTheta}$ and $\widetilde f_\bTheta,$ the limit distribution is bounded in probability by a constant that is independent of the sample size $n$.
\end{thm}

\subsection{Construction of the multiscale test}
With the previous theorem, we can now construct simultaneous statistical tests for the hypotheses \eqref{t5neu} and \eqref{t5neu1}. If the constant $c_d$ is positive then the method consists of rejecting the hypotheses $H_{0,+}^{\bt,h,\bv}$
in \eqref{t5neu} for  small  values of  $\widehat{T}_{\bt,h,\bv}$ and rejecting $H_{0,-}^{\bt,h,\bv}$ in \eqref{t5neu1}
for  large  values of $\widehat{T}_{\bt,h,\bv}$, and vice versa if $c_d$ is negative. 
Theorem \ref{thm.main} is used to control the multiple level of the tests.
Let $\alpha\in(0,1)$ and denote by $\kappa_n(\alpha)$ the smallest number such that
\[
 \P\left(\sup_{(\bt,h,\bv)\in\ca T_n}\beta_h\bigg( \frac{|\widehat X_{\bt,h,\bv}|}{\widehat{\sigma}_{\bt,h,\bv}} - \alpha_h\bigg)
\leq\kappa_n(\alpha)\right)\geq 1-\alpha.
\]
By Theorem \ref{thm.main}, $\kappa_n(\alpha)$ is bounded uniformly with respect to $n$. Define
 for $(\bt,h,\bv)\in\ca T_n$ the quantiles
\beq{t11j}
 \kappa^{\bt,h,\bv}_n(\alpha)=\frac{\widehat{\sigma}_{\bt,h,\bv}}{\sqrt{n}}\big(\beta_h^{-1}\kappa_n(\alpha)
+\alpha_h\big)
\eeq
and reject the  hypothesis  \eqref{t5neu}, if
\beq{t8}\mathrm{sgn}(c_d)\widehat T_{\bt,h,\bv} <-\kappa^{\bt,h,\bv}_n(\alpha).\eeq 
Similarly, the  hypothesis  \eqref{t5neu1} is rejected, whenever
\beq{t81}\mathrm{sgn}(c_d)\widehat T_{\bt,h,\bv} >\kappa^{\bt,h,\bv}_n(\alpha).\eeq
\begin{thm}\label{t10}
Let the assumptions of Theorem \ref{thm.main} hold and assume that the tests  \eqref{t8} and \eqref{t81} for the hypotheses \eqref{t5neu} and \eqref{t5neu1} are performed
simultaneously for all $(\bt,h,\bv)\in\ca T_n$.
The  probability of at least one false rejection of any of the tests is asymptotically at most $\alpha$, i.e.
\begin{align*}
\P\Big(\exists  (\bt,h,\bv)\in\ca T_n:|\widehat T_{\bt,h,\bv}|>\kappa^{\bt,h,\bv}_n(\alpha)\Big)\leq\alpha+o(1)
\end{align*}
for $n\rightarrow \infty$.
\end{thm}

Based on the previous result, we now propose a method for the detection and localization of modes on a subdomain. For convenience, we only consider the case of a hyperrectangle $[\mathbf a_1,\mathbf a_2] \subset \mathbb{R}^d$. We study the case that there is a mode $\bb_0$ in the interior $(\mathbf a_1,\mathbf a_2).$ For the multiscale test to have power we need that the set of local tests is rich enough. This can be expressed in terms of conditions on $\ca{T}_n.$ Let $H(\ca{T}_n)$ be the set of all scales/bandwidths such that for every scale $h$ in this set all triples $(\bt,h,\bv)$ are in $\ca{T}_n,$ where $\bt$ ranges over all grid points of an equidistant grid with component wise mesh size $h$ in the hyperrectangle $[\mathbf a_1+h,\mathbf a_2-h]$, and $\bv$ ranges over all grid points of a grid of $S^{d-1}$ and grid width
converging to zero with increasing sample size.

The testing procedure is as follows. For any $\bb_0$ in $(\mathbf a_1,\mathbf a_2),$ let $\ca{T}_n^{\bb_0}$ be the set of all sequences of triples  $(\bt_n,h_n,\bv_n)_n \in\ca{T}_n$ such that $ch_n\geq\|\bb_0-\bt_n\|\geq2h_n$ for some sufficiently large $c>2$ and $\angle (\bt_n-\bb_0,\bv_n) \rightarrow 0$ as $n\rightarrow\infty$, where $\angle$ denotes the angle between two vectors. The previous conditions ensure that several such sequences can be found. If for all  triples in $\ca{T}_n^{\bb_0}$ all local tests \eqref{t81} reject the hypotheses \eqref{t5neu2}, we have evidence for the existence of a mode at the point $\bb_0$.
By choosing the test locations as the vertices of an equidistant grid
no prior knowledge about the location of $\bb_0$ has to be assumed. Theorem \ref{.1} below  states that the procedure
detects all modes of the density
with probability converging to one as $n\rightarrow\infty$.

\begin{thm}\label{.1}
Assume the conditions of Theorem \ref{thm.main} and suppose that for any mode $\bb_0$ in $(\mathbf a_1,\mathbf a_2)$ there are functions $g_{\bb_0}:\mathbb{R}^d\rightarrow\mathbb{R}$,  
$r_{\bb_0}:\mathbb{R}\rightarrow\mathbb{R}$ 
such that  the density $f_\bbeta$
 has a representation of the form
\beq {frep}
 f_\bbeta(\bb)\equiv (1+g_{\bb_0}(\bb))r_{\bb_0}(\|\bb-\bb_0\|)
\eeq
in an open neighborhood of $\bb_0$.
Furthermore, let $g_{\bb_0}$ be differentiable in an open neighborhood of $\bb_0$ with $g_{\bb_0}(\bb)=o(1)$ and
$\langle\nabla g_{\bb_0}(\bb),\mathbf{e}\rangle=o(\|\bb-\bb_0\|)$ when $\bb\rightarrow \bb_0$  for all
$\mathbf{e}\in\mathbb{R}^d$ with $\|\mathbf{e}\|=1$. In addition, let $\widetilde{f}_{\bb_0}$  be differentiable in an open
neighborhood
of zero with $\widetilde{f}_{\bb_0}'(h)\leq-ch(1+o(1))$ for  $h\rightarrow0$.

If $\min \{h: h\in H(\ca T_n)\}\geq C\log(n)^{1/(2d+3)}n^{-1/(2d+3)}$ is nonempty for some sufficiently large constant $C>0,$ then the procedure described in the previous paragraph detects
 the mode  $\bb_0$  with probability converging to one as $n\rightarrow\infty.$
\end{thm}

The rate for the localization of the modes is $n^{-1/(2d+3)}$ (up to some logarithmic factor). Since $2d+3=d+2\tfrac{d-1}2+4,$ this rate matches the optimal rate for mode detection in an inverse problem with ill-posedness of degree $(d-1)/2$ over a 2-H\"older class. Assumption \ref{1.20} requires, however, $h_{\tn{max}}\lesssim \log(n)^{-14\gamma/(d-1)-5}n^{2\gamma/(d-1)-1}.$ To be able to include scales of the order $n^{-1/(2d+3)}$ we need $\gamma>(d^2-1)/(2d+3).$ The right hand side is smaller than one  for $d=2,3.$

\section{Finite sample properties}\label{Sec4}
\def\theequation{4.\arabic{equation}}
\setcounter{equation}{0}
In this section we illustrate the finite sample properties of the proposed test in a bivariate and a trivariate
setting. In the bivariate setting we illustrate  how simultaneous tests for the
hypotheses \eqref{t5neu} and \eqref{t5neu1} can be used to obtain a graphical representation of the local 
monotonicity properties of the density. In the trivariate setting we investigate
the performance of the  test for modality at a given point
$\bb_0$ (see the hypotheses in \eqref{t5neu2}) and the dependence of its power on the distribution of $\bX$. 

As test function we consider the simplest polynomial which satisfies the conditions of Assumption \ref{ass:radial_symmetric} for $d=2,3$, that is,
\[
 \phi(x)=c(56x^3+21x^2+6x+1)(1-x)^6\mathbbm{1}\{x\leq 1\},\quad x\in[0,\infty),
\]
with $c$ such that $\int \phi =1.$  Figure \ref{phi}
displays the function $\ca H_d\widetilde\phi^{(d-1)}$ for $d=2,3.$ Throughout this section the nominal level is fixed as $\alpha=0.05$, and all level and power statements are in percent. Except for Table \ref{ta3}, none of the simulations in this section assumes knowledge of the design density $f_\bTheta$ or uses a parametric specification of it.
\begin{figure}[ht]
\begin{center}
\includegraphics[width=0.3\textwidth]{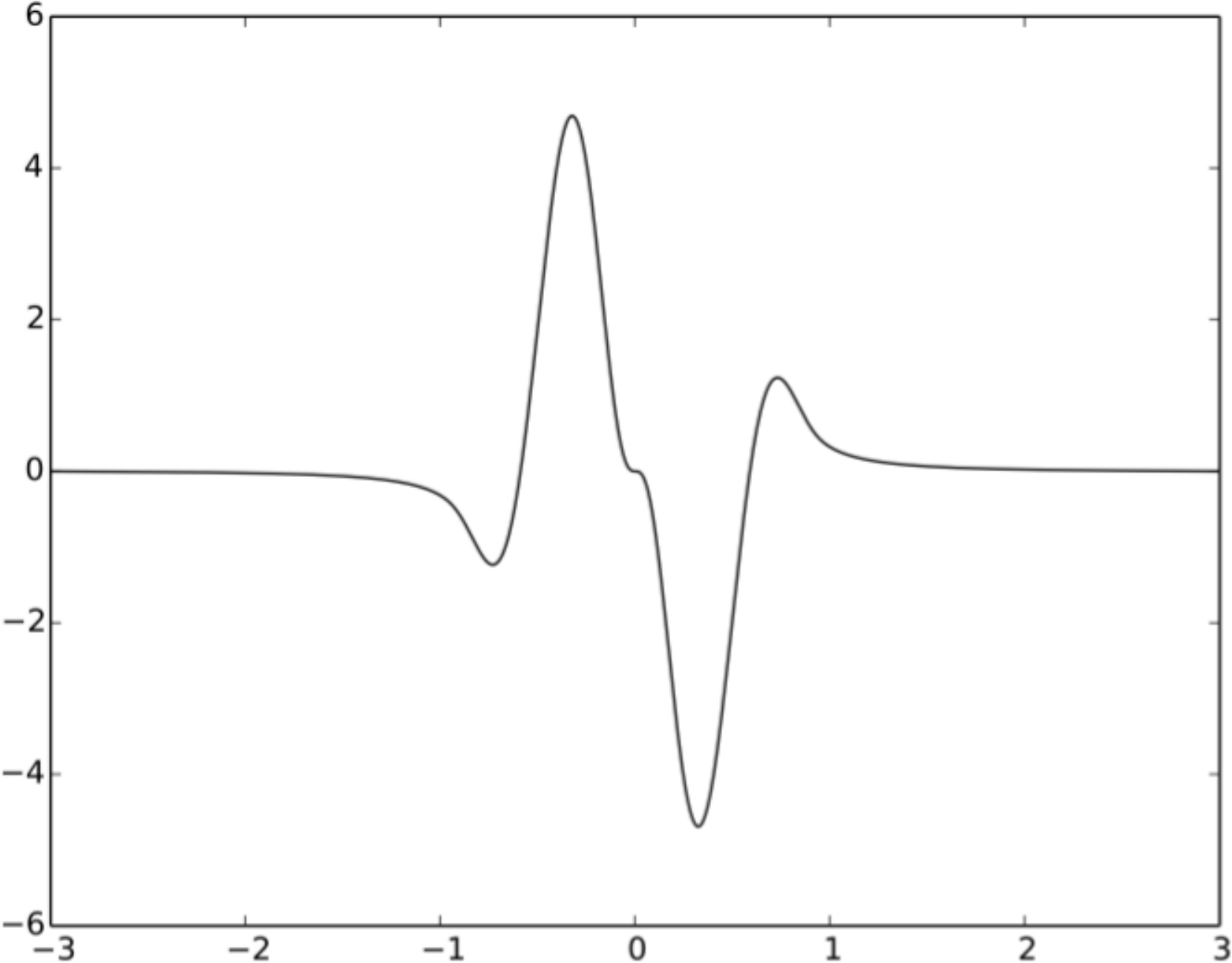}\qquad
\includegraphics[width=0.3\textwidth]{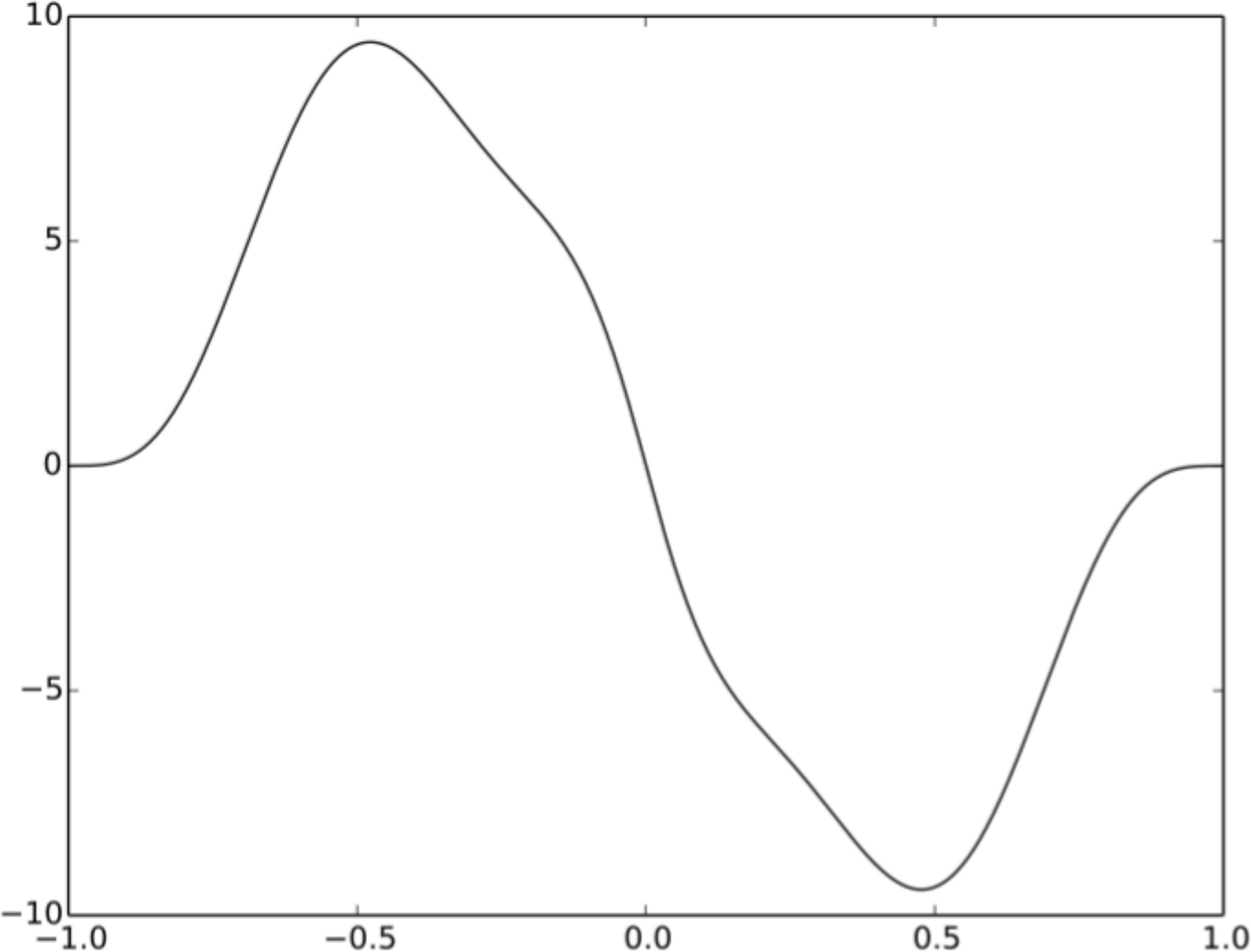}
\caption{\it The function $\ca H_d\widetilde\phi^{(d-1)}$ for $d=2$ (left) and $d=3$ (right).\\[-20pt]}\label{phi}
\end{center}
\end{figure}

\subsection{Inference about local monotonicity  of a bivariate density}\label{g27}
We follow the multiscale approach in Section \ref{subsec}
to obtain a graphical representation of the monotonicity behavior for a bivariate density of random coefficients. To test the hypotheses
\eqref{t5neu1} we use \eqref{t81} with $\ca T_n=\ca{T}_\bt\times\{h_0\}\times\ca T_\bv.$ Here,  $h=h_0=0.5$ is fixed and the set of test 
locations $\ca{T}_\bt$ is defined as the set of vertices on an equidistant grid in the square 
$[-1,2]^2$ with width one. Finally, the set of test directions is
$$\ca{T}_\bv=\big\{\bv_1=-\bv_3=\sqrt{2}^{-1}(1,1)^\top,\bv_2=-\bv_4=\sqrt{2}^{-1}(-1,1)^\top\big\}.$$
The data are simulated with $f_\bbeta$ the density of the normal mixture $\tfrac 13 {\cal N}((-0.4,-0.57)^\top,0.2I)+ \tfrac 13 {\cal N}((1.5,-0.52)^\top,0.2I) + \tfrac 13 {\cal N}((0.45,1.6)^\top,0.15I).$ The design is chosen such that $\bTheta_i$ is uniformly distributed
on the sphere $\mathbb S^1$. Figure \ref{g25} in Section \ref{subsec} displays the monotonicity behavior
of the density $f_\bbeta$ based on sample size $n=20000$. Each arrow
at a location $\bt$ in direction $\bv$ displays a rejection of a hypothesis
 \eqref{t5neu1}. The map indicates the existence of modes around the points $(-0.5,-0.5)^\top$, $(1.5,-0.5)^\top$, and $(0.5,1.5)^\top$ and thus detects the true modes fairly well.

\subsection{Influence of the design on the power}\label{loctest}
Given the random coefficients model with $d=3$, we study the power of the test for the existence of a mode
at a given location $\bb_0$ considering only few local tests. The postulated mode is given by the point
$\bb_0 =(0,0,0)^\top$ and we take $\ca T_n =\{ (\bt, h, \bv): h=1,\bt=\bv \in \{\pm \mathbf{e}_i, i=1,2,3\}\}$ with $\mathbf{e}_i$ the $i$-th standard unit vector in $\mathbb{R}^3.$ We conclude that $f_\bbeta$ has 
a  local maximum at the point $\bb_0$, whenever all hypotheses $H_{0,-}^{\bt_j,h_0,\bv_j},$ $j=1,\hdots,6,$ are rejected, i.e.
\beq{modtest}
\mathrm{sgn}(c_d)\widehat T_{\bt_j,h_0,\bv_j} >\kappa^{\bt_j,h_0,\bv_j}_n(\alpha), \quad \text{for all} \ j=1,\hdots,6,
 \eeq
 where $\kappa^{\bt_j,h_0,\bv_j}_n(\alpha) $ is defined by  \eqref{t11j}. Recall that the quantiles $\kappa^{\bt_j,h_0,\bv_j}_n(\alpha) $ are constructed in such a way that the probability of at least one false rejection within the six tests \eqref{modtest} is bounded by $\alpha$. However, the mode test detects the presence of a mode whenever all six tests \eqref{modtest} are rejected at the same time. The multiscale method is therefore rather conservative for the specific task of mode detection. In this simulation, we also study a calibrated version of the test where the quantiles are chosen such that the test keeps its nominal level  $\alpha=0.05$ and detects the presence of a non existing mode in about 5 percent of the simulation runs. For the calibration of the test we work under the null hypothesis assuming that $f_{\bbeta}$ is uniform. Therefore, knowledge about the true unknown density $f_\bbeta$ is not required. 

\textbf {Numerical simulations for random coefficients model without intercept:} At first, we consider model \eqref{eq:model} with uniform design $\bX_i \sim \Unif[-5,5]^3.$ To study the level of the test, we used $f_\bbeta(\bbeta) \propto \mathbf{1}(\bbeta \in [-5,5]^3).$ For the power we took $f_\bbeta$ as the density of a trivariate standard normal distribution. All results are based on the local tests \eqref{modtest} and 1000 repetitions. Level simulations with the theoretical quantiles confirm that the multiscale test keeps its nominal level as  the percentage of false rejections of at least one of the six hypotheses in  \eqref{modtest}  for sample size $n\in\{250,500,1000\}$ is nearly $5$ percent. The results of mode test are reported in Table \ref{ta5}. 

Next, we investigate an asymmetric distribution of the directions $\bTheta_i$ by sampling $\bX_i \sim \mathcal{N}((3,0,0)^\top,2I)$, with $I$ the $3\times 3$ identity matrix. We only consider the calibrated mode test. Results are reported in columns five and six in Table \ref{ta5}. Compared to the case of uniform design, we observe a decrease in the power of the mode test \eqref{modtest}. The explanation is that the uniform design on the $\bX_i$ induces a more uniform distribution of $\bTheta_i$ on the sphere which makes it simpler to recover information about the joint density as discussed in Section \ref{sec:inv_prob}.
\begin{table}[ht]
\begin{center}
\begin{tabular}{ cccc|cc}
&\multicolumn{3}{c|}{$\bX_i \sim \Unif[-5,5]^3$}&\multicolumn{2}{c}{$\bX_i \sim \mathcal{N}((3,0,0)^\top,2I)$}\\
$n$&power& level (cal.)&power (cal.)&level (cal.)&power (cal.) \\
\hline
250&9&4.5&91.0&4.7&66.1\\
500&15.7&4.6&99.1&4.5&78.8\\
1000&79.1&5.0&100&4.5&90.6\\
\end{tabular}
\caption{\it Simulated level and power of the test \eqref{modtest}  for  a mode considering uniform design $\bX_i \sim \Unif[-5,5]^3$ and normal design $\bX_i \sim \mathcal{N}((3,0,0)^\top,2I)$. Results with theoretical quantiles are in the second  column. Results where $\kappa_n^{\bt_j,h_0,\bv_j}(\alpha)$ in \eqref{modtest} are replaced by calibrated quantiles are in all other columns.\\[-40pt]}
\label{ta5}
\end{center}
\end{table}

\textbf{Numerical simulations for random coefficients model with intercept:}
We study model \eqref{eq:model_with_int} with $d=3.$
In a first simulation, we sample the random vectors $(X_{i,2},X_{i,3})^\top$ from a standard bivariate Cauchy distribution, such 
that the density $f_\bTheta$ is constant.  Except for the different design, we consider otherwise the same test settings as above. The simulated level and power of the calibrated version of the test \eqref{modtest} are reported in Table \ref{ta3}. To investigate the influence of the estimation of  the design density $f_\bTheta$ on the power of the test we also perform simulations in which we assume that the density $f_\bTheta$ is known to be constant.  These are shown  in Table \ref{ta3}, fourth and fifth column.
\begin{table}[ht]
\begin{center}
\begin{tabular}{ c cc | c c  }
&\multicolumn{2}{c|}{$f_\bTheta$ unknown}&\multicolumn{2}{c}{$f_\bTheta$ known}\\
$n$& level (cal.)&power (cal.)&level (cal.)&power (cal.) \\
\hline
250&4.8&91.3&5.0&93.3\\
500&5.2&99.0&5.2&99.7\\
1000&5.3&100&5.3&100\\
\end{tabular}
\caption{\it Same as Table \ref{ta5} but for random coefficients model with intercept and $(X_{i,2},X_{i,3})^\top$ from a standard bivariate Cauchy distribution. In the fourth and fifth column we assume that the density $f_\bTheta$ is known.\\[-40pt]}
\label{ta3}
\end{center}
\end{table}
Compared to the power approximations for unknown $f_\bTheta$, we observe only  a slight increase in power of the test for known $f_\bTheta$.

Finally, we consider two designs which do not satisfy Assumption \ref{p1}. Table \ref{ta6} reports the level and power for the same setting as above except that now $(X_{i,2},
X_{i,3})^\top$ is drawn from a standard normal distribution or $(X_{i,2},
X_{i,3})^\top \sim \Unif[-5,5]^2.$ We observe only a slight decrease in the power of the test for normally distributed design compared
to the setting where Assumption \ref{p1} holds. Even under uniform design, the test performs fairly well.
\begin{table}[ht]
 \centering
\begin{tabular}{ccc|cc }
&\multicolumn{2}{c|}{$(X_{i,2},
X_{i,3})^\top\sim \ca N (0,1)$}&\multicolumn{2}{c}{$(X_{i,2},
X_{i,3})^\top \sim \Unif[-5,5]^2$}\\
$n$& level (cal.)&power (cal.)&level (cal.)&power (cal.) \\
\hline
250&5.1&88.3&5.1&64.8\\
500&5.5&98.3&5.1&78.1\\
1000&5.1&100&5.5&89.4\\
\end{tabular}
\caption{\it Same as Table \ref{ta5} but for random coefficients model with intercept and $(X_{i,2},X_{i,3})^\top \sim \mathcal{N}((0,0)^\top, I)$ (second and third column) and $(X_{i,2},
X_{i,3})^\top \sim \Unif[-5,5]^2$ (fourth and fifth column).
}
\label{ta6}
\end{table}

\subsection{Multiscale mode testing}\label{multi}

For multimodal densities which have a second mode close to the test location $\bb_0$ testing  different bandwidths 
simultaneously can be advantageous to separate the modes. This is illustrated by the following example, where we consider the random coefficients model without
intercept with $d=2.$ The data are simulated with $f_\bbeta$ being the density of the normal mixture 
$$\frac 12 \ca{N}\Big((0,0)^\top,\Big( \;\begin{matrix} 0.05&0\\0&0.4 \end{matrix} \;\Big)\Big) + \frac 12 \ca{N}((2,0)^\top,0.1\cdot I).$$ We consider a design such that $\bTheta$ is uniformly distributed on the circle $\mathbb{S}^1$. The level $\alpha$ is fixed to five percent. We 
conducted simultaneously twelve tests with three different scales $h\in \{0.5, 1, 2.5\}$ for the hypotheses \eqref{t5neu2} with $\bb_0=(0,0)^\top$. The tests are given by $\{(\bt_i, h_j, \bv_i): i=1,\ldots,4; h_j \in \{0.5, 1, 2.5\}, \bv_i\in \{\pm \mathbf{e}_1,\pm \mathbf{e}_2\}, \bt_i=h_j\bv_i\}.$ We analyze the outcome of the twelve tests in two ways. Firstly, we investigate the performance of  each of the three tests for modality for the bandwidths $h=2.5$, $h=1$ and $h=0.5$ separately. Secondly, we consider the performance of a combined test for modality using two scales $h=1$ and $h=0.5$. The power approximations for sample sizes $n\in\{2000,5000,15000\}$ based on 1000 repetitions are displayed in Table \ref{ta8}. 

\begin{table}[ht]
 \centering
\begin{tabular}{ c c c  c | c  }
$n$& $h=2.5$&$h=1$&$h=0.5$&$h\in \{0.5, 1\}$\\
\hline
2000&100&0&0&25\\
5000&100&0&1&86.3\\
15000&100&0&68.7&100\\
\end{tabular}
\caption{\it Power of the multiscale test for modality. The mode tests for the bandwidths $h=2.5$, $h=1$ and $h=0.5$ are in the second, third and fourth column. The combined mode test using $h=1$ and $h=0.5$ is in the fifth column.
}
\label{ta8}
\end{table}
Figure \ref{g254} illustrates the results of the twelve tests for the hypotheses \eqref{t5neu2}
conducted simultaneously. Each arrow
at a location $\bt$ in direction $\bv$ displays a rejection of a hypothesis
 \eqref{t5neu2} and the length of the arrows corresponds to the respective bandwidths.
 
\begin{figure}
\begin{center}
\includegraphics[width=0.4\textwidth]{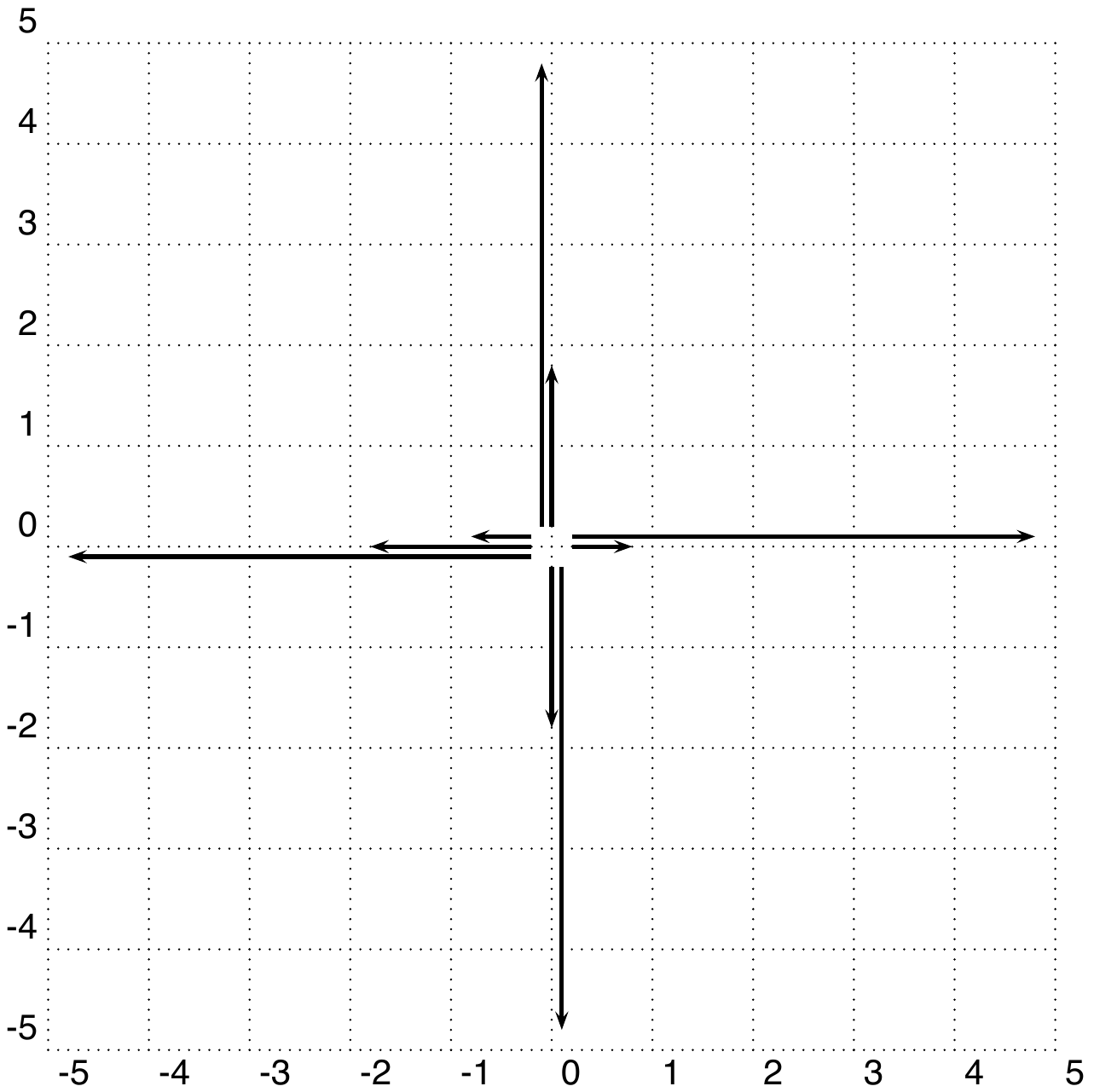}
\captionof{figure}{\it Rejected hypotheses of the twelve tests. Each arrow indicates a decrease on a scale that corresponds to the length of the arrow.} 
\label{g254}
\end{center}
\end{figure}
The results in Table \ref{ta8} show that the mode test for  $h=2.5$  detects in all cases a mode in a neighborhood of $(0,0)^\top$. However, the bandwidth $h=2.5$ is too large to distinguish between the underlying modes of $f_\bbeta$  at $(0,0)^\top$ and $(2,0)^\top$. The test for the bandwidth $h_0=1$ fails to detect the mode as $h=1$ is still too large to separate the two modes of $f_\bbeta.$
On the contrary, $h_0=0.5$ is too small to detect the decrease with small slope corresponding to the eigenvalue $0.4$ of 
the covariance matrix in the first mixture component of $f_\bbeta.$ Note that this effect vanishes for increasing sample size. By conducting the tests for the bandwidths $h=1$ and $h=0.5$ simultaneously, we are able to detect 
the mode at $(0,0)^\top$ in most of the simulations. 

\subsection{Comparison with the parametric random coefficient model}

It is common in the literature to use a parametric specification of the random coefficients. Usually, $\bbeta$ is assumed to have a multivariate normal distribution. In this section, we compare our nonparametric approach with mode estimation under a parametric specification of the random coefficient distribution. Our findings are similar to what is usually observed when parametric and nonparametric methods are compared. If the parametric specification of the model is consistent with the data generating process, the parametric approach outperforms the nonparametric one in terms of estimation errors, power of tests, and computation time. However, when the parametric specification and the data generating process differ, a model misspecification bias is present in the parametric estimation. This will not vanish, even if the sample size is large. The nonparametric model does not suffer from this problem and can in these cases perform better than the parametric model. We illustrate this effect in the following simulations.

We consider the parametric random coefficients model $Y=\tilde\bbeta^\top \bX$ where the density of $\tilde\bbeta$ belongs to some parametric family $f_{\tilde\bbeta} (\bb;\eta)$ with parameter $\eta$. 
Note that we can rewrite the equation as a mixed model $Y = \bgamma^\top \bX + \bu^\top \bX$. Here $\bgamma = \E(\tilde\bbeta)$ are fixed effects and $\bu = \tilde\bbeta - \E(\tilde\bbeta)$ are random effects with expectation $\bs 0$. If $f_{\tilde\bbeta} (\bb;\eta)$ is unimodal with a mode at $\E(\tilde\bbeta)$, we just need to derive confidence statements for $\bgamma$. This is in particular true if $f_{\tilde\bbeta} (\bb;\eta)$ is a normal density, which is the most common choice. A simple and efficient way to estimate $\bgamma$ is heteroscedasticity robust linear regression. Obviously, the computational complexity of this method is much smaller than the complexity of our algorithm.

Our test example is the random coefficients model with intercept \eqref{eq:model_with_int} with $d=3$ and $(X_{i,2},X_{i,3})^\top$ sampled from a standard bivariate Cauchy distribution. The distribution of the random coefficients $\bbeta_i$ in the data generating process is given by \textit{(i)} a standard Gaussian, \textit{(ii)} the normal mixture $0.5\ca N(\bs 0,0.1I)+0.5\ca{N}((2,0,0)^\top,0.1I)$ and \textit{(iii)} an exponential-2 distribution for $\beta_{i,1}$ and $(\beta_{i,2},\beta_{i,3})^\top\sim\ca N(\bs 0,0.1I)$ independent of the first component. Table \ref{OLS} reports estimates for $\bgamma$ obtained by transforming $Y$ and $\bX$ to $S$ and $\bTheta$ and running a heteroscedasticity robust regression.

\begin{table}[ht]
\begin{center}
\begin{tabular}{ccc|cc|cc}
&\multicolumn{2}{c|}{\textit{(i)} ($n=1000$)}&\multicolumn{2}{c|}{\textit{(ii)} ($n=1000$)}&\multicolumn{2}{c}{\textit{(iii)} ($n=1000$)}\\
\hline
$\widehat\gamma_1$&0.00&(0.02) &1.00&(0.02)&2.03&(0.03)\\
$\widehat\gamma_2$&0.02&(0.02)&0.00&(0.01)&-0.02&(0.02)\\
$\widehat\gamma_3$&0.00&(0.02)&0.01&(0.01)&0.02&(0.02)
\end{tabular}
\caption{\textit{Results of OLS for model $Y = \bgamma^\top \bX + \bu^\top \bX$ 
for $d=3$ and $(X_{i,2},X_{i,3})^\top$ drawn from a standard bivariate Cauchy distribution. The distribution of the true random coefficients $\bbeta_i$ is   a standard Gaussian in column \textit{(i)}, the normal mixture $0.5\ca N(\bs 0,0.1I)+0.5\ca{N}((2,0,0)^\top,0.1I)$ in column \textit{(ii)} and  an exponential-2 distribution for $\beta_{i,1}$ and $(\beta_{i,2},\beta_{i,3})^\top\sim\ca N(\bs 0,0.1I)$ independent of the first component in column \textit{(iii)}.\\[-20pt]}}
\label{OLS}
\end{center}
\end{table}
 
In column \textit{(i)}  of Table \ref{OLS} the parametric assumption holds and the procedure detects the true mode $\bs 0$ of the density  with high precision compared to the bandwidth choice $h=1$ in the simulations presented in Table \ref{ta3} in Section \ref{loctest}. In contrast, for the bimodal density in column \textit{(ii)} 
the coefficient vector $(1,0,0)^\top$ does not describe a representative member of the population because the misspecification bias is too large. For the skewed distribution of column \textit{(iii)} the estimator also fails to detect the mode of the density for the same reason. 

By applying our testing procedure in the setting of columns \textit{(ii)} and \textit{(iii)} we can show that the OLS results do not represent modes of the density. To this end, we set in \textit{(ii)} 
\[\bt=(0.5,0,0)^\top, \quad h=0.5, \quad \bv =(1,0,0)^\top\]
and in \textit{(iii)}
\[\bt=(1,0,0)^\top, \quad h=1, \quad \bv =(1,0,0)^\top.\]
Both tests reject $H_{0,-}^{\bt,h,\bv}$ (in \textit{(ii)} in $100\%$ and in \textit{(iii)} in $95.5\%$ percent of 1000 repetitions). Therefore, our procedure shows that neither  $(1,0,0)^\top$ in \textit{(ii)} nor $(2,0,0)^\top$ in \textit{(iii)} are modes of the underlying density. Of course, the parametric model would be able to detect the modes in cases \textit{(ii)} and \textit{(iii)} with a different parametric specification. However, this would require considerable a priori knowledge about the data. If we did not interpret the results as estimators for the mode but as estimators for $\E(\bbeta)$, the parametric method would perform well.

\section{Application to consumer demand data}
\label{sec:real_data}
Heterogeneity of consumers is a major challenge in modeling and estimating consumer demand. In several different demand models random coefficients were proposed to account for the heterogeneity in the population of consumers. 

\subsection{Model and data}

In this section we are interested in the almost ideal demand system (AIDS) which was initially proposed by \cite{Deaton:80} with fixed coefficients. This model does not explain demand for a product itself but explains the budget share spent on a product by a linear equation. The explanatory variables are log prices and the log of total expenditure divided by a price index. A detailed discussion of the model is contained in \cite{Lewbel:97}. 

Fixed coefficients in this model mean that all consumers are assumed to react in the same way when the price of a product changes. It is well known that some consumers are very price sensitive and change their behavior significantly with small variations in prices while other consumers are less price sensitive. This type of heterogeneity can be modeled by a random coefficient on log prices which is assumed to vary across the population of consumers. A similar argument suggests a random coefficient on log total expenditure. Recently, applications of the AIDS using a nonparametric random coefficient specification instead of fixed coefficients were presented in \cite{HHM:15} and \cite{Breunig16}.

We apply our multiscale test to detect modes in the random coefficients for budget shares for food at home (BSF)
\begin{align}\label{eq:bsf_model}
BSF_i = \beta_{i,1} + \beta_{i,2} \ln(TotExp_i) + \beta_{i,3} \ln(FoodPrice_i).
\end{align}
Food expenditure is a large fraction of total expenditure and is roughly about 20\%.


We analyze the data of the British Family Expenditure Survey which ran from 1961 to 2001. It reported yearly cross sections for household income, expenditure and other characteristics of roughly 7000 households. We use data of the years 1997--2001 only which gives a sample size of about 33000. 
Budget shares of food are generated by dividing the expenditure for all food by total expenditure. Food prices are reported as relative prices in comparison to a general prize index. The variable $TotExp$ is normalized to January 2000 real prizes.

Assumption \ref{p1} and the numerical simulations in Section \ref{Sec4} suggest that our test has more power when the normalized regressors are approximately uniform on the sphere. We can achieve this by symmetrizing the design in model \eqref{eq:bsf_model} as follows:
\begin{align}\label{eq:bsf_model_mod}
BSF_i = \tilde \beta_{i,1} + \tilde \beta_{i,2} \left(\ln(TotExp_i)-5\right) + \tilde \beta_{i,3} \left(25\ln(FoodPrice_i) - 0.3\right).
\end{align}
The relation of the modified model to the random coefficients in \eqref{eq:bsf_model} is $\beta_{i,1} = \tilde \beta_{i,1} - 5 \tilde \beta_{i,2} - 0.3  \tilde \beta_{i,3}$, $\beta_{i,2} = \tilde \beta_{i,2}$, $\beta_{i,3} = 25 \tilde \beta_{i,3}$.
Observations of the new variable $\ln(TotExp_i)-5$ lie between $-5$ and $3.7$. The observations  of $25\ln(FoodPrice_i) - 0.3$ range from $-1$ to $1.3$.

\subsection{Results}
For a first evaluation of the data we assumed fixed coefficients in model \eqref{eq:bsf_model_mod} and estimated the model with ordinary least squares (OLS).
\begin{table}[ht]
\begin{center}
\begin{tabular}{c|c|c}
$\tilde \beta_{1}$ & $\tilde \beta_{2}$ & $\tilde \beta_{3}$\\
\hline
$0.1940$ & $-0.0587$ & $0.0005$\\
$(0.000)$ & $(0.001)$ & $(0.001)$
\end{tabular}
\caption{\textit{Results of OLS for model \eqref{eq:bsf_model_mod} with fixed coefficients.\\[-20pt]}}
\label{tab:bsf_ols}
\end{center}
\end{table}

In order to find modes of the density we conducted simultaneously tests on the $5\%$ level of the form \eqref{t5neu2} on the two scales $h_1=0.75$ and $h_2=0.5$. Recall from Section \ref{loctest} that our testing procedure also performs well when the $\bX_i$ are not Cauchy distributed. Therefore, to obtain a testing procedure which is more flexible with respect to the design, we use the nonparametric density estimator $\widehat{f}_\bTheta$ instead of a parametric estimation procedure. We were testing for modes on the equidistant grid covering $[-1,1]^3$ with grid  width $1$. Hence, the grid had 27 nodes. For every grid point $\bb\in\R^3$ tests of the hypotheses \eqref{t5neu2} were conducted  for the
directions and locations
\begin{align*}
\begin{array}[t]{lllll}
 \bt_1=\bb+h_j\mathbf{e}_1,\;&\bv_1=\mathbf{e}_1,&\quad &\bt_4=\bb-h_j\mathbf{e}_2,\;&\bv_4=-\mathbf{e}_2,\\
\bt_2=\bb-h_j\mathbf{e}_1,\;&\bv_2=-\mathbf{e}_1,&\quad & \bt_5=\bb+h_j\mathbf{e}_3,\;&\bv_5=\mathbf{e}_3,\quad\quad\quad(j=1,2)\\
\bt_3=\bb+h_j\mathbf{e}_2,\;&\bv_3=\mathbf{e}_2,&\quad &\bt_6=\bb-h_j\mathbf{e}_3,\;&\bv_6=-\mathbf{e}_3,
\end{array}
\end{align*} 
where $\mathbf{e}_1,\mathbf{e}_2,\mathbf{e}_3 \in \mathbb{R}^3$ denote the standard unit vectors of $\R^3$. We detected a single mode in the neighborhood of the grid point $(0,0,0)^\top$ for the tests with bandwidth $h_1=0.75$. The test for the bandwidth $h_2=0.5$ did not detect a mode.

In the following we use nonparametric density estimation to motivate hypotheses for the testing procedure. It is important that this estimate and the test are independent, otherwise the testing procedure would be biased and could not guarantee a bound on the error rate. We meet the requirement by splitting the sample in two independent equally sized sub-samples. The first sub-sample is used for nonparametric estimation of the random coefficient density in model \eqref{eq:bsf_model_mod} with the estimator in \cite{hoderlein2008}. Figure \ref{Fig:2} gives contour plots for the joint densities of $f_{\tilde\beta_1,\tilde\beta_2}$, $f_{\tilde\beta_1,\tilde\beta_3}$, and $f_{\tilde\beta_2,\tilde\beta_3}$ based on about 16500 observations. We chose the smoothing parameters $h$ and $g$ in the estimator in \cite{hoderlein2008} equal to 0.05 and 0.1, respectively. Note that these bandwidth choices do not affect the level of the test performed below as the test does not depend on this estimator. The nonparametric estimate suggests that the random coefficient density of $f_{\tilde\beta_1,\tilde\beta_2,\tilde\beta_3}$ has one (well-pronounced) mode close to 
\begin{align}\label{eq:bsf_location}
(\beta_1,\beta_2,\beta_3) = (0.25,-0.07,0.02).
\end{align}
This is consistent with the results of the test above which found a mode close to $(0,0,0)$. Since the marginal densities of $\beta_1,\beta_2,\beta_3$ are nearly symmetric it is also consistent that the mode is close to the OLS estimates given in Table \ref{tab:bsf_ols}. With a significantly skewed or with a multimodal random coefficient density location of modes would differ from OLS.
\begin{figure}[ht]
\begin{center}
\includegraphics[width=0.32\textwidth]{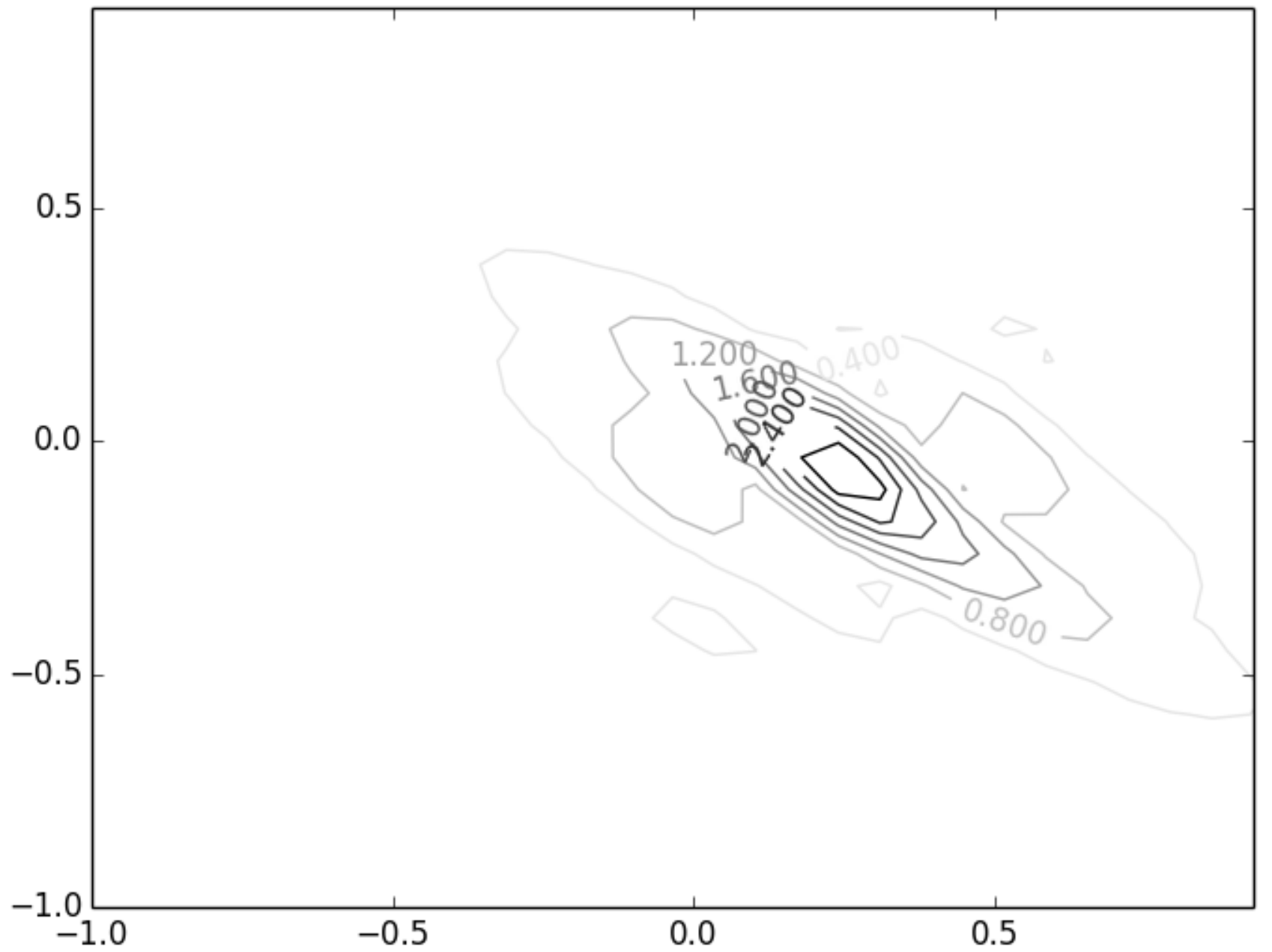}\;
\includegraphics[width=0.32\textwidth]{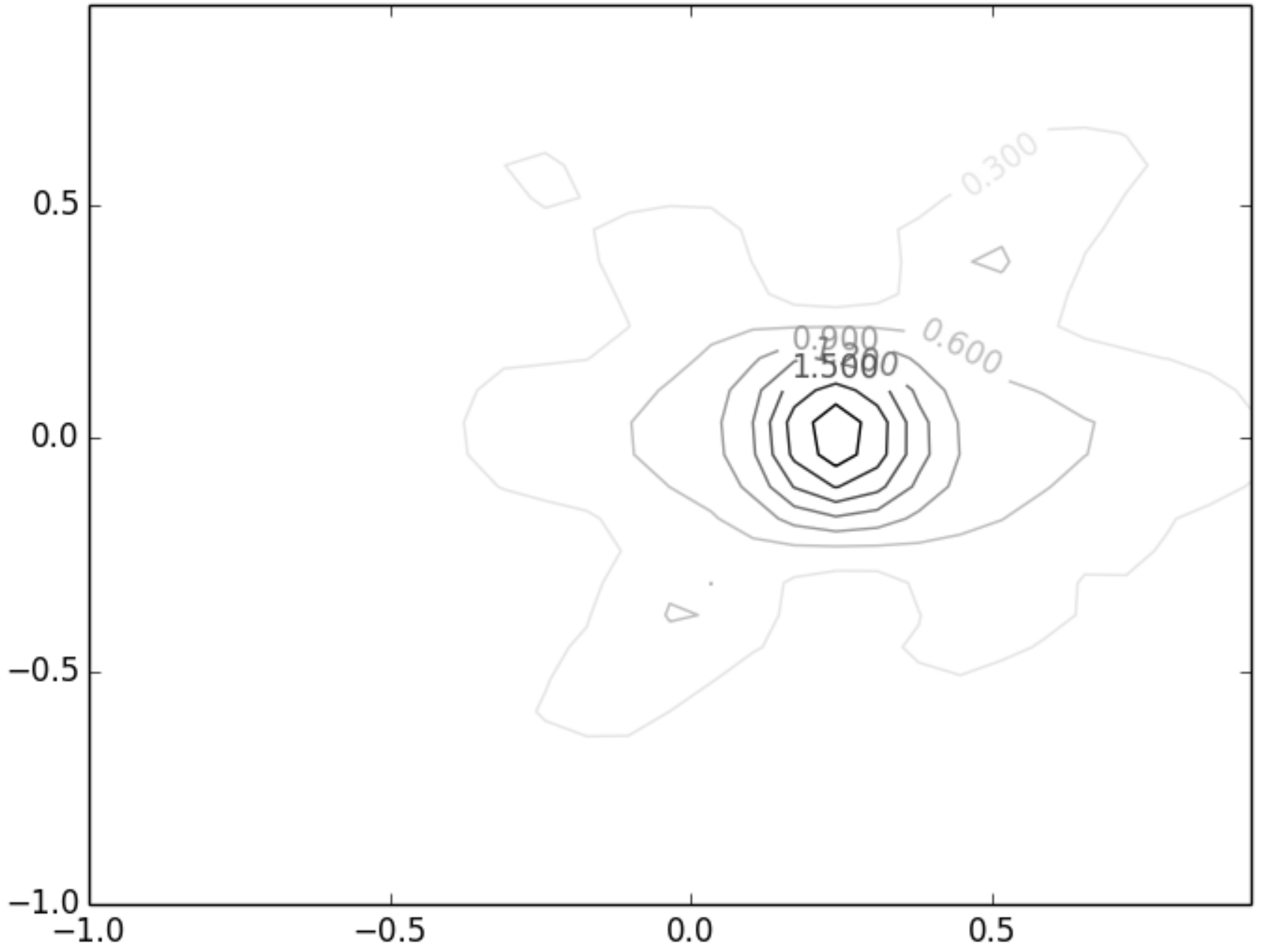}\;
\includegraphics[width=0.32\textwidth]{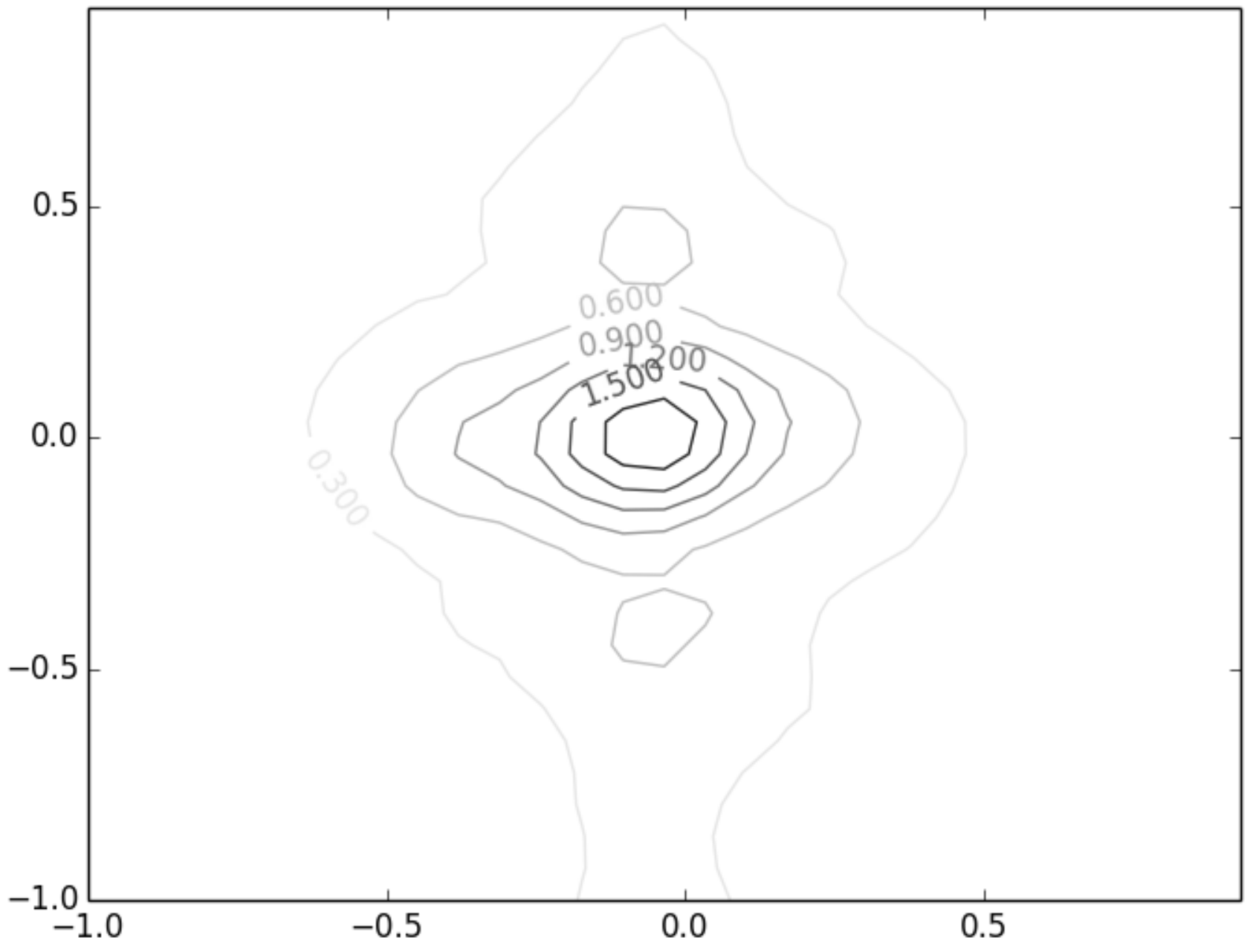}
\caption{\it Nonparametric estimates of the joint densities of $f_{\tilde\beta_1,\tilde\beta_2}$ (left), $f_{\tilde\beta_1,\tilde\beta_3}$ (middle), and $f_{\tilde\beta_2,\tilde\beta_3}$ (right).\\[-40pt]}
\label{Fig:2}
\end{center}
\end{figure}

With the second sub-sample we studied whether a mode can be found for the location given by \eqref{eq:bsf_location} on a smaller scale than in the test above. This would then indicate that the location of the mode is not $(0,0,0)^\top$. Our primary interest is in the coefficients on $\tilde\beta_2$ and $\tilde\beta_3$ on total expenditure and food prices. In order to see if the mode is indeed in a location where $\tilde\beta_2<0$ and $\tilde\beta_3>0$, we conduct two mode tests \eqref{eq:bsf_location} simultaneously for the bandwidths $h_1=0.07$ and $h_2=0.02$, that is, we test twelve hypotheses \eqref{t5neu2} for the following  locations and directions $(\bt_i, \bv_i).$ On scale $h_1=0.07$, we tested
\begin{align*} 
\begin{array}[t]{ll}
\bt_1=(0.32,-0.07,0.02)^\top,\;\bv_1=(1,0,0)^\top,\quad &\bt_4=(0.25,-0.14,0.02)^\top,\;\bv_4=-\bv_3,\\
\bt_2=(0.18,-0.07,0.02)^\top,\;\bv_2=-\bv_1,\quad &\bt_5=(0.25,-0.07,0.09)^\top,\;\bv_5=(0,0,1)^\top,\\
\bt_3=(0.25,0,0.02)^\top,\;\bv_3=(0,1,0)^\top,\quad &\bt_6=(0.25,-0.07,-0.05)^\top,\;\bv_6=-\bv_5\\
\end{array}
\end{align*} 
and on scale $h_2=0.02,$
\begin{align*} 
\begin{array}[t]{ll}
\bt_7=(0.27,-0.07,0.02)^\top,\;\bv_7=\bv_1,\quad &\bt_{10}=(0.25,-0.09,0.02)^\top,\;\bv_{10}=-\bv_3,\\
\bt_8=(0.23,-0.07,0.02)^\top,\;\bv_8=-\bv_1,\quad &\bt_{11}=(0.25,-0.07,0.04)^\top,\;\bv_{11}=\bv_5,\\
\bt_9=(0.25,-0.05,0.02)^\top,\;\bv_9=\bv_3,\quad &\bt_{12}=(0.25,-0.07,0)^\top,\;\bv_{12}=-\bv_5.\\
\end{array}
\end{align*} 
The test rejected all local hypotheses on scale $h_1=0.07$ but not on scale $h_2=0.02.$ This gives evidence that the mode is in a location where $\tilde\beta_2$ is negative but we cannot decide whether $\tilde\beta_3$ is positive at the mode.

Let us return to the initial model \eqref{eq:bsf_model}. The results of our test give evidence that a mode exists close to
\[
(\beta_1,\beta_2,\beta_3) = (0.5,-0.07,0.5) 
\] 
with strong evidence that $\beta_2$ is indeed negative. This vector of coefficients describes a representative member of the majority of consumers. It suggests that in the majority group food budget shares decrease with increasing log total expenditure. The nonparametric estimate in Figure \ref{Fig:2} shows that there is considerable variance among consumers around this representative member.

\section*{Acknowledgments} J. Schmidt-Hieber was partially funded by a TOP II grant from the Dutch science foundation.
F. Dunker acknowledges support by the Ministry of Education and Cultural Affairs of Lower Saxony in the project Reducing
Poverty Risk. K. Proksch acknowledges financial support by the German Research Foundation DFG through subproject A07 of CRC
755. K. Eckle has been supported by the Collaborative Research Center ``Statistical modeling of nonlinear dynamic processes''
(SFB 823, Project C4) of the German Research Foundation (DFG). We are very grateful to two referees and an associate editor for their constructive comments, which led to substantial improvement of an earlier version of this manuscript.

\bibliographystyle{apalike}
\bibliography{bibRC}

\begin{thebibliography}{}

\bibitem[Adler and Taylor, 2007]{Adler2007}
Adler, R.~J. and Taylor, J.~E. (2007).
\newblock {\em Random fields and geometry}.
\newblock Springer Monographs in Mathematics. Springer, New York.

\bibitem[Andrews, 2001]{andrews2001}
Andrews, D. (2001).
\newblock Testing when a parameter is on the boundary of the maintained
  hypothesis.
\newblock {\em Econometrica}, 69(3):683--734.

\bibitem[Bai et~al., 1988]{MR971170}
Bai, Z.~D., Rao, C.~R., and Zhao, L.~C. (1988).
\newblock Kernel estimators of density function of directional data.
\newblock {\em J. Multivariate Anal.}, 27(1):24--39.

\bibitem[Bates et~al., 2015]{sdsdsd}
Bates, D., M{\"a}chler, M., Bolker, B., and Walker, S. (2015).
\newblock Fitting linear mixed-effects models using {lme4}.
\newblock {\em Journal of Statistical Software}, 67(1):1--48.

\bibitem[Beran, 1993]{Beran1993}
Beran, R. (1993).
\newblock Semiparametric random coefficient regression models.
\newblock {\em Ann. Inst. Statist. Math.}, 45(4):639--654.

\bibitem[Beran et~al., 1996]{Beran1996}
Beran, R., Feuerverger, A., and Hall, P. (1996).
\newblock On nonparametric estimation of intercept and slope distributions in
  random coefficient regression.
\newblock {\em Ann. Statist.}, 24(6):2569--2592.

\bibitem[Beran and Hall, 1992]{Beran1992}
Beran, R. and Hall, P. (1992).
\newblock Estimating coefficient distributions in random coefficient
  regressions.
\newblock {\em Ann. Statist.}, 20(4):1970--1984.

\bibitem[Berry et~al., 1995]{BLP:95}
Berry, S., Levinsohn, J., and Pakes, A. (1995).
\newblock Automobile prices in market equilibrium.
\newblock {\em Econometrica}, 63(4):841--890.

\bibitem[Berry and Pakes, 2007]{Berry07}
Berry, S. and Pakes, A. (2007).
\newblock The pure characteristics demand model.
\newblock {\em Internat. Econom. Rev.}, 48(4):1193--1225.

\bibitem[Breunig and Hoderlein, 2018]{Breunig16}
Breunig, C. and Hoderlein, S. (2018).
\newblock {Specification Testing in Random Coefficient Models}.
\newblock {\em Quantitative Economics}, forthcoming.

\bibitem[Butucea et~al., 2007]{butucea2007}
Butucea, C., Guţă, M., and Artiles, L. (2007).
\newblock Minimax and adaptive estimation of the wigner function in quantum
  homodyne tomography with noisy data.
\newblock {\em Ann. Statist.}, 35(2):465--494.

\bibitem[Chernozhukov et~al., 2017]{Chernozhukov2017}
Chernozhukov, V., Chetverikov, D., and Kato, K. (2017).
\newblock Central limit theorems and bootstrap in high dimensions.
\newblock {\em Ann. Probab.}, 45(4):2309--2352.

\bibitem[Davison, 1983]{Davison1983}
Davison, M.~E. (1983).
\newblock The ill-conditioned nature of the limited angle tomography problem.
\newblock {\em SIAM J. Appl. Math.}, 43(2):428--448.

\bibitem[Deaton and Muellbauer, 1980]{Deaton:80}
Deaton, A. and Muellbauer, J. (1980).
\newblock An almost ideal demand system.
\newblock {\em American Economic Review}, 70:312--326.

\bibitem[Dubé et~al., 2012]{Dube12}
Dubé, J.-P., Fox, J.~T., and Su, C.-L. (2012).
\newblock Improving the numerical performance of static and dynamic aggregate
  discrete choice random coefficients demand estimation.
\newblock {\em Econometrica}, 80(5):2231--2267.

\bibitem[{D\"umbgen} and Spokoiny, 2001]{duembgen2001}
{D\"umbgen}, L. and Spokoiny, V.~G. (2001).
\newblock Multiscale testing of qualitative hypotheses.
\newblock {\em Ann. Statist.}, 29(1):124--152.

\bibitem[D{\"u}mbgen and Walther, 2008]{duembgen2008}
D{\"u}mbgen, L. and Walther, G. (2008).
\newblock Multiscale inference about a density.
\newblock {\em Ann. Statist.}, 36(4):1758--1785.

\bibitem[Dunker et~al., 2013]{DHK:13}
Dunker, F., Hoderlein, S., and Kaido, H. (2013).
\newblock Random coefficients in static games of complete information.
\newblock {\em cemmap Working Papers}, CWP12/13.

\bibitem[Dunker et~al., 2017]{DHK:17}
Dunker, F., Hoderlein, S., and Kaido, H. (2017).
\newblock Nonparametric identification of random coefficients in endogenous and
  heterogeneous aggregate demand models.
\newblock {\em cemmap Working Papers}, CWP11/17.

\bibitem[Dunker et~al., 2018]{DHKS:18}
Dunker, F., Hoderlein, S., Kaido, H., and Sherman, R. (2018).
\newblock Nonparametric identification of the distribution of random
  coefficients in binary response static games of complete information.
\newblock {\em Journal of Econometrics}, forthcoming.

\bibitem[Eckle et~al., 2017a]{Eckle}
Eckle, K., Bissantz, N., and Dette, H. (2017a).
\newblock Multiscale inference for multivariate deconvolution.
\newblock {\em Electron. J. Stat.}, 11(2):4179--4219.

\bibitem[Eckle et~al., 2017b]{Eckle2017}
Eckle, K., Bissantz, N., Dette, H., Proksch, K., and Einecke, S. (2017b).
\newblock Multiscale inference for a multivariate density with applications to
  x-ray astronomy.
\newblock {\em Annals of the Institute of Statistical Mathematics},
  https://doi.org/10.1007/s10463-017-0605-1.

\bibitem[Feuerverger and Vardi, 2000]{Feuerverger2000}
Feuerverger, A. and Vardi, Y. (2000).
\newblock Positron emission tomography and random coefficients regression.
\newblock {\em Ann. Inst. Statist. Math.}, 52(1):123--138.

\bibitem[Fox and Gandhi, 2016]{Fox16}
Fox, J.~T. and Gandhi, A. (2016).
\newblock Nonparametric identification and estimation of random coefficients in
  multinomial choice models.
\newblock {\em The RAND Journal of Economics}, 47(1):118--139.

\bibitem[Frikel, 2013]{Frikel2013}
Frikel, J. (2013).
\newblock Sparse regularization in limited angle tomography.
\newblock {\em Appl. Comput. Harmon. Anal.}, 34(1):117--141.

\bibitem[Gautier and Hoderlein, 2012]{Gautier12}
Gautier, E. and Hoderlein, S. (2012).
\newblock {A triangular treatment effect model with random coefficients in the
  selection equation}.
\newblock {\em cemmap Working Papers}, CWP39/12.

\bibitem[Gautier and Kitamura, 2013]{Gautier13}
Gautier, E. and Kitamura, Y. (2013).
\newblock Nonparametric estimation in random coefficients binary choice models.
\newblock {\em Econometrica}, 81(2):581--607.

\bibitem[Gin{\'e} and Guillou, 2001]{MR1876841}
Gin{\'e}, E. and Guillou, A. (2001).
\newblock On consistency of kernel density estimators for randomly censored
  data: rates holding uniformly over adaptive intervals.
\newblock {\em Ann. Inst. H. Poincar\'e Probab. Statist.}, 37(4):503--522.

\bibitem[Greenland, 2000]{Greenland00}
Greenland, S. (2000).
\newblock When should epidemiologic regressions use random coefficients?
\newblock {\em Biometrics}, 56(3):915--921.

\bibitem[Gustafson and Greenland, 2006]{Greenland06}
Gustafson, P. and Greenland, S. (2006).
\newblock The performance of random coefficient regression in accounting for
  residual confounding.
\newblock {\em Biometrics}, 62(3):760--768.

\bibitem[Helgason, 2011]{Helgason2011}
Helgason, S. (2011).
\newblock {\em Integral geometry and {R}adon transforms}.
\newblock Springer, New York.

\bibitem[Hoderlein et~al., 2015]{HHM:15}
Hoderlein, S., Holzmann, H., and Meister, A. (2015).
\newblock The triangular model with random coefficients.
\newblock {\em cemmap Working Papers}, CWP33/15.

\bibitem[Hoderlein et~al., 2010]{hoderlein2008}
Hoderlein, S., Klemel{\"a}, J., and Mammen, E. (2010).
\newblock Analyzing the random coefficient model nonparametrically.
\newblock {\em Econometric Theory}, 26(3):804--837.

\bibitem[Hohmann and Holzmann, 2016]{Hohmann2016}
Hohmann, D. and Holzmann, H. (2016).
\newblock Weighted angle radon transform: Convergence rates and efficient
  estimation.
\newblock {\em Statistica Sinica.}, 26:157--175.

\bibitem[Hsiao, 2014]{Hsiao14}
Hsiao, C. (2014).
\newblock {\em Analysis of Panel Data}.
\newblock Cambridge University Press.
\newblock Cambridge Books Online.

\bibitem[Hsiao and Pesaran, 2004]{Hasio04}
Hsiao, C. and Pesaran, M.~H. (2004).
\newblock {Random Coefficient Panel Data Models}.
\newblock CESifo Working Paper Series 1233, CESifo Group Munich.

\bibitem[Ichimura and Thompson, 1998]{Ichimura98}
Ichimura, H. and Thompson, T. (1998).
\newblock Maximum likelihood estimation of a binary choice model with random
  coefficients of unknown distribution.
\newblock {\em Journal of Econometrics}, 86(2):269 -- 295.

\bibitem[Lewbel, 1997]{Lewbel:97}
Lewbel, A. (1997).
\newblock Consumer demand systems and household equivalence scales.
\newblock In Pesaran, M.~H. and Schmidt, P., editors, {\em Handbook of applied
  econometrics}, volume~2, chapter~4, pages 167--201. Blackwell, Oxford.

\bibitem[Masten and Torgovitsky, 2014]{Masten14}
Masten, M. and Torgovitsky, A. (2014).
\newblock Instrumental variables estimation of a generalized correlated random
  coefficients model.
\newblock {\em cemmap Working Papers}, CWP02/14.

\bibitem[Masten, 2017]{Masten15}
Masten, M.~A. (2017).
\newblock Random coefficients on endogenous variables in simultaneous equations
  models.
\newblock {\em The Review of Economic Studies}, page rdx047.

\bibitem[Nevo, 2001]{Nevo:01}
Nevo, A. (2001).
\newblock Measuring market power in the ready-to-eat cereal industry.
\newblock {\em Econometrica}, 69(2):307--342.

\bibitem[Petrin, 2002]{Petrin:02}
Petrin, A. (2002).
\newblock Quantifying the benefits of new products: The case of the minivan.
\newblock {\em Journal of Political Economy}, 110(4):705--729.

\bibitem[{Proksch} et~al., 2016]{proksch16}
{Proksch}, K., {Werner}, F., and {Munk}, A. (2016).
\newblock {Multiscale scanning in inverse problems}.
\newblock {\em ArXiv Preprint}, arXiv:1611.04537.

\bibitem[Schmidt-Hieber et~al., 2013]{schmidthieber13}
Schmidt-Hieber, J., Munk, A., and D{\"u}mbgen, L. (2013).
\newblock Multiscale methods for shape constraints in deconvolution: confidence
  statements for qualitative features.
\newblock {\em Ann. Statist.}, 41(3):1299--1328.

\bibitem[Swamy, 1970]{swamy1970}
Swamy, P. (1970).
\newblock Efficient inference in a random coefficient regression model.
\newblock {\em Econometrica}, 38(2):311--323.

\end{thebibliography}

\newpage
\setcounter{page}{1}
\appendix

\section*{Supplementary material to ``Tests for qualitative features in the random coefficients model''}

\section{Nonparametric estimators for the densities $f_\mathbf{\Theta}$ and $f_{S,\mathbf{\Theta}}$}\label{3.1}
In this section we discuss the estimation of the densities $f_\bTheta$ and $f_{S,\bTheta}$ and related quantities. We use kernel density estimators based on the second half of the observations
$(S_i,\bTheta_i),$ $i=n+1,\hdots,2n$. The density of the random vector  $\bX_1$ is denoted by $ f_\bX$. In the random coefficients model without intercept $f_\bX$ is a $d$-variate density and a $(d-1)$-variate density in the random coefficients model with intercept. Throughout the following, $K:\R\rightarrow\R$ is assumed to be Lipschitz continuous, non-negative and $\int K=1.$

In the random coefficients model without intercept, we introduce the kernel density estimator
\beq{p3}
 \widehat f_\bTheta(\btheta)=\frac{C(h_*)}{nh_*^{d-1}}\sum_{i=n+1}^{2n}K\Big
 (\frac{1-\langle \bTheta_i,\btheta\rangle}{h_*^2}\Big),\quad h_*>0,
 \eeq
with normalization constant
\[
 C(h_*) :=h_*^{d-1}\Big(\int_{\sd}K\Big(\frac{1-\langle \btheta',\btheta\rangle}{h_*^2}\Big)d\btheta'\Big)^{-1}.
\]
As shown in \cite{MR971170}, the integral does not depend on $\btheta$ and $C(h_*)$ converges to some positive constant
 as $h_*\rightarrow0$. For the joint density of $(S, \bTheta),$ we propose the kernel density estimator
\beq{1.12}
 \widehat f_{S,\bTheta}(s,\btheta)=\frac{1}{nh_+^d} C(h_+)\sum_{i=n+1}^{2n}K\Big(\frac{1-\langle\bTheta_i,
 \btheta\rangle}{h_+^2}\Big)K\Big(\frac{S_i-s}{h_+}\Big),\quad h_+>0.
\eeq

In the random coefficients model with intercept the symmetrizations $S_i = \zeta_i Y_i/\|\mathbf{X}_i\|$ and $\bTheta_i = \zeta_i \mathbf{X}_i/\|\mathbf{X}_i\|$ with Rademacher variables $\zeta_i$ correspond to   point reflections of the densities at the origin. Thus, the density $f_\bTheta$ is in general not continuous on the boundary of the hemisphere $\sd_+$ (see also \eqref{eq.ftheta_explicit}). Smoothness is, however, necessary to control the bias. Therefore, we use a two step procedure for the estimation of $f_\bTheta$ and $f_{S,\bTheta}$ in the random coefficients model with intercept. First, we estimate the density of the non-symmetrized samples $Y_i/\|\mathbf{X}_i\|$ and $\bX_i/\|\bX_i\|$, $i=n+1,\hdots,2n$ on the hemisphere $\sd_+$ and on $\R\times\sd_+$, respectively. The  estimators for $f_\bTheta$ and  $f_{S,\bTheta}$ are then the same as \eqref{p3} and \eqref{1.12} except that now $\btheta\in\sd_+$ and the normalization constant $C(h_*)$ is replaced by a function $\btheta\mapsto C(h_*,\btheta)$, defined by
\[
  C(h_*,\btheta) :=h_*^{d-1}\Big(\int_{\sd_+}K\Big(\frac{1-\langle \btheta',\btheta\rangle}{h_*^2}\Big)d\btheta'\Big)^{-1}\quad\big(\btheta\in\sd_+\big).
\] In a second step, we symmetrize the estimators and divide by two to get estimators of the densities $f_\bTheta$ and $f_{S,\bTheta}$ on the whole domain. 

In the next lemma we establish convergence of these estimators. If in Lemma \ref{g22} $f_\bX$ follows a multivariate Cauchy distribution, then $\gamma :=1$. Otherwise, $\gamma$ comes from Assumption \ref{p1} \textit{(iii)}.

\begin{lem2}\label{g22}
Suppose Assumption \ref{p1} is satisfied for some $\gamma>0$  and set $\gamma=1$ in the case of the random coefficients model with intercept. In both models, the estimator $\widehat f_\bTheta$ with bandwidth $h_*=O(\log(n)^{-3/\gamma})$
and  $h_*\geq\log(n)^{7/(d-1)}n^{-1/(d-1)}$ satisfies

\begin{tabular}{ll}
(i) & \  $\sup_{\btheta \in \sd} \big |\E\big[  \widehat  f_\bTheta(\btheta)\big] - f_\bTheta(\btheta) \big| 
	= O\big( h_*^\gamma \big) ;$ \\[0.3cm] 
(ii) & \  $\sup_{\btheta \in \sd} \big | \widehat  f_\bTheta(\btheta) - \E\big[  \widehat  f_\bTheta(\btheta) \big] \big| 
	= O_\P\Big( \sqrt{\frac{\log(n)}{nh_*^{d-1}}}\Big);$ \\[0.3cm] 
(iii) & \   $\sup_{\btheta \in \sd} \big | \widehat  f_\bTheta(\btheta) -    f_\bTheta(\btheta)  \big| 
	= O\big(\log(n)^{-1}\big),\ \text{for }n\rightarrow\infty,\text{ almost surely}.$ \\
\end{tabular}
\end{lem2}

The proof is delayed to the end of this section. Let us now discuss properties of the density $f_{S,\bTheta}(s,\btheta).$ By \eqref{1.3}, $f_{S,\bTheta}(s,\btheta)=f_\bTheta(\btheta) Rf_\bbeta(s,\btheta).$ Under Assumption \ref{ass:f_beta}, $f_\bbeta$
is compactly supported and consequently, $f_{S,\bTheta}(s,\btheta)$ and $s\mapsto \log(|s|)^2f_{S,\bTheta}(s,\btheta)$ for $|s|\geq1$ are uniformly bounded. Moreover, $f_{S,\bTheta}$ is H\"older continuous with H\"older constant
$\gamma$. This is a straightforward consequence of the  H\"older $\gamma$-continuity of $f_\bTheta$  shown in the proof of Lemma  \ref{g22} and the identity
\[
 Rf_\bbeta(s,\btheta)=\int_{\R^{d-1}}f_\bbeta( s\btheta+x_1\btheta_1+\hdots+x_{d-1}\btheta_{d-1})d\bs x,
\]
where $\btheta_1,\hdots,\btheta_{d-1}$ denote an orthonormal basis of the orthogonal complement of $\rm{span}\{\btheta\}$, together with the 
compact support  and the Lipschitz-continuity of $f_\bbeta$ (following from Assumption \ref{ass:f_beta}). Moreover, the properties of $f_{S,\bTheta}$ discussed above also hold in quantum homodyne tomography under Assumption 2'. We point out that the marginal densities of the Wigner function, which are 
given by the Radon transform, are nonnegative. This will be used later to bound the standard deviation $\sigma_{\bt,h,\bv}$ away from zero.

If in Lemma \ref{p2} $f_\bX$ follows a multivariate Cauchy distribution, then $\gamma :=1$. Otherwise, $\gamma$ comes from Assumption  \ref{p1} \textit{(iii)}.
\begin{lem2}\label{p2}
Let Assumptions \ref{ass:f_beta} resp. 2' and \ref{p1} hold and consider the estimator $\widehat f_{S,\bTheta}$ in \eqref{1.12} with bandwidth choice $h_+=O(\log(n)^{-3/\gamma})$ and  $h_+\geq\log(n)^{3/d}n^{-1/(2d)}$. 
 Then
 \[
  \sup_{(s,\btheta)\in\R\times\sd}\big|\widehat f_{S,\bTheta}(s,\btheta)-f_{S,\bTheta}(s,\btheta)\big|=O(\log(n)^{-2})\quad\text{for }n\rightarrow\infty\text{ almost surely}.
 \]
\end{lem2}

For the estimation of the test statistic $\widehat{T}_{\bt,h,\bv}$ and the limiting process $\widehat{X}_{\bt,h,\bv}$ the 
quantities $1/f_\bTheta$ and $\sqrt{f_{S,\bTheta}}$ need to be estimated. The functions $(\cdot)^{-1}$ and $\sqrt{\cdot}$ are
not smooth in zero and we therefore introduce the cut-off estimators
\beq{1.5}
\widetilde f_\bTheta:=\widehat f_\bTheta\vee \log(n)^{-1}\quad\tn{and}\quad \widetilde f_{S,\bTheta}:=\widehat f_{S,\bTheta}\vee
\log(n)^{-2}.
\eeq
By the boundedness from below of $f_\bTheta$ and Lemma \ref{g22} it holds
$\widetilde f_\bTheta=\widehat f_\bTheta$ almost surely for $n$ sufficiently large.

\begin{proof}[Proof of Lemma \ref{g22}]
We only consider the case without intercept, that is, \textit{(i)-(iii)} in Assumption \ref{p1} hold. In the case with intercept, we use $C(h_*)\leq C(h_*,\btheta)\leq 2C(h_*)$ and the fact that $C(h_*,\btheta)$ is Lipschitz-continuous with respect to $\btheta$ to arrive at the same conclusion. 

To prove \textit{(i)} observe that
 \bal
 \big|\E\big[  \widehat  f_\bTheta(\btheta)\big] - f_\bTheta(\btheta)\big|
 =&
 \Big|\frac{C(h_*)}{h_*^{d-1}}  \int_{\sd} K\Big(\frac{1-\langle \btheta',\btheta\rangle}{h_*^2}\Big)f_\bTheta(\btheta')d\btheta'
-f_\bTheta(\btheta)\Big|\\
 \leq& \frac{C(h_*)}{h_*^{d-1}}  
 \int_{\|\btheta'-\btheta\|\lesssim h_*} K\Big(\frac{1-\langle \btheta',\btheta\rangle}{h_*^2}\Big)
 \big|f_\bTheta(\btheta')-f_\bTheta(\btheta)\big|d\btheta'.
 \end{align*}
 Here, we used the compact support of $K$ and  the identity $1-\langle\btheta',\btheta\rangle=\|\btheta'-\btheta\|^2/2$. Since $f_\bTheta(\btheta)=\int_{0}^\infty r^{d-1}f_\bX(r\btheta)dr$, we have by Assumption \ref{p1} \textit{(iii)} $|f_\bTheta(\btheta')-f_\bTheta(\btheta)|\leq \int_0^\infty r^{d-1}|f_\bX(r\btheta')-f_\bX(r\btheta)|dr
 \lesssim h_*^\gamma$ for $\|\btheta'-\btheta\|\lesssim h_*$. By definition of the constant $C(h_*)$,
we obtain $\big|\E\big[  \widehat  f_\bTheta(\btheta)\big] - f_\bTheta(\btheta)\big|\lesssim
 h_*^\gamma $ and this proves \textit{(i)}.

Next, we bound the stochastic error term \textit{(ii)} using an entropy argument and Bernstein's inequality. Observe that by the Lipschitz-continuity of $K,$ $\widehat f_\bTheta-\E[\widehat f_\bTheta]$ is Lipschitz-continuous with Lipschitz constant of order $h_*^{-d-1}.$ For 
$a_n :=\sqrt{\frac{\log(n)}{nh_*^{d-1}}},$ let $\{\btheta_j: j=1,\hdots,M\}$ be defined as the set of smallest cardinality such that
$\bigcup_{j=1}^MB_{c'h_*^{d+1}a_n}(\btheta_j)\supset \sd$ for some constant $c'>0.$ If $c'>0$ is chosen small enough, then
\begin{align}
\P\Big(\sup_{\btheta \in \sd} \big | \widehat  f_\bTheta(\btheta) - \E\big[  \widehat  f_\bTheta(\btheta) \big] \big|>ca_n\Big)
\leq  \sum_{j=1}^M \P\Big( \big | \widehat  f_\bTheta(\btheta_j) - \E\big[  \widehat  f_\bTheta(\btheta_j) \big] \big|>\frac{c} 2 a_n\Big).
\label{eq.sup_to_max}
\end{align}
In order to bound the probability, we apply Bernstein's inequality to $\widehat  f_\bTheta(\btheta_j) - \E [  \widehat  f_\bTheta(\btheta_j) ]= \sum_{i=n+1}^{2n} Z_i$ with
\[
 Z_i:=\frac{1}{nh_*^{d-1}}C(h_*)K\Big(\frac{1-\langle \bTheta_i,\btheta\rangle}{h_*^2}\Big)
 -\frac{1}{n}\E\big[\widehat f_\bTheta(\btheta_j) \big].
\]
We find $|Z_i|\leq \frac{C_1}{nh_*^{d-1}}$ for some constant $C_1>0,$ and for some constant $C_2>0,$
 \bal
\E\big[ Z_i^2\big]\leq&
\frac{C(h_*)^2}{n^2h_*^{2d-2}}\int_{\sd}K^2\Big(\frac{1-\langle \btheta',\btheta\rangle}{h_*^2}\Big)f_\bTheta(\btheta')d\btheta'
\leq \frac{C_2}{n^2h_*^{d-1}}
\end{align*}
using the boundedness of $K$ and $f_\bTheta$ as well as the definition of $C(h_*)$. Hence, an application of Bernstein's inequality yields with \eqref{eq.sup_to_max},
\bal
  \P\Big( \sup_{\btheta \in \sd} \big | \widehat  f_\bTheta(\btheta) - \E\big[  \widehat  f_\bTheta(\btheta) \big] 
  \big|>ca_n\Big)\leq & 2M\exp\Big(
  \frac{-a_n^2c^2/8}{C_2n^{-1}h_*^{-d+1}+cC_1a_n(6n)^{-1}h_*^{-d+1}}\Big).
\end{align*}
Since $M$ is a polynomial power of $n,$ the claim follows by choosing the constant $c$ large enough. 

For \textit{(iii)} one proceeds similarly with the choice $a_n=(\log(n)\log\log(n))^{-1}$ using the summability of the probabilities.
 \end{proof}

To prove Lemma \ref{p2}, we make use of a slightly modified version of Proposition 2.2 in \cite{MR1876841} which we state below as 
 Proposition \ref{.6}. If $\mathcal{F}$ is a uniformly bounded class of measurable functions on a measurable space $(S, \ca{S})$ with a measurable and bounded envelope $F$, then $\ca F$ is said to be a measurable uniformly bounded VC class of functions  if there are constants
$A, v > 0$ such that 
\[
 \sup_Q N(\mathcal{F}, L_2(Q) ,\varepsilon\|F\|_{L^2(Q)} ) \leq \Big(\frac{A}{\varepsilon}\Big)^v
\]
 for all $0 < \varepsilon < 1$, where $N(T,d,\ve)$ denotes the $\ve$-covering number of the metric space $(T,d)$ and the supremum is taken over all  probability measures on $(S,\ca{S})$.
 
\begin{prop2}\label{.6}
Let $P$ be any probability measure on $(S,\ca S)$ and let $\xi_i,\;i=1,\hdots,n,$ be  independent with common law $P$.
 Let further $\ca{F}$ be a measurable uniformly bounded VC class of functions and let $\sigma^2$ and $U$ be any numbers such 
 that $\sigma^2\geq\sup_{f\in\ca{F}}\Var_Pf,\;U\geq\sup_{f\in\ca{F}}\|f\|_\infty$ and $0<\sigma\leq U$. Then there exist 
  universal constants $C,K',L>0$ such that the exponential inequality
 \begin{align}\begin{split}\label{1.1} 
 &\P\bigg(\sup_{f\in\ca{F}}\Big|\sum_{i=1}^n\big(f(\xi_i)-\E[f(\xi_i)]\big)\Big|>t\bigg)\\
\leq  &K'\exp\bigg(-\frac{1}{K'}\frac{t}{U}\log\bigg(1+\frac{tU}{\big(\sqrt{n}\sigma+L\sqrt{v}U\sqrt{\log(AU\sigma^{-1}}\big)^2}\bigg)\bigg)
 \end{split}\end{align}
 is valid for all $t\geq C(vU\log(AU\sigma^{-1})+\sqrt{vn}\sigma\sqrt{\log(AU\sigma^{-1})}).$
\end{prop2} 

In contrast to Proposition 2.2 in \cite{MR1876841}, Proposition \ref{.6} contains the explicit  dependence of the right hand side of \eqref{1.1} on the constants $A$ and $v$.

\begin{proof}[Proof of Lemma \ref{p2}]
Similarly as in the proof of Lemma \ref{g22}, it is enough to consider the random coefficients model without intercept only and to work under Assumption \ref{p1}, \textit{(i)-(iii)}. If the design density is multivariate
Cauchy, we can derive the properties in a similar way for   $\R\times\sd_+$. An upper bound of the bias can be derived  similarly to Lemma \ref{g22} \textit{(i)}. For the stochastic error
\bal
\P\bigg( \sup_{(s,\btheta)\in\R\times\sd}\Big|\widehat f_{S,\bTheta}(s,\btheta)-\E\big[\widehat f_{S,\bTheta}(s,\btheta)\big]\Big|
>\log\log(n)^{-1}\log(n)^{-2}\bigg)
\end{align*}
we apply Proposition \ref{.6} to the function class
\[
\ca F_n:= \bigg\{(S,\bTheta)\mapsto K\Big(\frac{1-\langle\bTheta,
 \btheta\rangle}{h_+^2}\Big)K\Big(\frac{S-s}{h_+}\Big),\;(s,\btheta)\in\R\times\sd\bigg\},
\]
which depends on $n$ via $h_+$.
By the boundedness of $K$ we find $U,\sigma\lesssim1$.
To show that $\ca F_n$ is a  VC class of functions, we introduce a discretization of $\R\times\sd$ as follows: Let $c>0$ be a sufficiently small constant only
depending on the kernel $K$. We chose a grid $\{\btheta_j:j=1,\hdots,M_1\}$ of $\sd$ with grid
width at most $c\ve h_+^2$. Obviously, this is possible with $M_1 \lesssim (\ve^{-1}h_+^{-2})^{d-1}$. Moreover, introduce the set of intervals $I_k=[k,k+1)$, $k\in\mathbb{Z}$. For each probability measure $Q$ there are at most $\lceil(c\ve)^{-2}\rceil$
sets $I_{i_j}\times\sd,\;j=1,\hdots, \lceil (c\ve)^{-2}\rceil,$ such that $Q(I_{i_j}\times\sd)\geq (c\ve)^2$.
Let  $\{s_j:j=1,\hdots,M_2\}$ be an equidistant grid of 
\[
 \td I_{h_+}:=\Bigg\{s\in\R:\mathrm{dist}\Bigg(s,\bigcup_{j=1}^{\lceil(c\ve)^{-2}\rceil}I_{i_j}\Bigg)\leq 1\Bigg\}
\]
with grid width $ch_+\ve$ 
and let
$ s_{M_2+1}$ denote an arbitrary point in $  \td I_{h_+}^C.$
Basic calculations show 
$M_2\lesssim \ve^{-3}h_+^{-1}.$ Moreover, the subset of  $\ca F_n$ indexed by 
\[
 \big\{s_j:j=1,\hdots M_2+1\big\}\times\big\{\btheta_j:j=1,\hdots,M_1\big\}=:\big\{(s_j,\btheta_j):j=1,\hdots,M_1(M_2+1)\big\}
\]
is an $\ve$-covering set of $\ca F_n$. To see this,
fix $(s,\btheta)\in \td I_{h_+}\times\sd$. Then
\bal
&\int_\R\int_{\sd} \Big|K\Big(\frac{1-\langle\bTheta,
 \btheta\rangle}{h_+^2}\Big)K\Big(\frac{S-s}{h_+}\Big)- K\Big(\frac{1-\langle\bTheta,
 \btheta_j\rangle}{h_+^2}\Big)K\Big(\frac{S-s_j}{h_+}\Big)\Big|^2dQ(S,\bTheta)\\\lesssim&
 \frac{\|\btheta-\btheta_j\|^2}{h_+^4}
 +\frac{|s-s_j|^2}{h_+^2}
\end{align*}
by the Lipschitz continuity of $K$. Hence, by construction of the set $\td I_{h_+}\times\sd$ there exists $j\in
\{1,\hdots,M_1(M_2+1)\}$ such that
\[
  \frac{\|\btheta-\btheta_j\|^2}{h_+^4}
 +\frac{|s-s_j|^2}{h_+^2}< \ve^2.
\]
For $(s,\btheta)\in (\td I_{h_+}\times\sd)^C$ we obtain
\bal
&\int_\R\int_{\sd} \Big(K\Big(\frac{S-s}{h_+}\Big)+K\Big(\frac{S-s_{M_2+1}}{h_+}\Big)\Big)^2dQ(S,\bTheta)< \ve^2
\end{align*}
since the support of $K(\frac{\cdot-s}{h_+})$ is compact  and does not intersect with any of the sets
$I_{i_j}\times\sd,\;j=1,\hdots, \lceil (c\ve)^{-2}\rceil$ for $h_+$ sufficiently small.
A similar argument applies to $K(h_+^{-1}(\cdot-s_{M_2+1}))$. Hence,
\[
 N(\mathcal{F}, L_2(Q) ,\varepsilon ) \lesssim \Big(\varepsilon^{-1}h_+^{(-2d+1)/(d+2)}\Big)^{d+2}
\]
and $\ca F_n$ is a VC class of functions with $v=d+2$ and $A=A_n=h_+^{(-2d+1)/(d+2)}$.
An application of Proposition \ref{.6} yields
\bal
&\P\Big(  \sup_{(s,\btheta)\in\R\times\sd}\big|\widehat f_{S,\bTheta}(s,\btheta)-\E[\widehat f_{S,\bTheta}(s,\btheta)]\big|
>\log\log(n)^{-1}\log(n)^{-2}\Big)\\
=&\P\bigg(\sup_{f\in\ca{F}_n}\Big|\sum_{i=n+1}^{2n}\big(f(S_i,\bTheta_i)-\E[f(S_i,\bTheta_i)]\big)\Big|>\frac{nh_+^d}{C(h_+)\log\log(n)\log(n)^2}\bigg)\\
 \lesssim& \exp\bigg(-\frac{1}{4K'\sigma^2C(h_+)^2}\frac{nh_+^{2d}}{\log\log(n)^{2}\log(n)^4}\bigg)
\end{align*}
for $n$ sufficiently large. We have used that $\log(1+x)=x(1+o(1))$ for $x\rightarrow 0$. The last line of the equation converges to zero 
at a summable rate since $h_+\geq\log(n)^{3/d}n^{-1/(2d)}$ by assumption which concludes the proof of the uniform almost sure
convergence of $\widehat f_{S,\bTheta}$.
\end{proof}

Let us turn to the standard deviation
\begin{align}
	\sigma_{\bt,h,\bv} =  \Big( \int_{\mathbb{S}^{d-1}}\int_{\mathbb{R}} \langle \btheta, \bv \rangle^2
	\big((\mathcal{H}_d\widetilde \phi^{(d-1)})(s) \big)^2 \frac{ f_{S,\bTheta}(\langle \bt, \btheta \rangle+hs,
	\btheta)}{ f_{\bTheta}(\btheta)^2} dsd\btheta  \Big)^{1/2}
	\label{eq.sigma_def}
\end{align}
and its estimator defined in \eqref{eq.sigma_th_def}. The following lemma shows that it is uniformly bounded from above and below. The proof is deferred to Appendix \ref{B}.
\begin{lem2}\label{1.6}
Under Assumptions \ref{ass:radial_symmetric}-\ref{p1} there exist universal constants $C_1,C_2,n_0>0$ such that for any $n>n_0,$
 \[
  C_1\leq \sigma_{\bt,h,\bv}\leq C_2.
 \]
\end{lem2}
The proof is given in Appendix \ref{B}. The consistency of the estimates $\widehat f_\bTheta$ and $\widehat f_{S,\bTheta}$ shows that $\widehat\sigma_{\bt,h,\bv}$
is a consistent estimator of the standard deviation $\sigma_{\bt,h,\bv}$.

\begin{lem2}\label{1.9}
Under Assumptions \ref{ass:radial_symmetric}- \ref{p1},
 \[
  \sup_{(\bt,h,\bv)\in \ca{T}_n} \big|\widehat{\sigma}_{\bt,h,\bv}-{\sigma}_{\bt,h,\bv}\big|
=O\big(\log(n)^{-1}\big)\quad\text{for }n\rightarrow\infty,\text{ almost surely}.
 \]
\end{lem2}
\begin{proof}
By Lemma \ref{1.6},
 \bal
  \big|\widehat \sigma_{\bt,h,\bv}-\sigma_{\bt,h,\bv} \big|&\leq
  \frac{ \big|\widehat \sigma_{\bt,h,\bv}^2-\sigma_{\bt,h,\bv}^2 \big|}{\sigma_{\bt,h,\bv}}\\
  &\hspace{-15mm}\lesssim \int_{\mathbb{S}^{d-1}}\int_{\mathbb{R}} \langle \btheta, \bv \rangle^2
	\big((\mathcal{H}_d\widetilde \phi^{(d-1)})(s) \big)^2\Big|
	\frac{ f_{S,\bTheta}(\langle \bt, \btheta \rangle+hs,
	\btheta)}{ f_{\bTheta}(\btheta)^2}-\frac{\widetilde f_{S,\bTheta}(\langle \bt, \btheta \rangle+hs,
	\btheta)}{\widetilde f_{\bTheta}(\btheta)^2}\Big| dsd\btheta.
 \end{align*}
By Assumption \ref{p1}, $f_\bTheta$ is uniformly bounded from below. Thus, $\widetilde f_\bTheta$ is almost surely uniformly bounded from below for sufficiently large $n$ by Lemma \ref{g22}. This shows that
 \begin{align*}
   &f_{S,\bTheta}(\langle \bt, \btheta \rangle+hs,
	\btheta) \Big|
	\frac{ 1}{ f_{\bTheta}(\btheta)^2}-\frac{ 1}{\widetilde f_{\bTheta}(\btheta)^2}\Big|
	+\frac{ 1}{\widetilde f_{\bTheta}(\btheta)^2}\big|f_{S,\bTheta}(\langle \bt, \btheta \rangle+hs,
	\btheta)-\widetilde f_{S,\bTheta}(\langle \bt, \btheta \rangle+hs,
	\btheta)\big| \\
	&=O\big((\log(n)^{-1}\big)\quad\text{almost surely}.
 \end{align*}
Here we used the boundedness of $f_{S,\bTheta}$ and
\begin{align*}
&\big|f_{S,\bTheta}(\langle \bt, \btheta \rangle+hs,
\btheta)-\widetilde f_{S,\bTheta}(\langle \bt, \btheta \rangle+hs,
\btheta)\big|\\
\leq \; & \big|f_{S,\bTheta}(\langle \bt, \btheta \rangle+hs,
\btheta)-\widehat f_{S,\bTheta}(\langle \bt, \btheta \rangle+hs,
\btheta)\big|+\big|\widehat f_{S,\bTheta}(\langle \bt, \btheta \rangle+hs,
\btheta)-\widetilde f_{S,\bTheta}(\langle \bt, \btheta \rangle+hs,
\btheta)\big|\\
= \; & O\big(\log(n)^{-2}\big)\quad\text{almost surely}
\end{align*}
by Lemma \ref{p2}. The claim follows now from the integrability of 
$\big((\mathcal{H}_d\widetilde \phi^{(d-1)})(s) \big)^2$ proved in Lemma \ref{lem.rad_symm_kernel_explicit}
\textit{(ii)}.
\end{proof}

\begin{lem2}\label{1.8}
Under  Assumptions \ref{ass:f_beta} resp. 2' and  \ref{p1} we have
  \[
  \sup_{(s,\btheta)\in\R\times\sd}\Big|\sqrt{\widetilde f_{S,\bTheta}(s,\btheta)}-\sqrt{f_{S,\bTheta}(s,\btheta)}\Big|=O(\log(n)^{-1})\quad\text{for }n\rightarrow\infty\text{ almost surely}.
 \]
\end{lem2}
\begin{proof}
 This is a direct consequence of
  \[
  \sup_{(s,\btheta)\in\R\times\sd}\Big|{\widetilde f_{S,\bTheta}(s,\btheta)}-{f_{S,\bTheta}(s,\btheta)}\Big|
  =O(\log(n)^{-2})\quad\text{for }n\rightarrow\infty\text{ almost surely}
 \]
 as shown in the proof of Lemma \ref{1.9} and
 \[
  \big|\sqrt{\widetilde f_{S,\bTheta}(s,\btheta)}-\sqrt{f_{S,\bTheta}(s,\btheta)}\big|=
  \frac{ \big|{\widetilde f_{S,\bTheta}(s,\btheta)}-{f_{S,\bTheta}(s,\btheta)} \big|}
  {\sqrt{\widetilde f_{S,\bTheta}(s,\btheta)}+\sqrt{f_{S,\bTheta}(s,\btheta)}}
  \leq \frac{ \big|{\widetilde f_{S,\bTheta}(s,\btheta)}-{f_{S,\bTheta}(s,\btheta)} \big|}{\log(n)^{-1}}.
 \]
\end{proof}
We discussed in Section \ref{sec:multiscale_test} 
that the test statistic $T_{\bt,h,\bv}$ relies on the unknown density
$f_\bTheta$ and therefore we introduced the statistic $\widehat T_{\bt,h,\bv}$, where the density $f_\bTheta$ is replaced
by the estimate $\widetilde f_\bTheta$. An important part of the proof 
of Theorem \ref{thm.main} consists of showing that this replacement is asymptotically negligible. To this end,
the bias of the estimate $1/\widetilde f_\btheta(\btheta)$ has to be controlled.
\begin{lem2}\label{g3}
Under Assumption \ref{p1},
\begin{align*}
	\sup_{\btheta \in \mathbb{S}^{d-1}} \Big| \frac{1}{f_\bTheta(\btheta)} - \E\Big[ \frac{1}{\widetilde f_\bTheta(\btheta)} \Big] \Big|
	=O\big(h_*^\gamma\big)\quad\text{for }n\rightarrow\infty.
\end{align*}
\end{lem2}
\begin{proof}
Uniformly over $\btheta \in \mathbb{S}^{d-1},$
 \bal
 &\Big|\E\Big[\frac{1}{f_\bTheta(\btheta)}-\frac{1}{\widetilde{f}_\bTheta(\btheta)}\Big]\Big|=
\Big| \E\Big[\frac{\widetilde{f}_\bTheta(\btheta)-f_\bTheta(\btheta)}
 {\widetilde{f}_\bTheta(\btheta)f_\bTheta(\btheta)}\Big]\Big|\\
 & \hspace{2cm}\lesssim\Big| \E\Big[\frac{\widetilde{f}_\bTheta(\btheta)-{f}_\bTheta(\btheta)}{{f}_\bTheta(\btheta)^2}\Big]\Big|
 +\Big|\E\Big[\frac{\widetilde{f}_\bTheta(\btheta)-{f}_\bTheta(\btheta)}{\widetilde{f}_\bTheta(\btheta){f}_\bTheta(\btheta)}
  \mathbbm{1}\big\{\exists\btheta':\widehat{f}_\bTheta(\btheta')\leq {f}_\bTheta(\btheta')/2\big\}\Big]\Big|\\
  &\hspace{2cm} \lesssim \big|\E\big[\widetilde{f}_\bTheta(\btheta)-{f}_\bTheta(\btheta)\big]\big|+
  \log(n)h_*^{-d+1}\mathbb{P}\big(\exists\btheta':\widehat{f}_\bTheta(\btheta')\leq {f}_\bTheta(\btheta')/2\big).
 \end{align*}
Following the line of arguments in the proof of Lemma \ref{g22}, it is easy to see that $  \mathbb{P}\big(\exists\btheta':\widehat{f}_\bTheta(\btheta')\leq {f}_\bTheta(\btheta')/2\big)$ decays at a rate which is faster than polynomial. In particular, $\widetilde f_\bTheta=\widehat{f}_\bTheta$ except on a set with probability decaying faster than any polynomial, which concludes the proof using Lemma \ref{g22} \textit{(i)}.
\end{proof}

\section{Proofs related to properties of the Radon transform}\label{B}

 {\it Proof of \eqref{eq.rcm_with_i_fthete_fX}:} Recall that in the random coefficients model with intercept
\[
\bTheta_1 = \zeta_1\frac{(1, X_{1,2}, X_{1,3}, \ldots, X_{1,d})}{\|(1, X_{1,2}, X_{1,3}, \ldots, X_{1,d})\|} 
,
\]
where $\zeta_1$ is a Rademacher variable.
Hence, we obtain
\begin{align*}
 f_\bTheta(\btheta)&=\frac{1}{2}\int_0^\infty r^{d-1} \delta(r\theta_1-1)
 f_\bX(r\theta_2,\hdots,r\theta_d)dr\\&\quad+ \frac{1}{2}\int_0^\infty r^{d-1} \delta(r\theta_1+1)
 f_\bX(-r\theta_2,\hdots,-r\theta_d)dr\\
 &=\frac 1 2\int_0^\infty \frac{r^{d-1}}{\theta_1^{d}} \delta(r-1) f_\bX\big(\tfrac{r}{\theta_1}\theta_2,\hdots,
 \tfrac{r}{\theta_1}\theta_d)dr\mathbbm{1}\{\theta_1> 0\}\\&\quad+
 \frac 1 2\int_0^\infty \frac{r^{d-1}}{\theta_1^{d-1}|\theta_1|} \delta(r+1) f_\bX\big(-\tfrac{r}{\theta_1}\theta_2,\hdots,
 -\tfrac{r}{\theta_1}\theta_d)dr\mathbbm{1}\{\theta_1< 0\}
\\
 &=\frac{1}{2|\theta_1|^d}f_\bX\Big(\frac{\theta_2}{\theta_1},\hdots,\frac{\theta_d}{\theta_1}\Big).
\end{align*}\qed

\begin{proof}[Proof of Lemma \ref{lem.rad_symm_kernel_explicit}] 
By assumption, $\phi_{\bt,h}$ is radially symmetric and satisfies \eqref{eq.def_radially_symm}. We fix a direction $\bv \in \mathbb{S}^{d-1}$ and consider the directional derivative
\[
\partial_{\bv} \phi_{\bt,h}(\bb) = \frac 1{h^{d+1}\Vol(\mathbb{S}^{d-2})} \phi'\left(\frac{\|\bb-\bt\|}{h}\right)\frac{\langle \bb-\bt,\bv \rangle}{\|\bb-\bt\|},
\]
where $\phi'$ is the usual derivative of $\phi$. The Radon transform of this directional derivative is
\begin{align*}
R(\partial_{\bv} \phi_{\bt,h})(s,{\btheta}) &= \int_{\langle \bb,{\btheta}\rangle=s} \partial_{\bv} \phi_{\bt,h}(\bb) d\mu_{d-1}(\bb)\\
&= \frac 1{h^{d+1}\Vol(\mathbb{S}^{d-2})} \int_{\langle \bb,{\btheta}\rangle=s}  \phi'\left(\frac{\|\bb-\bt\|}{h}\right)\frac{\langle \bb-\bt,\bv \rangle}{\|\bb-\bt\|} d\mu_{d-1}(\bb)\\
&=  \frac 1{h^{2}\Vol(\mathbb{S}^{d-2})} \int_{\langle \bb,{\btheta}\rangle = h^{-1}(s-\langle \bt, {\btheta} \rangle)} \phi'\left(\left\|\bb\right\|\right)\frac{\langle \bb,\bv \rangle}{\|\bb\|} d\mu_{d-1}(\bb).
\end{align*}
Set $\widetilde{s} = h^{-1}(s-\langle \bt, {\btheta} \rangle).$ For $d>2,$ using the definition of $\widetilde \phi$ in \eqref{eq.tilde_phi_n_def}, 
\begin{align*}
R(\partial_{\bv} \phi_{\bt,h})(s,{\btheta}) &= \frac 1{h^{2}\Vol(\mathbb{S}^{d-2})} \int_{\langle \bb,{\btheta}\rangle = \widetilde{s}} \frac{\phi'\left(\left\|\bb\right\|\right)}{\|\bb\|}\langle \bb,\bv \rangle d\mu_{d-1}(\bb)\\
& = \frac 1{h^{2}\Vol(\mathbb{S}^{d-2})} \int_0^\infty \frac{\phi'\big(\sqrt{\widetilde{s}^2+r^2}\big)}{\sqrt{\widetilde{s}^2+r^2}} \int_{\mathbf{w} \perp {\btheta}, \|\mathbf{w}\| = r}  \langle {\btheta} \widetilde{s} + \mathbf{w}, \bv \rangle d\mathbf{w} dr\\
& = \frac 1{h^{2}\Vol(\mathbb{S}^{d-2})} \int_0^\infty \frac{\phi'\big(\sqrt{\widetilde{s}^2+r^2}\big)}{\sqrt{\widetilde{s}^2+r^2}} \int_{\mathbf{w} \perp {\btheta}, \|\mathbf{w}\| = r}  \langle {\btheta} \widetilde{s}, \bv \rangle d\mathbf{w} dr\\
& =\frac{\langle {\btheta} , \bv \rangle}{h^2}  \int_0^\infty r^{d-2} \phi'\left(\sqrt{\widetilde{s}^2+r^2}\right)\frac{\widetilde{s}}{\sqrt{\widetilde{s}^2+r^2}} dr\\
& = \frac{\langle {\btheta} , \bv \rangle}{h^2}  \int_0^\infty r^{d-2} \frac{\partial}{\partial \widetilde s} \phi\left(\sqrt{\widetilde{s}^2+r^2}\right) dr \\
 &=  \frac{\langle {\btheta} , \bv \rangle}{h^2} \widetilde{\phi}(\widetilde{s})\\
&=  \frac{\langle {\btheta} , \bv \rangle}{h^2} \widetilde{\phi}\left(\frac{s - \langle \bt, {\btheta} \rangle}{h}\right).
\end{align*}
For $d=2$ let $\mathbf{w}\perp\btheta$ with $\|\mathbf{w}\|=1$ and write $\bb=\btheta\widetilde{s}+r\mathbf{w}$ for $r\in\mathbb{R}$. Then
\begin{align*}
R(\partial_{\bv} \phi_{\bt,h})(s,{\btheta})
& = \frac 1{h^{2}\Vol(\mathbb{S}^{0})} \int_{-\infty}^\infty \frac{\phi'\big(\sqrt{\widetilde{s}^2+r^2}\big)}{\sqrt{\widetilde{s}^2+r^2}}   \langle {\btheta} \widetilde{s} + r\mathbf{w}, \bv \rangle  dr\\
& = \frac 1{h^{2}\Vol(\mathbb{S}^{0})} \int_{-\infty}^\infty \frac{\phi'\big(\sqrt{\widetilde{s}^2+r^2}\big)}{\sqrt{\widetilde{s}^2+r^2}} \langle {\btheta} \widetilde{s}, \bv \rangle  dr\\
& =\frac{\langle {\btheta} , \bv \rangle}{h^2}  \int_0^\infty  \phi'\left(\sqrt{\widetilde{s}^2+r^2}\right)\frac{\widetilde{s}}{\sqrt{\widetilde{s}^2+r^2}} dr,
\end{align*}
as $r\mapsto \frac{\phi'\left(\sqrt{\widetilde{s}^2+r^2}\right)}{\sqrt{\widetilde{s}^2+r^2}} $ is an even function. Now we can proceed similarly as in the case $d>2$.

If $d$ is odd, the proof of the representation of $A(\partial_{\bv} \phi_{\bt,h})(s,{\btheta})$ is completed by taking the $(d-1)$-th derivative with respect to the variable $s$. 
If $d$ is even, for any function $f:\mathbb{R}\rightarrow\mathbb{R}$, for any fixed $z\in\mathbb{R}$, and for $h>0$
\begin{align*}
	\mathcal{H}_d\Big(s \mapsto f\Big(\frac{s-z}{h}\Big)\Big)(u)
	&=\frac{1}{\pi} \text{p.v.}  \int_{-\infty}^{\infty} f\Big( \frac{s-z}{h}\Big) \frac{1}{u-s} ds \\
	&= \frac{1}{\pi}\lim_{\epsilon\rightarrow 0^+} 
	\int_{(-\infty, u-\epsilon] \cup [u+\epsilon, \infty) }  f\Big( \frac{s-z}{h}\Big) \frac{1}{u-s} ds \\
	&=\frac{1}{\pi} \lim_{\epsilon\rightarrow 0^+} 
	\int_{(-\infty, (u-z)/h-\epsilon] \cup [(u-z)/h+\epsilon, \infty) } f(s) \frac{1}{(u-z)/h -s} ds \\
	&= \big(\mathcal{H}_df\big) \Big(\frac{u-z}{h}\Big)\quad\tn{for }u\in\R,
\end{align*}
by substitution. That $\ca{H}_df$ exists is shown below for the choice $f=\widetilde{\phi}^{(d-1)}$. Hence, we obtain $A(\partial_{\bv}\phi_{\bt,h})(s,\btheta) = \langle \btheta, \bv \rangle h^{-d-1} (\mathcal{H}_d\widetilde \phi^{(d-1)})(\tfrac{s-\langle \bt, \btheta\rangle}{h})$ for $d$ even.\\

Next we prove $\|\widetilde{\phi}^{(k)}\|_\infty<\infty$ for $k=0,\hdots,d+1.$ The case $k=0$ is obvious. We use the chain rule for higher order 
derivatives given by Fa\`a di Bruno's formula
\begin{equation}
\frac{d^k}{dz^k} f_1(f_2(z)) = \sum_{(m_1,...,m_k) \in \mathcal{M}_k} \frac{k!}{m_1!...m_k!} f_1^{(m_1+...+m_k)}(f_2(z)) 
\prod_{j=1}^k \Big( \frac{f_2^{(j)}(z)}{j!} \Big)^{m_j} ,
\label{eq.FaaDiBruno}
\end{equation}
where $\mathcal{M}_k$ is the set of all $k$-tuples of non-negative integers satisfying $\sum_{j=1}^k j m_j = k$. 
Since $z\mapsto\phi(\sqrt{z^2+r^2})$ is a.e. $(k+1)$-times continuously differentiable, we can interchange the integral with
the $k$-fold differentiation for the variable $z$ provided that
\begin{align*}
 \int_0^\infty\Big| r^{d-2} \frac{\partial^{k+1}}{\partial z^{k+1}} \phi\left(\sqrt{z^2+r^2}\right)\Big| dr
\end{align*}
exists for all $k=1,\hdots,d+1$. Applying \eqref{eq.FaaDiBruno} with $f_1=  \phi$ and $f_2=\sqrt{\cdot^2+r^2}$ gives
\begin{align*}
	\frac{\partial^{k+1}}{\partial z^{k+1}} \phi\left(\sqrt{z^2+r^2}\right) 
	= \sum_{(m_1,...,m_{k+1}) \in \mathcal{M}_{k+1}} C_{m_1,\ldots,m_{k+1}}
	\phi^{(M)}\left(\sqrt{z^2+r^2}\right) \prod_{j=1}^{k+1} \big( f_2^{(j)}(z)\big)^{m_j}
\end{align*}
for suitable constants $C_{m_1,\ldots,m_{k+1}}$ and $M= \sum_{j=1}^{k+1} m_j.$ Applying the chain rule to $f_2^{(j)}$ yields
\begin{align*}
	f_2^{(j)}(z) = \sum_{\{\ell_j, k_j: \ell_j+2k_j =j\}} C_{\ell_j,k_j} z^{\ell_j} (z^2+r^2)^{1/2-\ell_j-k_j}
\end{align*}
for non-negative integers $\ell_j,k_l$ and suitable constants $C_{\ell_j,k_j}$. As $\phi$ is compactly supported, it remains to show that each of the functions
\begin{align}\label{u1}
	z\mapsto \int_0^{\sqrt{1-z^2}} r^{d-2} \phi^{(M)}\left(\sqrt{z^2+r^2}\right) |z|^{\sum_{j=1}^{k+1} \ell_j m_j}(z^2+r^2)^{M/2-\sum_{j=1}^{k+1} (\ell_j+k_j) m_j} dr
\end{align}
for $|z|\leq 1$ is uniformly bounded, where $\ell_j,k_j$ are arbitrary elements of the set $\{\ell_j, k_j: \ell_j+2k_j =j\},\;j=1,\hdots,k+1$. Notice that 
\[
 M/2-\sum_{j=1}^{k+1} (\ell_j+k_j) m_j=\sum_{j=1}^{k+1} (\frac{1}{2}-\ell_j-k_j) m_j<0.
\]
A uniform bound for the integral on the right hand side of \eqref{u1} can be found
easily when $z$ is bounded away from zero. We can thus assume that $|z|\leq\sqrt{1-z^2}$.
Splitting the integral  $\int_0^{\sqrt{1-z^2}}= \int_0^{|z|}+\int_{|z|}^{\sqrt{1-z^2}}$ and using that by Taylor expansion and Assumption \ref{ass:radial_symmetric},
\[
 \phi^{(j)}\left(\sqrt{z^2+r^2}\right)\lesssim (z^2+r^2)^{(3-j)/2}\text{ for }j=1,2\text{ and } \phi^{(M)}\lesssim1\text{ for }M\leq d+2
\]
as well as
$\max\{z^2,r^2\}\leq z^2+r^2\leq 2\max\{z^2,r^2\}$, we obtain an upper bound (up to some constant) for the integral on the right hand side of \eqref{u1} by
\begin{align*}
 & |z|^{\sum_{j=1}^{k+1} \ell_j m_j+M-2\sum_{j=1}^{k+1} (\ell_j+k_j) m_j+\max\{3-M,0\}}  \int_{0}^{|z|}r^{d-2}dr\\&+
 |z|^{\sum_{j=1}^{k+1}\ell_j m_j}\int_{|z|}^{\sqrt{1-z^2}} r^{d-2+\max\{3-M,0\}+M-2\sum_{j=1}^{k+1} (\ell_j+k_j) m_j}dr
\\
 \lesssim& |z|^{\sum_{j=1}^{k+1} \ell_j m_j+M-2\sum_{j=1}^{k+1} (\ell_j+k_j) m_j+d-1+\max\{3-M,0\}}+1.
\end{align*}
By the use of $\ell_j+2k_j=j$, $\sum_{j=1}^{k+1} jm_j ={k+1}$ and $k\leq d+1$, we find that this is bounded by
$ z^{-3+M+\max\{3-M,0\}}+1 $ which proves the result.

Next, we prove that $\mathcal{H}_d\widetilde{\phi}^{(d-1)}$ exists. Recall that $\|\widetilde{\phi}^{(d)}\|_\infty<\infty$  and consequently, $\widetilde{\phi}^{(d-1)}$ is Lipschitz continuous. For any 
Lipschitz continuous function $f$ with compact support,
\[
 \Big|\int_{-\infty}^{u-1}\frac{f(x)}{u-x}dx\Big| \vee \Big|\int_{u+1}^\infty\frac{f(x)}{u-x}dx\Big| \leq \|f\|_\infty\lambda(\tn{supp}f),
\]
where $\lambda(\tn{supp}f)$ denotes the Lebesgue measure of the support of $f$. Moreover,
\begin{align*}
 \lim_{\epsilon \rightarrow 0^+}\Big( \int_{u-1}^{u-\ve}\frac{f(x)}{u-x}dx +\int_{u+\ve}^{u+1}\frac{f(x)}{u-x} dx\Big)	
	= \lim_{\epsilon \rightarrow 0^+} \int_{\ve}^{1}\frac{f(u-x)-f(u+x)}{x}dx .
\end{align*} 
By the Lipschitz-continuity of $f$, $|f(u-x)-f(u+x)|\lesssim |x|$ such that the r.h.s. can be bounded by a constant that
does not depend on $u$. The result follows with $f=\widetilde{\phi}^{(d-1)}$. This proves assertion \textit{(i)} in the Lemma.

Finally, we prove \textit{(ii)}. As shown above, $\widetilde{\phi}^{(d-1)}$ is bounded. For odd dimension $d$ the claim therefore follows from substitution and the
compact support of $\widetilde{\phi}^{(d-1)}$. For $d$ even, substitution and the fact that the Hilbert transform $\mathcal{H}_d$ defines a 
bounded operator  $L^k(\mathbb{R})\rightarrow L^k(\mathbb{R})$ for all $1<k<\infty$ yield the required result.

\end{proof}

\begin{proof}[Proof of Lemma \ref{1.6}]
The existence of a uniform upper bound of $\sigma_{\bt,h,\bv}$ follows directly from the boundedness of $f_{S,\bTheta}$. The uniform lower bound of $f_\bTheta$ follows from Assumption \ref{p1}. The integrability of $\big(\ca{H}_d(\widetilde{\phi}^{(d-1)})(s)\big)^2$ is shown in the proof 
of Lemma \ref{lem.rad_symm_kernel_explicit} \textit{(ii)}. For the lower bound of $\sigma_{\bt,h,\bv}$ recall that
\bal
 \frac{f_{S,\bTheta}(\langle \bt,\btheta\rangle+hs,\btheta)}{f_\bTheta(\btheta)}
 &=f_{S|\bTheta}(\langle \bt,\btheta\rangle+hs,\btheta)
 = \int_{\langle  \bb,\btheta \rangle=\langle \bt,\btheta\rangle+hs  } f_{\bbeta}(\bb) d\mu_{d-1}(\bb).
\end{align*}
By Assumption \ref{ass:f_beta}, $f_\bbeta(\bb)\geq c_\bbeta>0$ for all $\bb\in[\mathbf a_1,\mathbf a_2]$ and $f_\bbeta$ is uniformly
continuous. Hence, there exists  $\delta>0$, 
which does not depend on $h$, such that $f_\bbeta$
is uniformly bounded from below in the ball $B_{\delta}(\bt)$ of radius $\delta$ around any $\bt\in[\mathbf a_1,\mathbf a_2]$,
say, $f_{\bbeta}(\bb)>c_{\bbeta}/2$ for all $\bb\in \bigcup_{\bt\in[\mathbf a_1,\mathbf a_2]}B_\delta(\bt)$.
 Define for $s^2<\delta^2/(dh^2)$
 \begin{align*} 
A_{\delta,\bt,h}:=\bigl\{\bb\in\R^d:\bb=\bt+hs\btheta+\rho_2\btheta_2^{\bot}+\hdots+\rho_d\btheta_d^\bot,\,\rho_j^2<\delta^2/d,\;j=2,\hdots,d
\bigr\},
 \end{align*}
 where $\btheta_2^\bot,\hdots,\btheta_d^\bot$ form an orthonormal basis of the orthogonal complement of $\rm span \{\btheta\}$.
 Clearly, $\mu_{d-1}(A_{\delta,\bt,h})=(2\delta)^{d-1}d^{(1-d)/2}>0$, and all $\bb\in A_{\delta,\bt,h}$ satisfy
 \begin{align*}
 \|\bt-\bb\|^2=(hs)^2+\rho_2^2+\hdots+\rho_d^2< \frac{\delta^2}{d}+\delta^2\frac{d-1}{d}=\delta^2 \quad\text{and}\quad \langle \bb ,\btheta \rangle=\langle \bt,\btheta\rangle+hs.
 \end{align*}
In particular, $A_{\delta,\bt,h}\subset B_{\delta}(\bt).$ Thus,
\begin{align*}
  \int_{\langle \bb,\btheta  \rangle=\langle \bt,\btheta\rangle+hs  } f_{\bbeta}(\bb) d\mu_{d-1}(\bb)\geq
  \int_{ A_{\delta,\bt,h}  } f_{\bbeta}(\bb) d\mu_{d-1}(\bb)\geq \frac{c_{\bbeta}}{2}\mu_{d-1}(A_{\delta,\bt,h})>0.
\end{align*}
Hence,
 $Rf_\bbeta(\langle \bt,\btheta\rangle+hs,\btheta)$ is uniformly bounded from below for all
 $\mathbf a_1+ h\leq \bt \leq \mathbf a_2-h$, $\btheta\in \mathbb{S}^{d-1}$ and $|s|<\delta/(\sqrt{d}h)$. Therefore,
\beq{qht3}
\sigma_{\bt,h,\bv}^2\gtrsim
\int_{-\delta/(\sqrt{d}h)}^{\delta/(\sqrt{d}h)}\big(\ca{H}_d(\widetilde{\phi}^{(d-1)})(s)\big)^2ds
\geq\int_{-\delta/(\sqrt{d}h_{\max})}^{\delta/(\sqrt{d}h_{\max})}\big(\ca{H}_d(\widetilde{\phi}^{(d-1)})(s)\big)^2ds,
\eeq
where the inequality holds uniformly over $\ca{T}$. 

In quantum homodyne tomography, Assumption 2' \textit{(iii)}
yields
\begin{align*}
  \int_{\langle \bb,\btheta  \rangle=\langle \bt,\btheta\rangle+hs  } f_{\bbeta}(\bb) d\mu_{d-1}(\bb)\geq c_\bbeta
\end{align*}
 for $s^2<\delta^2/(dh^2)$ if $\delta$ is sufficiently small. Hence, \eqref{qht3} holds in this case as well. 
  Furthermore, since $\ca{H}_d(\widetilde{\phi}^{(d-1)})\in L^2(\R),$ we obtain
\[
\sigma_{\bt,h,\bv}^2\gtrsim\int_\R\big(\ca{H}_d(\widetilde{\phi}^{(d-1)})(s)\big)^2ds
+o(1)\quad\text{for }n\rightarrow\infty.\]
If $\|\ca{H}_d(\widetilde{\phi}^{(d-1)})\|_2\neq0$  there exists $n_0=n_0(\delta,d,\phi)\in\mathbb{N}$ 
such that
\[
\sigma_{\bt,h,\bv}^2\gtrsim\frac{1}{2}\int_\R\big(\ca{H}_d(\widetilde{\phi}^{(d-1)})(s)\big)^2ds=\frac{1}{2}\big\|\widetilde{\phi}^{(d-1)}\big\|_2^2
\]
for all $n> n_0$.
The equality on the r.h.s. is trivial for odd dimensions $d$ and follows for even dimensions from the anti self-adjointness of the Hilbert transform and $\ca{H}_d\ca{H}_df=-f$. 
\end{proof}

\section{Proof of Theorem \ref{thm.main}}
If $\|\widetilde{\phi}^{(d-1)}\|_2 = 0$, Theorem \ref{thm.main} obviously holds. In the following we assume $\|\widetilde{\phi}^{(d-1)}\|_2\neq 0$ and define
\[
 a_{\bt,h,\bv}(s,\btheta):=h^{d+1} \Lambda(\partial_{\bv} \phi_{\bt,h})(s,{\btheta}) = \langle \btheta,\bv\rangle(\mathcal{H}_d\widetilde \phi^{(d-1)} )\Big(
 \frac{s - \langle \bt, {\btheta} \rangle}{h}\Big),
\]
where the equality follows from Lemma \ref{lem.rad_symm_kernel_explicit} \textit{(i)}.

\subsection{Controlling the effect of density estimation in the test statistic}

\begin{thm2}\label{1.16}
Under the assumptions of Theorem \ref{thm.main},
\[
\sup_{(\bt,h,\bv)\in \ca{T}_n} \beta_h\sqrt{n}\frac{\big||\widehat T_{\bt,h,\bv}-\E[T_{\bt,h,\bv}]|
-| T_{\bt,h,\bv}-\E[T_{\bt,h,\bv}]|\big|}
  	{\sigma_{\bt,h,\bv}}
 =o_{\mathbb{P}}(1), \ \text{as} \ n\rightarrow\infty.
 \]
\end{thm2}

\begin{proof}
By the triangle inequality
  \begin{align*}
  \big||\widehat T_{\bt,h,\bv}-\E[T_{\bt,h,\bv}]|-| T_{\bt,h,\bv}-\E[T_{\bt,h,\bv}]|\big|
  \leq U_{\bt, h, \bv} + V_{\bt, h, \bv}
   \end{align*}
   with $U_{\bt, h, \bv}:=  |\widehat T_{\bt,h,\bv}-T_{\bt,h,\bv}-\E[\widehat T_{\bt,h,\bv}-T_{\bt,h,\bv}]|$ and $V_{\bt, h, \bv}:=
  | \E[\widehat T_{\bt,h,\bv}-T_{\bt,h,\bv}]|.$  We first bound  $V_{\bt, h, \bv}$ using
\[
V_{\bt, h, \bv} = \Big|\frac{1}{\sqrt{h}}\int_{\mathbb{S}^{d-1}}\int_{\R}a_{\bt,h,\bv}(s,\btheta)\E
\Big[\frac{1}{\widetilde f_\bTheta(\btheta)}-\frac{1}{f_\bTheta
 (\btheta)}\Big]f_{S,\bTheta}(s,\btheta)dsd\btheta\Big|
\]
and
\begin{align}\label{1.11}
\int_{\R}\big|a_{\bt,h,\bv}(s,\btheta)f_{S,\bTheta}(s,\btheta)\big|ds
&\lesssim{h}\int_{\R}\big|(\mathcal{H}_d\widetilde \phi^{(d-1)} )(s)\big|
	f_{S,\bTheta}(hs+\langle \bt,\btheta\rangle,\btheta)ds \lesssim h \log( h)^2.
\end{align}
The last inequality follows for odd dimension $d$ by the boundedness of $f_{S,\bTheta}$ and the integrability of $\widetilde \phi^{(d-1)}.$
For even dimension, recall that $\mathcal{H}_d\widetilde \phi^{(d-1)}$ is bounded as shown in the  
proof of Lemma \ref{lem.rad_symm_kernel_explicit}. Notice that
\[
 \int_{2}^{4/h^2}\frac{\big|(\mathcal{H}_d\widetilde \phi^{(d-1)} )(s)\big|}{\log(s)^2}\log(s)^2
	f_{S,\bTheta}(hs+\langle \bt,\btheta\rangle,\btheta)ds\lesssim\log( h)^2
\]
by $ |(\mathcal{H}_d\widetilde{\phi}^{(d-1)})(s)|\lesssim(1+s^2)^{-1/2}$ (which holds for any function with compact support
and bounded Hilbert transform) and the integrability of 
$(1+s^2)^{-1/2}\log(s)^{-2}$ for $s\geq 2$. For the remainder, we find
\[
 \int_{4/h^2}^\infty\frac{\big|(\mathcal{H}_d\widetilde \phi^{(d-1)} )(s)\big|}{\log(s)^2}\frac{\log(s)^2}{\log(hs+\langle \bt,\btheta\rangle)^2}
 \log(hs+\langle \bt,\btheta\rangle)^2
	f_{S,\bTheta}(hs+\langle \bt,\btheta\rangle,\btheta)ds\lesssim1
\]
by the boundedness of $s\mapsto\log(|s|)^2f_{S,\bTheta}(s,\btheta)$
for all $|s|\geq 2,\btheta\in \mathbb{S}^{d-1},$  and
\[
 \frac{\log(s)}{\log(hs+\langle \bt,\btheta\rangle)}\leq \frac{\log(s)}{\log(hs/2)}= \frac{\log(s)}{\log(h/2)+\log(s)}\leq2,
\]
as $\log(h/2)\geq -\log(s)/2$ for $s\geq 4/h^2.$ A similar argument can be used to bound the integral $ \int_{-\infty}^{-2}|a_{\bt,h,\bv}(s,\btheta)f_{S,\bTheta}(s,\btheta)|ds.$ Applying Lemma \ref{g3} with bandwidth $h_* = \log(n)^{7/(d-1)} n^{-1/( d-1)}$ gives
  \beq{g5}
 \sup_{(\bt,h,\bv)\in \ca{T}_n}V_{\bt, h , \bv}\lesssim\log(h_{\tn{max}})^2
 \sqrt{h_{\tn{max}}}\log(n)^{7\gamma/(d-1)}n^{-\gamma/( d-1)}.
\eeq
Next, we prove $\rho_n:=\mathbb{P}(\sup_{(\bt,h,\bv)\in \ca{T}_n} U_{\bt, h , \bv} \geq \delta_n )\rightarrow 0$ as $n\rightarrow \infty,$ where $\delta_n :={(n\log(n))}^{-1/2}$. If for some positive constant $c$ 
\begin{align*}
	A_n :=\Big\{ (S_i, \bTheta_i)_{i=n+1, \ldots, 2n} : \sup_{\btheta\in \mathbb{S}^{d-1}}\big|\widetilde f_\bTheta(\btheta) -\mathbb{E}[\widetilde f_\bTheta(\btheta)]\big|\leq c\sqrt{\tfrac{\log n}{nh_*^{d-1}}} \Big\},
\end{align*}
then by Lemma \ref{g22}
\begin{align*}
 \rho_n\leq &\;\mathbb{E}\Big[\mathbb{P}\Big(\sup_{(\bt,h,\bv)\in\ca{T}_n} U_{\bt, h , \bv}\geq \delta_n\;\big|\; (S_i, \bTheta_i)_{i=n+1, \ldots, 2n} \Big) \mathbbm{1} (A_n)\Big] + \mathbb{P}(A_n^c)\\
 \leq & \sum_{(\bt,h,\bv)\in\ca{T}_n} \mathbb{E}\Big[\mathbb{P}\Big( U_{\bt, h , \bv}\geq \delta_n\;\big|\; (S_i, \bTheta_i)_{i=n+1, \ldots, 2n} \Big) \mathbbm{1} (A_n)\Big] + o(1)
\end{align*}
for sufficiently large $c.$ Now we apply Bernstein's inequality to
\begin{align*}
 U_{\bt,h,\bv}
 &=\Big| \sum_{i=1}^n\Big\{\frac{1}{n\sqrt{h}}
 a_{\bt,h,\bv}(S_i,\bTheta_i)\Big(\tfrac{1}{\widetilde f_\bTheta(\bTheta_i)}-\tfrac{1}{f_\bTheta(\bTheta_i)}\Big)
 -\frac{1}{n}\E[\widehat T_{\bt,h,\bv}-T_{\bt,h,\bv}]\Big\} \Big|.
 \end{align*}
 By Lemma \ref{lem.rad_symm_kernel_explicit} \textit{(i)}, $|a_{\bt,h,\bv}(s,\btheta)|$ can be bounded by a constant uniformly over $(s,\btheta)\in \R\times \mathbb{S}^{d-1}.$ Moreover,
\[
 \Big|\frac{1}{\widetilde f_\bTheta(\btheta )}-\frac{1}{f_\bTheta(\btheta)}
 \Big|\leq \frac{  \big|\widetilde f_\bTheta(\btheta)-\E[\widetilde  f_\bTheta(\btheta)] \big|
 + \big|\E[\widetilde f_\bTheta(\btheta)]-f_\bTheta(\btheta) \big|}
 {\widetilde f_\bTheta(\btheta)f_\bTheta(\btheta)} .
\]
The inequality $\widetilde f_\bTheta\geq \log(n)^{-1}$, the uniform lower bound of $f_\bTheta$,
 $\widetilde f_\bTheta=\widehat f_\bTheta$ almost surely for $n$ sufficiently large, Lemma \ref{g22}, and the definition of $h_*$ imply that each summand in $U_{\bt,h,\bv}$ is bounded on $A_n$ by
 \begin{align*}
 &\leq C\frac{\log(n)}{n\sqrt{h}}\Big(\sqrt{\frac{\log(n)}{nh_*^{d-1}}}+h_*^\gamma\Big)\leq C_1
 \frac{1}{n\sqrt{h_{\tn{min}}}\log(n)^{2}}
\end{align*}
for some constants $C,C_1>0$.
By a change of variables in the integral for the variable $s$, the uniform boundedness of $f_{S,\bTheta}$, and the integrability of $a_{\bt,h,\bv}^2$
  as shown in Lemma \ref{lem.rad_symm_kernel_explicit} \textit{(i)}, we find for the conditional variance with a similar argument as above
\bal
  \Var \big (U_{\bt,h,\bv} \, | \, (S_i, \bTheta_i)_{i=n+1, \ldots, 2n} \big)
&\leq C \frac{1}{n^2}\sup_{\btheta\in \mathbb{S}^{d-1}}\Big(
 \tfrac{1}{\widetilde f_\bTheta(\btheta)}-\tfrac{1}{f_\bTheta(\btheta)}\Big)^2
 \leq C_2 n^{-2}\log(n)^{-4}
\end{align*}
with some constants $C,C_2>0$. Bernstein's inequality yields
\bal
 \rho_n\lesssim& |\ca{T}_n|\exp\Big(-\frac{\delta_n^2/2}{C_2n^{-1} \log(n)^{-4} 
  +\frac{C_1\delta_n}{3n\sqrt{h_{\tn{min}}}\log(n)^{2}}}\Big)+o(1)\\=&
  |\ca{T}_n|\exp\Big(-\frac{(n\log(n))^{-1}/2}{C_2n^{-1} \log(n)^{-4} 
  +\frac{C_1}{3n^{3/2}\sqrt{h_{\tn{min}}}\log(n)^{5/2}}}\Big)+o(1)=o(1),
\end{align*}
as $h_{\tn{min}}\geq n^{-1}$. Finally, the claim follows from $\beta_h\lesssim \frac{\sqrt{\log(n)}}{\log\log(n)}$ and the boundedness from below of $\sigma_{\bt,h,\bv}$ shown in Lemma \ref{1.6}.  
 \end{proof}
 
\subsection{Approximation of the limit statistic}
Define the process
 \begin{align*}
 X_{\bt,h,\bv} =  h^{-1/2}\int_{\mathbb{S}^{d-1}}\int_{\mathbb{R}}  \langle {\btheta} ,
 \bv \rangle(\mathcal{H}_d\widetilde \phi^{(d-1)} )\left(\frac{s - \langle \bt, {\btheta} 
 \rangle}{h}\right) \frac{\sqrt{ f_{S,\bTheta}(s  ,\btheta)}}{ f_{\bTheta}(\btheta)} W(dsd\btheta).
\end{align*}
Note that $X_{\bt,h,\bv}$ corresponds to the process $\widehat X_{\bt,h,\bv}$ where the density estimators have been replaced
by the true densities.  
The proof of Theorem \ref{thm.main} relies on a recently obtained Gaussian approximation result which is reproduced here for convenience. 

\begin{thm2}[\cite{Chernozhukov2017}, Proposition 2.1]\label{1}
 Let $\mathbf{X}_1,\hdots,\mathbf{X}_n$ be independent random vectors in $\R^{2p}$ with  $\E[\mathrm{X}_{i,j}]=0$ and $\E[\mathrm{X}_{i,j}^2]<\infty$
 for $i=1,\hdots,n,\;j=1,\hdots,2p$. Moreover, let $\mathbf{Y}_1,\hdots,\mathbf{Y}_n$ be independent random vectors in $\R^{2p}$ 
 with $\mathbf{Y}_i\sim N(\mathbf 0,\E[\mathbf{X}_i\mathbf{X}_i^\top]),\;i=1,\hdots,n$. Let $b,q>0$ be some constants and let $B_n\geq1$ be a sequence of constants,
 possibly growing to infinity as $n\rightarrow\infty$. Denote further by $\ca{A}_{2p}'$  the set of all hyperrectangles in
 $\R^{2p}$ of the form $A=\big\{\bs x\in\R^{2p}:\bs a\leq\bs x \leq \bs b\big\}$
for $-\infty\leq \bs a\leq\bs b\leq\infty$. Assume that

\begin{tabular}{ll}
(i) & \  $n^{-1}\sum_{i=1}^n\E[\mathrm{X}_{i,j}^2]\geq b$ for all $1\leq j\leq 2p$; \\[0.3cm] 
(ii) & \  $n^{-1}\sum_{i=1}^n\E[|\mathrm{X}_{i,j}|^{2+k}]\leq B_n^k$ for all $1\leq j\leq 2p$ and $k=1,2$; \\[0.3cm] 
(iii) & \  $\E\big[\big(\max_{1\leq j\leq 2p}|\mathrm{X}_{i,j}|/B_n\big)^q\big]\leq 2$ for all $i=1,\hdots,n$\\
\end{tabular}

and define
\[
 D_n^{(1)}:=\Big(\frac{B_n^2\log^7(2pn)}{n}\Big)^\frac{1}{6},\quad D_{n,q}^{(2)}:=\Big(\frac{B_n^2\log^3(2pn)}{n^{1-2/q}}\Big)^\frac{1}{3}.
\]
Then there exists a constant $C$ only depending on $b$ and $q,$ such that
\[
 \sup_{A\in\ca{A}_{2p}'}\Big|\P\Big( \frac{1}{\sqrt{n}}\sum_{i=1}^n\mathbf{X}_i\in A\Big)-\P\Big(\frac{1}{\sqrt{n}}\sum_{i=1}^n\mathbf{Y}_i\in A\Big)\Big|
 \leq C(D_n^{(1)}+D_{n,q}^{(2)}).
\]
\end{thm2}

\begin{thm2}\label{1.15}
Under the assumptions of Theorem \ref{thm.main},
 \bal
 &\Big( \beta_h\Bigl(\sqrt{n} \frac{|T_{\bt,h,\bv}-E[T_{\bt,h,\bv}]|}{{\sigma}_{\bt,h,\bv}} - \alpha_h\Bigr)
  \Big)_{(\bt,h,\bv)\in\ca{T}_n} \leftrightarrow \Big( \beta_h\Bigl( \frac{| X_{\bt,h,\bv}|}{{\sigma}_{\bt,h,\bv}} - \alpha_h\Bigr)
  \Big)_{(\bt,h,\bv)\in\ca{T}_n}.
 \end{align*}
\end{thm2}
\begin{proof}
To take absolute values into account, we introduce the set
\beq{1.13}
 \ca{T}_n':=\ca{T}_n\cup\big\{(\bt,h,-\bv):(\bt,h,\bv)\in\ca{T}_n\big\}=:\big\{(\bt_j,h_j,\bv_j):j=1,\hdots,2p\big\}.
\eeq
Moreover, for $i=1,\hdots,n,$ let $\mathbf{X}_i
 :=( \mathrm{X}_{i,1},\hdots, \mathrm{X}_{i,2p})^\top$ with
\[
 \mathrm{X}_{i,j}:=\Upsilon_j(S_i,\bTheta_i)-\E[\Upsilon_j(S_i,\bTheta_i)],\ \  \text{and} 
 \ \ \Upsilon_j(s,\btheta):=\frac{a_{\bt_j,h_j,\bv_j}(s, \btheta)}{\sigma_{\bt_j,h_j,\bv_j}\sqrt{h}_j f_{\bTheta}(\btheta)}, \ \ \text{for} \ j=1,\hdots,2p.
\]
Notice that $\sum_{i=1}^nX_{i,j}=n\sigma_{\bt_j,h_j,\bv_j}^{-1}(T_{\bt_j,h_j,\bv_j}-\E[T_{\bt_j,h_j,\bv_j}])$. In a first step, we show that for $\mathbf{Z}\sim N(\mathbf 0,\E[ \mathbf{X}_1 \mathbf{X}_1^\top]),$
\begin{align}
	\sup_{A\in\ca{A}_{2p}'}\Big|\P\Big( \frac{1}{\sqrt{n}}\sum_{i=1}^n\mathbf{X}_i\in A\Big)-\P\Big(\mathbf{Z}\in A\Big)\Big|
 \rightarrow 0.
 \label{eq.to_prove_by_CCK}
\end{align}

Observe that by \eqref{1.11} and the uniform lower bound of $\sigma_{\bt_j,h_j,\bv_j}$ established in Lemma \ref{1.6},
\begin{align}\begin{split}\label{11}
 \big|\E[\Upsilon_j(S_1,\bTheta_1)]\big|
	\lesssim\log(h_j)^2\sqrt{h_j}.\end{split}
\end{align}
Because of this bound, the expectation $\E[\Upsilon_j(S_1,\bTheta_1)]$ in the definition of $X_{i,j}$ will only provide terms of negligible order if we check the conditions of Theorem \ref{1}. In particular,  condition \textit{(i)} is a direct consequence of the definition of $\sigma_{\bt_j,h_j,\bv_j}$ in \eqref{eq.sigma_def}. By Lemma \ref{lem.rad_symm_kernel_explicit} \textit{(ii)}, the uniform lower bound
of $\sigma_{\bt_j,h_j,\bv_j}$ in Lemma \ref{1.6}, the lower bound of $f_\bTheta$, and the boundedness of $f_{S,\bTheta}$, we find for $k=1,2,$ $\max_{j=1,\hdots,2p}\E[|\Upsilon_j(S_1,\bTheta_1)|^{2+k}] \lesssim h_{\min}^{-k/2}$. This implies condition \textit{(ii)} of Theorem \ref{1} with  $B_n\asymp h_{\tn{min}}^{-1/2}$. 

Lemma \ref{lem.rad_symm_kernel_explicit} \textit{(i)} implies $\max_{j=1,\hdots,2 p}|\mathrm{X}_{i,j}|\lesssim h_{\tn{min}}^{-1/2}$ which proves assertion \textit{(iii)} in the theorem for any $q>0 $ and $B_n=ch_{\tn{min}}^{-1/2}$, provided that the constant $c$ is
chosen sufficiently large. Consequently, Theorem \ref{1} applies and for $\mathbf{Z}\sim N(\mathbf 0,\E[\mathbf{X}_1\mathbf{X}_1^\top])$
\[
 \sup_{A\in\ca{A}_{2p}'}\Big|\P\Big( \frac{1}{\sqrt{n}}\sum_{i=1}^n\mathbf{X}_i\in A\Big)-\P\Big(\mathbf{Z}\in A\Big)\Big|
 \lesssim \Big(\frac{h_{\tn{min}}^{-1}\log^7(n)}{n}\Big)^\frac{1}{6}+\Big(\frac{h_{\tn{min}}^{-1}\log^3(n)}{n^{1-2/q}}\Big)^\frac{1}{3}
 \rightarrow 0
\]
choosing $q$ large enough and using Assumption \ref{1.20}. 

In a second step, we show that there exists a version of the Gaussian noise $W$ such that
\[
  \max_{j=1,\hdots,2p}\big|\mathrm{Z}_j-W(\Upsilon_j\sqrt{f_{S,\bTheta}})\big|
  =O_{\mathbb{P}}\big(|\log(h_{\tn{max}})|^{3}\sqrt{h_{\tn{max}}}\big).
\]
To this end, we define the Gaussian process $(\widetilde{W}(f))_{f\in L^\infty(\ca{Z})}$ indexed by $L^\infty(\ca{Z})$ as the centered Gaussian process with covariance function
\begin{align*}
 &\int_{\ca{Z}} f_1(s,\btheta)f_2(s,\btheta)f_{S,\bTheta}(s,\btheta)dsd\btheta- \int_{\ca{Z}}
 f_1(s,\btheta)f_{S,\bTheta}(s,\btheta)dsd\btheta \int_{\ca{Z}} f_2(s,\btheta)f_{S,\bTheta}(s,\btheta)dsd\btheta.
 \end{align*}
Thus, there exists a version of $\widetilde{W}(f)$ such that $\mathbf{Z}=\big(\widetilde{W}(\Upsilon_1),\hdots,\widetilde{W}(\Upsilon_{2p})\big)^\top.$ Recall that $(W(f))_{f\in L^2(\nu)}$ defines a Gaussian process whose mean and 
covariance functions are $0$ and $\int_{\mathbb{S}^{d-1}}\int_{\R}  f_1(s,\btheta)f_2(s,\btheta)dsd\btheta$, respectively. Basic calculations show that there exists a version of $W$ such that
\[
 \widetilde{W}(f)=W(f\sqrt{f_{S,\bTheta}})-\int_{\mathbb{S}^{d-1}}\int_{\R}f(s,\btheta)
 f_{S,\bTheta}(s,\btheta)dsd\btheta \;W(\sqrt{f_{S,\bTheta}}).
\]
Hence,
 \[
  \big|\widetilde{W}(\Upsilon_j)-W(\Upsilon_j\sqrt{f_{S,\bTheta}})\big|=\Big|\int_{\mathbb{S}^{d-1}}\int_{\R}
  \Upsilon_j(s,\btheta){f_{S,\bTheta}(s,\btheta)}dsd\btheta\; W(\sqrt{f_{S,\bTheta}})\Big|.
 \]
By \eqref{11}, $|\int_{\mathbb{S}^{d-1}}\int_{\R} \Upsilon_j(s,\btheta){f_{S,\bTheta}(s,\btheta)}dsd\btheta|\lesssim \log(h_j)^2\sqrt{h_j}.$ Furthermore, $W(\sqrt{f_{S,\bTheta}}) \sim N(0,1)$ which implies that $\E[  \max_{j=1,\hdots,2p} |\widetilde{W}(\Upsilon_j)-W(\Upsilon_j\sqrt{f_{S,\bTheta}})|]\lesssim \log(h_{\tn{max}})^{2}\sqrt{h_{\tn{max}}}.$ An application of Markov's inequality finally proves
 \[
   \max_{j=1,\hdots,2p}\big|\widetilde{W}(\Upsilon_j)-W(\Upsilon_j\sqrt{f_{S,\bTheta}})\big|=O_{\mathbb{P}}\big(|\log(h_{\tn{max}})|^{3}\sqrt{h_{\tn{max}}}\big).
 \]
The insertion of the bandwidth normalization terms has no influence on the convergence as
translation and multiplication preserve the interval structure.
\end{proof}

\subsection{Boundedness of the limit statistic}\label{1.14}
Recall from Lemma \ref{1.6} that $\sigma_{\bt,h,\bv}$ is uniformly bounded from below whenever $h$ is sufficiently small, where
the upper bound $\ol h$ for $h$ only depends on $\phi$, $d$ and $f_\bbeta.$ We therefore introduce the set 
\[
 \ol{\ca{T}}:=\big\{(\bt,h,\bv)\in\ca T : h\leq \ol h\big\}.
\]

\begin{thm2}
\label{thm2.bd_limit_stat}
Under the assumptions of Theorem \ref{thm.main}, $\sup_{(\bt,h,\bv)\in \ol{\mathcal{T}}} \beta_h (|X_{\bt,h,\bv}|/
\sigma_{\bt,h,\bv} - \alpha_h)$ is almost surely bounded.
\end{thm2}

\begin{proof}
We apply Theorem 6.1  in  \cite{duembgen2001} to the non-normalized process 
  \begin{align*}
    Y_{\bt,h,\bv }:=\frac{1}{\sigma_{\bt,h,\bv}}\int_{\mathbb{S}^{d-1}}\int_{\mathbb{R}}\langle
    \btheta, \bv\rangle (\mathcal{H}_d\widetilde \phi^{(d-1)})\Big(\frac{s-\langle \bt, \btheta \rangle}{h}\Big)
    \frac{\sqrt{f_{S,\bTheta}(s,{\btheta})}}{{f_{\bTheta}({\btheta})}} W(dsd\btheta).
  \end{align*}
Denote by $\rho$ the canonical pseudo-metric on $\ca{T}$, induced by $Y_{\bt,h,\bv }$
  \begin{align*}
    \rho:\begin{cases}
        \ol{\ca{T}}\times \ol{\ca{T}} \rightarrow \R_0^+\\
        \bigl((\bt,h,\bv),(\bt',h',\bv')\bigr)\mapsto \Bigl(\E\big|Y_{\bt,h,\bv }-Y_{\bt',h',\bv' }\big|^2\Bigr)^{\frac{1}{2}}.
      \end{cases}
  \end{align*}
In the next step, we prove   
  \begin{align}
    \rho\bigl(
       (\bt,h,\bv),(\bt',h',\bv')\bigr)\lesssim\big(\|\bv-\bv'\|^2+{\|\bt-\bt'\|}+{|h-h'|}\big)^{1/2}.
       \label{eq.rho_to_show}
  \end{align} 
By the uniform lower and upper bound for $\sigma_{\bt,h,\bv},$
    \begin{align}\begin{split}\label{u2}
 \big|Y_{\bt,h,\bv }-Y_{\bt',h',\bv' }\big|^2\lesssim&
 \big|\sigma_{\bt,h,\bv}Y_{\bt,h,\bv }-\sigma_{\bt',h',\bv'}Y_{\bt',h',\bv' }\big|^2+
  \big|Y_{\bt',h',\bv' }\big|^2\big| \sigma_{\bt',h',\bv'} - \sigma_{\bt,h,\bv}\big|^2.
  \end{split}\end{align}
In order to bound the expectation of the first term on the right hand side of \eqref{u2}, we use the boundedness properties of $f_\bTheta$ and $f_{S,\bTheta}$ 
  \begin{align*}
    &\E\big|\sigma_{\bt,h,\bv}Y_{\bt,h,\bv }-\sigma_{\bt',h',\bv'}Y_{\bt',h',\bv' }\big|^2\\
   \lesssim  &\int_{\mathbb{S}^{d-1}}\int_{\R} \biggl|\langle\btheta,\bv-\bv'\rangle(\mathcal{H}_d\widetilde \phi^{(d-1)})\left(\frac{s - \langle \bt, {\btheta} \rangle}{h}\right)\biggr|^2\,ds\,d\btheta   \\
    &+\int_{\mathbb{S}^{d-1}}\int_{\R} \biggl|(\mathcal{H}_d\widetilde \phi^{(d-1)})\left(\frac{s - \langle \bt, {\btheta} \rangle}{h}\right)-(\mathcal{H}_d\widetilde \phi^{(d-1)})\left(\frac{s - \langle \bt', {\btheta} \rangle}{h}\right)\biggr|^2\,ds\,d\btheta \\
    &+\int_{\mathbb{S}^{d-1}}\int_{\R} \biggl|(\mathcal{H}_d\widetilde \phi^{(d-1)})\left(\frac{s - \langle \bt', {\btheta} \rangle}{h}\right)-(\mathcal{H}_d\widetilde \phi^{(d-1)})\left(\frac{s - \langle \bt', {\btheta} \rangle}{h'}\right)\biggr|^2\,ds\,d\btheta \\
    %
        %
    =:&\;\rho_1+\rho_2+\rho_3.
  \end{align*}
We show that the three terms can be bounded by the squared  r.h.s. in \eqref{eq.rho_to_show}. From Lemma \ref{lem.rad_symm_kernel_explicit} \textit{(ii)} we obtain $\rho_1\lesssim h \|\bv-\bv'\|^2.$ For $\rho_2,$ we distinguish between the cases $\|\bt-\bt'\|> h$ and  $\|\bt-\bt'\|\leq h$. In the first case, the triangle inequality and Lemma \ref{lem.rad_symm_kernel_explicit} \textit{(ii)} give $\rho_2\lesssim h<\|\bt-\bt'\|.$ In the second case, the integral w.r.t. the variable $s$ in $\rho_2$ is equal to
\begin{align*}
2h\int_{\mathbb{R}}\big((\mathcal{H}_d\widetilde \phi^{(d-1)})(s)\big)^2ds-2h
\int_{\mathbb{R}}(\mathcal{H}_d\widetilde \phi^{(d-1)})(s)(\mathcal{H}_d\widetilde \phi^{(d-1)})
\Big(s+\frac{\langle \bt,\btheta\rangle-\langle\bt',\btheta\rangle}{h'}\Big)ds.
\end{align*}
Recall that the Hilbert transform and the differentiation operator commute. Therefore, using the differentiability of 
$\widetilde \phi^{(d-1)}$ which has been shown in the proof of Lemma \ref{lem.rad_symm_kernel_explicit}, we find that
 \[
h(\mathcal{H}_d\widetilde \phi^{(d-1)})
\Big(s+\frac{\langle \bt,\btheta\rangle-\langle\bt',\btheta\rangle}{h}\Big)=
h(\mathcal{H}_d\widetilde \phi^{(d-1)})(s)+\langle \bt-\bt',\btheta\rangle(\mathcal{H}_d\widetilde \phi^{(d)})(\xi)
 \]
for some $\xi$ between $s$ and $ s+\tfrac{\langle \bt-\bt',\btheta\rangle}{h}$. Hence, 
\bal
\rho_2\lesssim\|\bt-\bt'\|\int_{\mathbb{R}}(\mathcal{H}_d\widetilde \phi^{(d-1)})(s)(\mathcal{H}_d\widetilde \phi^{(d)})(\xi)ds
\lesssim\|\bt-\bt'\|\int_{\mathbb{R}}\big(1+(s/2)^2\big)^{-1}ds\lesssim\|\bt-\bt'\|.
\end{align*}
Here, we used the boundedness of $\mathcal{H}_d\widetilde \phi^{(d-1)}$ shown in Lemma 
\ref{lem.rad_symm_kernel_explicit} \textit{(i)}, $|\xi|\geq |s|/2$ for all $|s|\geq2$ in the case $\|\bt-\bt'\|\leq h$, and
$ |(\mathcal{H}_d\widetilde{\phi}^{(d)})(u)|\lesssim(1+u^2)^{-1}$. The latter is obvious for $d$ odd. For $d$ even we find that 
 $\widetilde{\phi}$ is an odd function and therefore $\widetilde{\phi}^{(d)}$
is an odd function. Moreover, for any odd function $f$ such that $\ca{H}_df$ exists, we have, up to some constant,
\begin{align*}
	(\mathcal{H}_df)(u) =\int_{-\infty}^0\frac{f(x)}{u-x}dx+ \int_0^\infty \frac{f(x)}{u-x} dx=
	\int_0^\infty {f(x)}\Big(\frac{-1}{u+x}+\frac{1}{u-x}\Big)dx=\int_0^\infty\frac{2xf(x)}{u^2-x^2}dx.
\end{align*}
Here all integrals are understood in the principal value sense. Finally, a similarly argument as in the proof of Lemma \ref{lem.rad_symm_kernel_explicit} shows that $\ca{H}_d\widetilde\phi^{(d)}$ exists and that 
 $ |(\mathcal{H}_d\widetilde{\phi}^{(d)})(u)|\lesssim(1+u^2)^{-1}$ by the compact support of $\widetilde\phi$.

We finally turn to $\rho_3$. Without loss of generality, we may assume $h\leq h'$. We study the cases $h\leq h'/2$ and 
$h>h'/2$, separately. In the first case, the triangle inequality and Lemma \ref{lem.rad_symm_kernel_explicit} \textit{(ii)} give $\rho_3\lesssim h+h'\lesssim |h'-h|.$ If $h'/2<h\leq h',$ we argue as for the upper bound of $\rho_2$ and find
\bal
\rho_3\lesssim(h'-h)\int_{\mathbb{R}}\big((\mathcal{H}_d\widetilde \phi^{(d-1)})(s)\big)^2ds-2h\big(-1+\tfrac{h}{h'}\big)
\int_{\mathbb{R}}(\mathcal{H}_d\widetilde \phi^{(d-1)})(s)s(\mathcal{H}_d\widetilde \phi^{(d)})(\xi)ds
\end{align*}
for some $\xi$ between $s$ and $\tfrac{h}{h'} s$.
Recall that $ |(\mathcal{H}_d\widetilde{\phi}^{(d)})(u)|\lesssim(1+u^2)^{-1}$ and 
$ |(\mathcal{H}_d\widetilde{\phi}^{(d-1)})(u)|\lesssim(1+u^2)^{-1/2}$. Thus,
\begin{align*}
 \int_{\mathbb{R}}\big|(\mathcal{H}_d\widetilde \phi^{(d-1)})(s)s(\mathcal{H}_d\widetilde \phi^{(d)})(\xi)\big|ds
 \lesssim \int_{\mathbb{R}}\big(1+s^2\big)^{-1/2}|s|\big(1+(s/2)^2\big)^{-1}ds<\infty,
\end{align*}
where we used that $|\xi|\geq \frac{ h}{h'}|s|>|s|/2$. Finally,  $|h^2/h'-h|\leq h'-h$ implies $\rho_3\lesssim | h'-h |.$

For the second term on the right hand side of \eqref{u2} we use that $\Var(Y_{\bt,h,\bv})=h$. Using again the uniform boundedness from above and below of $\sigma_{\bt,h,\bv}$ and the fact that $|\sqrt{x}-\sqrt{y}|^2 \leq |x-y|$ for all $x,y \geq 0$ gives
  \begin{align*}
  &\E\big|Y_{\bt',h',\bv' }\big|^2 |\sigma_{\bt,h,\bv}-\sigma_{\bt',h',\bv'}|^2\\
  =&h'\Big|\frac{(\E|\sigma_{\bt,h,\bv}Y_{\bt,h,\bv}|^2)^{1/2}}{\sqrt{h}}-\frac{(\E|\sigma_{\bt',h',\bv'}Y_{\bt',h',\bv'}|^2)^{1/2}}{\sqrt{h'}} \Big|^2\\
  \lesssim& \E\big|\sigma_{\bt,h,\bv}Y_{\bt,h,\bv}\big|^2\Big|\frac{\sqrt{h'}}{\sqrt{h}}-1\Big|^2+
  \big|(\E|\sigma_{\bt,h,\bv}Y_{\bt,h,\bv}|^2)^{1/2}-(\E|\sigma_{\bt',h',\bv'}Y_{\bt',h',\bv'}|^2)^{1/2}\big|^2\\
  \lesssim& |h-h'| + \E\big|\sigma_{\bt,h,\bv}Y_{\bt,h,\bv }-\sigma_{\bt',h',\bv'}Y_{\bt',h',\bv' }\big|^2.
\end{align*}
For the second term in the last line, the bounds above apply which completes the proof for \eqref{eq.rho_to_show}.

Set $\sigma^2(\bt,h,\bv):=h$ and $\widetilde{\rho}\bigl((\bt,h,\bv),(\bt',h',\bv')\bigr):=\big(\|\bv-\bv'\|^2+{\|\bt-\bt'\|}+{|h-h'|}\big)^{1/2},$ such that
  \begin{align*}
    \sigma^2(\bt,h,\bv)-\sigma^2(\bt',h',\bv')\leq \widetilde{\rho}^2\bigl((\bt,h,\bv),(\bt',h',\bv')\bigr)
    \ \ \text{for all} \  ((\bt,h,\bv),(\bt',h',\bv')\bigr)\in\ol{\ca{T}}\times\ol{\ca{T}}.
  \end{align*}  
For fixed $\bigl((\bt,h,\bv),(\bt',h',\bv')\bigr)\in \ol{\ca{T}}\times\ol{\ca{T}}$, the random variable $ Y_{\bt,h,\bv}-Y_{\bt',h',\bv'}$ follows a normal distribution
   with mean zero and variance bounded by a constant multiple of $\widetilde{\rho}^2\bigl(
     (\bt,h,\bv),(\bt',h',\bv')\bigr)$. 
Thus, there exists a constant $M>0$ such that for any $\eta> 0$,
\begin{align*}
   \mathbb{P}\bigl(|Y_{\bt,h,\bv}-Y_{\bt',h',\bv'}|\geq\widetilde{\rho}((\bt,h,\bv),(\bt',h',\bv'))\eta\bigr)
\lesssim\exp(-\eta^2/M).
\end{align*}   
Furthermore, $\mathbb{P}\bigl(Y_{\bt,h,\bv}>\sqrt{h}\eta\bigr)\lesssim\exp(-\eta^2/2),$ as $h^{-1/2}Y_{\bt,h,\bv}$ corresponds to a standard normal distributed random variable. Thus, conditions (i) and (ii) of 
Theorem 6.1  in \cite{duembgen2001} are satisfied. As in
\cite{Eckle} one shows that condition (iii) of  Theorem 6.1  in  \cite{duembgen2001} holds with $V=(3d-1)/2$ and that
 the process $Y_{\bt,h,\bv}$ is almost surely continuous on $\ol{\ca{T}}$ with respect to ${\rho}$.
The boundedness of 
$	\sup_{(\bt,h,\bv)\in \ol{\mathcal{T}}}\bigl(\beta_h \frac{|X_{\bt,h,\bv}|}{\sigma_{\bt,h,\bv}} - \alpha_h\beta_h\bigr)$ follows by an application of Theorem 6.1 and Remark 1 in  \cite{duembgen2001}.
\end{proof}

\subsection{Replacing the true densities in the limit process by estimators}
\begin{thm2}
Under the assumptions of Theorem \ref{thm.main},
 \begin{align*}
  \sup_{(\bt,h,\bv)\in \ca{T}_n}\beta_h \frac{\big||X_{\bt,h,\bv}|-|\widehat X_{\bt,h,\bv}|\big|}{\sigma_{\bt,h,\bv}}
  =o_{\mathbb{P}}(1)\quad\text{for }n\rightarrow\infty.
\end{align*}
\end{thm2}
\begin{proof}
 Recall the definition of the symmetrized set $ \ca{T}_n'$ in \eqref{1.13} and let  
\begin{align*}
    \widehat F(s,\btheta):= \frac{\sqrt{{\widetilde f_{S,\bTheta}(s,\btheta)}}}{\widetilde{f}_\bTheta(\btheta)}
    -\frac{\sqrt{ f_{S,\bTheta}(s,\btheta)}}{{f}_\bTheta(\btheta)}.
 \end{align*}
Lemma \ref{1.8} and an argument as in the proof of Lemma \ref{1.9}  show that $\sup_{(s,\btheta)\in\R\times\sd}\big|\widehat F(s,\btheta)\big|=O(\log(n)^{-1})$ for $n\rightarrow \infty$ almost surely. Define
  \begin{align*}
    \Delta_{\bt,h,\bv}:=X_{\bt,h,\bv}-\widehat X_{\bt,h,\bv}=\frac{1}{\sqrt{h}}\int_{\mathbb{S}^{d-1}}
    \int_{\mathbb{R}}\langle\btheta,\bv\rangle(\ca{H}_d\widetilde \phi^{(d-1)} )
    \biggl(\frac{s-\langle\bt,\btheta\rangle}{h}\biggr) \widehat F(s,\btheta)\,dW_{s,\btheta}
  \end{align*}
and $ \Delta_{\infty,\bt,h,\bv}:=\log(n)^{-1}X_{\bt,h,\bv}.$ We write $\widetilde {\mathbb{P}}$ and $\widetilde{\mathbb{E}}$ for the probability and expectation conditionally on
$(S_i,\bTheta_i)$, $i=n+1,\hdots,2n$. Under $\widetilde {\mathbb{P}}$, the vectors
$(\Delta_{\bt,h,\bv})_{(\bt,h,\bv)\in\ca{T}_n'}$ and  $(\Delta_{\infty,\bt,h,\bv})_{(\bt,h,\bv)\in\ca{T}_n'}$
are centered and normally distributed with
\begin{align*}
  \widetilde{\mathbb{E}}\bigl(\Delta_{\bt,h,\bv}-\Delta_{\bt',h',\bv'}\bigr)^2+ \widetilde{\mathbb{E}}\bigl(\Delta_{\infty,\bt,h,\bv}
  -\Delta_{\infty,\bt',h',\bv'}\bigr)^2\lesssim\log(n)^{-2}\quad \forall (\bt,h,\bv), (\bt',h',\bv')\in\ca{T}_n'
\end{align*}
almost surely. 
Hence, an application of Theorem 2.2.5 in \cite{Adler2007} gives
\begin{align*}
  \Big|\widetilde{\mathbb{E}}\Bigl(\sup_{(\bt,h,\bv)\in\ca{T}_n'}\Delta_{\bt,h,\bv}\Bigr)-
  \widetilde{\mathbb{E}}\Bigl(\sup_{(\bt,h,\bv)\in\ca{T}_n'}
  \Delta_{\infty,\bt,h,\bv}\Bigr)\Big| =O\big(\log(n)^{-1/2}\big) \quad \mbox{almost surely.}
\end{align*}
 Moreover, by the almost sure asymptotic boundedness of $\sup_{(\bt,h,\bv)\in {\mathcal{T}_n'}} \beta_h (|X_{\bt,h,\bv}|/\sigma_{\bt,h,\bv} - \alpha_h)$ proved in Theorem \ref{thm2.bd_limit_stat}, we have $|\widetilde{\mathbb{E}}(\sup_{(\bt,h,\bv)\in\ca{T}_n'}
  \Delta_{\infty,\bt,h,\bv})| =O(\log(n)^{-1/2})$ almost surely. Finally, for some constant $C>0$,
\begin{align*}
  \widetilde{\mathbb{P}}\Bigl(\sup_{(\bt,h,\bv)\in\ca{T}_n}\beta_h\frac{|\Delta_{\bt,h,\bv}|}{\sigma_{\bt,h,\bv}}>
  \log\log(n)^{-1/2}\Bigr)
  &\leq
  \widetilde{\mathbb{P}}\Bigl(\sup_{(\bt,h,\bv)\in\ca{T}_n'}\Delta_{\bt,h,\bv}>C\log\log(n)^{1/2}{\log(n)}^{-1/2}\Bigr)
  \\&=O\bigl(\log\log(n)^{-1/2}\bigr)
\end{align*}
for $n\rightarrow\infty$ almost surely, by Markov's inequality. The constants introduced above do not depend on the 
second sample $(S_i,\bTheta_i)$, $i=n+1,\hdots,2n$ and therefore the claim follows by an application of the law of iterated expectations.
\end{proof}

\subsection{Replacement of the standard deviation by an estimator}\label{appc5}
\begin{thm2}
Under the assumptions of Theorem \ref{thm.main},
 
 \begin{tabular}{ll}
(i) & \  $\sup_{(\bt,h,\bv)\in \ca{T}_n} \beta_h\sqrt{n} |\widehat T_{\bt,h,\bv}-E[T_{\bt,h,\bv}]|
  \Big|\frac{1}{\widehat{\sigma}_{\bt,h,\bv}} - \frac{1}{{\sigma}_{\bt,h,\bv}} \Big|
  =o_\P(1)\quad\text{for }n\rightarrow\infty
  \text{;}$ \\[0.3cm] 
(ii) & \  
 	$\sup_{(\bt,h,\bv)\in \ca{T}_n}\beta_h |\widehat X_{\bt,h,\bv}|
 \Big|\frac{1}{\widehat{\sigma}_{\bt,h,\bv}} - \frac{1}{{\sigma}_{\bt,h,\bv}} \Big|=
o_\P(1)\quad\text{for }n\rightarrow\infty
  \text{.}$\\
\end{tabular}
\end{thm2}
\begin{proof}
We only prove \textit{(i)} as \textit{(ii)} follows by a similar argument. By Lemma \ref{1.6} and Lemma \ref{1.9}, $\sigma_{\bt,h,\bv}$ and $\widehat \sigma_{\bt,h,\bv}$ are almost surely uniformly bounded from below for all sufficiently large $n$. Thus,
\begin{align*}
 &\sup_{(\bt,h,\bv)\in \ca{T}_n} \beta_h\sqrt{n} |\widehat T_{\bt,h,\bv}-E[T_{\bt,h,\bv}]|
 \Big|\frac{1}{\widehat{\sigma}_{\bt,h,\bv}} - \frac{1}{{\sigma}_{\bt,h,\bv}} \Big|\\\lesssim&
  \sup_{(\bt,h,\bv)\in \ca{T}_n }\beta_h\biggl(\sqrt{n}\frac{|\widehat T_{\bt,h,\bv}-\mathbb{E}[T_{\bt,h,\bv}]|}
 {\sigma_{\bt,h,\bv}}-\alpha_h\biggr)\sup_{(\bt,h,\bv)\in \ca{T}_n }|\sigma_{\bt,h,\bv}-\widehat{\sigma}_{\bt,h,\bv}|
 \\&+\frac{\log(n)}{\log\log(n)}\sup_{(\bt,h,\bv)\in \ca{T}_n }|\sigma_{\bt,h,\bv}-\widehat{\sigma}_{\bt,h,\bv}|
\end{align*}
almost surely. The claim follows from Lemma \ref{1.9}, Theorems \ref{1.16} and \ref{1.15} and the almost sure boun\-dedness of $\sup_{(\bt,h,\bv)\in {\mathcal{T}_n}} \beta_h (|X_{\bt,h,\bv}|/\sigma_{\bt,h,\bv} - \alpha_h)$ established in Theorem \ref{thm2.bd_limit_stat}.
\end{proof}

\section{Proofs of Theorems \ref{t10} and \ref{.1}}

\begin{proof}[{Proof of Theorem \ref{t10}}]
We have
\begin{align*}\P\Big(\exists  (\bt,h,\bv)\in\ca T_n:|\widehat T_{\bt,h,\bv}|>\kappa^{\bt,h,\bv}_n(\alpha)\Big)&=
 1-\mathbb{P}\Big(\sup_{(\bt,h,\bv)\in \ca{T}_n} \beta_h\Bigl(\sqrt{n} \frac{|\widehat T_{\bt,h,\bv}|}{\widehat{\sigma}_{\bt,h,\bv}}
- \alpha_h\Bigr)\leq\kappa_n(\alpha)\Big)
\\&=1-\mathbb{P}\Big(\sup_{(\bt,h,\bv)\in \ca{T}_n}\beta_h\Bigl( \frac{|\widehat X_{\bt,h,\bv}|}{\widehat{\sigma}_{\bt,h,\bv}} - \alpha_h\Bigr)\leq \kappa_n(\alpha)\Big)+o(1)
 \\&\leq\alpha+o(1)\end{align*}
for $n\rightarrow\infty$. Here we used \eqref{t11j} for the first equality and
 Theorem \ref{thm.main} for the second.
\end{proof}

\begin{proof}[{Proof of Theorem \ref{.1}}]
We assume in the following that $c_d>0$. The case $c_d<0$ can be treated similarly. The following statement can be derived similarly as in the proof of Theorem 3.3 in \cite{Eckle}. For a null sequence $0<(\alpha_n)_{n\in\mathbb{N}}< 1$ converging sufficiently slowly and for the set $\ca T_n'\subseteq\ca T_n$ of all triples for which the inequality
\beq{t20}
 \int_{\mathbb{R}^d} \phi_{\bt,h}(\bb)\partial_{\bv} f_{\bbeta} (\bb)\,d\bb<-2c_d^{-1}h^{-d-1/2}\kappa_n^{\bt,h,\bv}(\alpha_n)
\eeq
is satisfied it holds that 
 \[
\P\Big( \widehat T_{\bt,h,\bv}>\kappa^{\bt,h,\bv}_n(\alpha_n)\tn{ for all }(\bt,h,\bv)\in\ca T_n'\Big)=1-o(1).
 \]
Hence, the hypotheses \eqref{t5neu2} are rejected simultaneously on the set of scales $\ca T_n'$ with asymptotic probability one. Moreover, for a mode $\bb_0$  in $(\mathbf a_1,\mathbf a_2)$ of $f_\bbeta$ and any triple  $(\bt,h,\bv)\in\ca{T}_n^{\bb_0}$, one can 
prove that $\partial_{\bv}f_\bbeta(\bb)\lesssim
-h$ for all $\bb\in \tn{supp}\phi_{\bt,h}$ by following the arguments in the proof of Theorem 10 in \cite{Eckle2017}. Consequently, 
$\int_{\mathbb{R}^d} \phi_{\bt,h}(\bb)\partial_{\bv}f_\bbeta(\bb)\tn{d}\bb\lesssim -h.$

As seen in Appendices \ref{1.14}-\ref{appc5} 
\[\bigg|\sup_{(\bt,h,\bv)\in\ca T_n}\beta_h\bigg( \frac{|\widehat X_{\bt,h,\bv}|}{\widehat{\sigma}_{\bt,h,\bv}} - \alpha_h\bigg)
-\sup_{(\bt,h,\bv)\in\ca T_n}\beta_h\bigg( \frac{| X_{\bt,h,\bv}|}{{\sigma}_{\bt,h,\bv}} - \alpha_h\bigg)\bigg|=o_\P(1)\] 
for $n\rightarrow\infty$ and $\sup_{(\bt,h,\bv)\in\ca T_n}\beta_h( \frac{| X_{\bt,h,\bv}|}{{\sigma}_{\bt,h,\bv}} - \alpha_h)$ is
finite almost surely  for $n$ sufficiently large. Moreover,
$\widehat{\sigma}_{\bt,h,\bv}$ is almost surely uniformly bounded by Lemmas \ref{1.6} and \ref{1.9} for $n$ sufficiently large, such that
\[
h^{-d-1/2}\kappa_n^{\bt,h,\bv}(\alpha_n)\lesssim \sqrt{\frac{\log n}{n}} h^{-d-1/2}
\]
almost surely. In order to verify \eqref{t20}, we need to pick $h$ such that $h^{d+3/2}\gtrsim (\log (n)/n)^{1/2}.$ Thus, \eqref{t20} holds for $h\geq C\log(n)^\frac{1}{2d+3}n^{-\frac{1}{2d+3}}$ with some sufficiently large constant $C>0$.
\end{proof}


\end{document}